\documentclass[11pt]{article}

\usepackage{amsmath} 
\usepackage{amsthm} 
\usepackage{amssymb}	
\usepackage{graphicx} 
\usepackage{multicol} 
\usepackage{color}
\usepackage[dvips,letterpaper,margin=1in,bottom=1in]{geometry}

\usepackage[utf8]{inputenc}
\usepackage[english]{babel}

\newtheorem{theorem}{Theorem}[section]

\newtheorem{lemma}[theorem]{Lemma}

\newtheorem{proposition}[theorem]{Proposition}
\newtheorem{definition}{Definition}[section]
\newtheorem{example}{Example}
\newtheorem{remark}{Remark}[section]

\newcommand{\abs}[1]{\left| #1 \right|} 
\newcommand{\Abs}[1]{\lVert #1 \rVert} 
\newcommand{\ket}[1]{\left| #1 \right>} 
\newcommand{\bra}[1]{\left< #1 \right|} 

\newcommand{\ceil}[1] {\left\lceil #1 \right\rceil}

\usepackage{latexsym}
\usepackage{CJK}

\usepackage{qcircuit}

\usepackage{enumerate}
\usepackage{ulem}

\usepackage{algorithm}
\usepackage{algpseudocode}

\usepackage{stmaryrd}
\usepackage{booktabs}
\newcommand{\tabincell}[2]{\begin{tabular}{@{}#1@{}}#2\end{tabular}}

\usepackage{hyperref}
\newcommand{\footremember}[2]{%
    \footnote{#2}
    \newcounter{#1}
    \setcounter{#1}{\value{footnote}}%
}

\begin{document}

    \title{Quantum Random Access Stored-Program Machines}
        \author{
            Qisheng Wang \footremember{alley}{Qisheng Wang is with the Department of Computer Science and Technology, Tsinghua University, China (e-mail: \url{QishengWang1994@gmail.com}).}
            \and Mingsheng Ying \footremember{trailer}{Mingsheng Ying is with the State Key Laboratory of Computer Science, Institute of Software, Chinese Academy of Sciences, China, and also with the Department of Computer Science and Technology, Tsinghua University, China (e-mail: \url{yingms@ios.ac.cn}).}
        }
        \date{}
        \maketitle

    \begin{abstract}
    Random access machines (RAMs) and random access stored-program machines (RASPs) are models of computing that are closer to the architecture of real-world computers than Turing machines (TMs). They are also convenient in complexity analysis of algorithms. The relationships between RAMs, RASPs and TMs are well-studied. However, clear relationships between their quantum counterparts are still missing in the literature.
    We fill in this gap by formally defining the models of quantum random access machines (QRAMs) and quantum random access stored-program machines (QRASPs) and clarifying the relationships between QRAMs, QRASPs and quantum Turing machines (QTMs).
    In particular, we show that $\textbf{P} \subseteq \textbf{EQRAMP} \subseteq \textbf{EQP} \subseteq \textbf{BQP} = \textbf{BQRAMP}$, where \textbf{EQRAMP} and \textbf{BQRAMP} stand for the sets of problems that can be solved by polynomial-time QRAMs with certainty and bounded-error, respectively. At the heart of our proof, is a standardisation of QTM with an extended halting scheme, which is of independent interest.
    \end{abstract}


    \newpage

    \tableofcontents
    \newpage
    
    \section{Introduction}

\quad \quad \textbf{Models of Quantum Computing}:  Various traditional models of computing have been generalised to the quantum setting as the models of quantum computing (cf. \cite{Nie02}), including quantum Turing machines (QTMs) \cite{Deu85} and quantum circuits \cite{Deu89}. Several novel quantum computing models that have no classical counter-parts have also been proposed, e.g. measurement-based and one-way quantum computing \cite{Rau01, Rau03},  adiabatic quantum computing \cite{Far00}. Furthermore, the relationships between these models have been thoroughly studied \cite{Yao93, Ber97, Aha07}.

 {\vskip 3pt}

\textbf{Quantum Random Access Machines}: Random access machines (RAMs)\footnote{In the quantum computation and information literature, QRAM is also used to refer Quantum Random Access Memory
proposed by Giovannetti, Lloyd, and Maccone (\textit{Phys. Rev. Lett.} 100(2018), 160501).} and random access stored-program machines (RASPs) are another model of computing that is closer to the architecture of real-world computers than Turing machines (TMs).
They are also convenient in complexity analysis of algorithms \cite{Coo73, Aho74}.

The notion of quantum random access machine (QRAM) was first introduced in \cite{Kni96}, as a basis of the studies of quantum programming. Essentially, it is a RAM in the traditional sense with the ability to perform a set of quantum operations on quantum registers, including: (1) state preparation, (2) certain unitary operations, and (3) quantum measurements.

Recently, several quantum computer architectures have been proposed based on the QRAM model with some practical quantum instruction sets, including IBM OpenQSAM \cite{IBM18}, Rigetti's guil \cite{Rob16Quil} and Delft's eQASM \cite{Fu19}.

{\vskip 3pt}

\textbf{Contributions of This Paper}: However, a clear relationship between QRAMs and QTMs is still missing in the literature. The aim of this paper is to fill in this gap.
In \cite{Kni96}, QRAMs were described only in an informal way, and to our best knowledge, QRASPs have not been introduced in the existing literature. For our purpose, we first formally define the models of QRAMs and QRASPs as appropriate generalisation of RAMs and RASPs \cite{Aho74}.

It is worth mentioning that the formal model of QRASPs also provides us with a theoretical foundation of quantum programming (see \cite{Mis11} and Section 8.1 of \cite{Ying16} for a discussion of the significance of such a foundation).
In particular, currently a quantum computer is essentially a quantum processor controlled by a classical computers, and quantum programs are stored in the classical computer (as classical data). In this sense, our QRASPs are an appropriate model of the existing quantum computers.

Based on the formal models of QRAMs and QRASPs, we then clarify the relationships between QTMs, QRAMs and QRASPs. Our main results are:
\begin{enumerate}
  \item A $T(n)$-time QRAM (resp. QRASP) can be simulated by an $O(T(n))$-time QRASP (resp. QRAM).
  \item A $T(n)$-time QRAM under the logarithmic (resp. constant) cost criterion can be simulated by an $\tilde O(T(n)^4)$-time (resp. $\tilde O(T(n)^8)$-time) QTM.
  \item A $T(n)$-time QTM can be simulated within error $\varepsilon > 0$ by an $O(T(n)^2 \operatorname{polylog}(T(n), 1/\varepsilon))$-time QRAM (under both the logarithmic and constant cost criterions).
\end{enumerate}

In comparison with the classical counterparts \cite{Coo73}, $T(n)$-time RAMs under the logarithmic (resp. constant) criterion can be simulated by $O(T(n)^2)$-time (resp. $O(T(n)^3)$-time) TMs. Conversely, $T(n)$-time TMs can be simulated by $\tilde O(T(n))$-time RAMs.

The above results have some immediate corollaries on computational complexity. We define two complexity classes:\begin{itemize}\item \textbf{EQRAMP} stands for exact quantum random access machine polynomial-time, and \item \textbf{BQRAMP} stands for bounded-error quantum random access machine polynomial-time.\end{itemize}
Then it holds that
    \begin{equation}\label{inclusions}
        \textbf{P} \subseteq \textbf{EQRAMP} \subseteq \textbf{EQP} \subseteq \textbf{BQP} = \textbf{BQRAMP}.
    \end{equation}
Not much has been known in the literature about the relationship between \textbf{EQP} and other complexity classes. Here, an inclusion between \textbf{EQP} and \textbf{EQRAMP} is established. However, we still do not know which inclusion in (\ref{inclusions}) is proper.

{\vskip 3pt}

\textbf{Major Challenge}:
    The main difficulty in comparing the computational power of QTMs, QRAMs and QRASPs comes from the difference between their halting schemes:
    \begin{itemize}\item There have been a bunch of discussions about the halting scheme of QTMs \cite{Mye97, Lin98, Oza98, Shi02, Miy05}, where QTMs are required to \textit{terminate exactly at a fixed time} (depending on the input) with certainty.
    \item On the other hand, it is reasonable to allow QRAMs and QRASPs to terminate at any time with an appropriate probability. \end{itemize}

    We resolve this issue by extending the model of QTMs \cite{Ber97} so that QTMs are allowed to terminate at any time with a non-zero probability. This extension of QTMs makes that QTMs match the halting scheme of QRAMs and QRASPs, and thus it is convenient to simulate QRAMs and QRASPs by QTMs. In order to deal with the extended halting scheme of QTMs, on the other hand, we introduce the notion of \textit{standard} QTM as a stepping stone and give a constructive proof that every well-formed QTM within time $T(n)$ can be simulated by a \textit{standard} QTM within time $O(T(n)\log^2 T(n))$, called a standardisation of QTM. Based on the notion of standard QTM, we are also able to give an alternative definition of complexity classes \textbf{EQP} and \textbf{BQP} in terms of our extended QTMs.

{\vskip 3pt}

\textbf{Organisation of the Paper}: In Section \ref{sec:qtm}, we recall from \cite{Ber97} the definition of QTMs and then introduce the notions of extended QTMs and standard QTMs. In Section \ref{sec:qram-qrasp}, the formal definitions of QRAMs and QRASPs are given. Our main results  are presented in Section \ref{sec:main-results}.

The remaining sections are then devoted to provide all of the details. The computations of QRAMs and QRASPs are carefully described in Sections \ref{sec:def-qram} and \ref{sec:def-qrasp}, respectively. The simulations of QRAMs and QRASPs with each other are described in Section \ref{sec:sim-qram-qrasp}, and the simulations of QRAMs and QTMs with each other are given in Section \ref{sec:sim-qram-qtm}. The standardisation of QTM is given in Section \ref{sec:standard-qtm}.
\section{QTMs} \label{sec:qtm}

   The purpose of this section is two-fold. For convenience of the reader, in the first two subsections, we review some basic notions of quantum Turing machines (QTMs). Our exposition is mainly based on \cite{Ber97}. In the last subsection, we define the notion of standard QTM and show that every well-formed QTM can be efficiently simulated by a standard QTM. This result will serve as a step stone in comparing the computational power of QTMs with that of QRAMs and QRASPs.

    \subsection{Single-Track QTMs}

    Let $T: \mathbb{N} \to \mathbb{N}$ be a mapping from natural numbers to themselves. We write $\mathbb{C}(T(n))$ for the set of all $T(n)$-time computable complex numbers, i.e. for every $x \in \mathbb{C}(T(n))$, there is a $T(n)$-time deterministic Turing machine $M$ such that
    $
        \abs{M(1^n) - x} < 2^{-n},
    $
    where $M(1^n)$ denotes the output floating point complex number of $M$ on input $1^n$.
    Let $\tilde{\mathbb{C}}$ be the set of all polynomial-time computable complex numbers, i.e.
    $
        \tilde{\mathbb{C}} = \bigcup_{k=1}^\infty \mathbb{C}(n^k).
    $

\begin{definition}    A Quantum Turing Machine (QTM) is a $5$-tuple $M = (Q, \Sigma, \delta, q_0, q_f)$,
    where:
    \begin{enumerate}
      \item $Q$ is a finite set of states;
      \item $\Sigma$ is a finite alphabet with blank symbol $\#$;
      \item $\delta: Q \times \Sigma \times \Sigma \times Q \times \{L, R\} \to \tilde{\mathbb{C}}$ is the transition function;
      \item $q_0 \in Q$ is the initial state; and
      \item $q_f \in Q$ is the final state ($q_0 \neq q_f$).
    \end{enumerate}\end{definition}

A configuration of the tape is described by a function $\mathcal{T}: \mathbb{Z} \to \Sigma$ such that $\mathcal{T}(m) = \#$ except for finitely many integers $m$. Thus, the symbol at position $m$ on the tape is denoted $\mathcal{T}(m)$. We write $\Sigma^\# \subseteq \Sigma^\mathbb{Z}$ for the set of all possible tape configurations. Moreover, a (computational) configuration of $M$ is a $3$-tuple $c = (q, \mathcal{T}, \xi) \in Q \times \Sigma^\# \times \mathbb{Z}$, where $q$ is the current state, $\mathcal{T}$ is the tape configuration and $\xi$ is the head position. It represents a basis state $\ket c = \ket{q}_{Q} \ket{\mathcal{T}}_{\Sigma^\#} \ket{\xi}_\mathbb{Z} = \ket {q, \mathcal{T}, \xi}$ of the quantum machine. Therefore, the state Hilbert space of $M$ is
$\operatorname{span}\{\ket c\}$, where $c$ ranges over all configurations of $M$.
The time evolution operator $U$ of $M$ is defined by $\delta$ as follows:
    \[
        U \ket{p, \mathcal{T}, \xi} = \sum_{\sigma, q, d} \delta(p, \mathcal{T}(\xi), \sigma, q, d) \ket{q, \mathcal{T}_{\xi}^{\sigma}, \xi+d},
    \]
    where
    \[
    \mathcal{T}_{\xi}^{\sigma}(m) = \begin{cases}
        \mathcal{T}(m) & m \neq \xi \\
        \sigma & \text{otherwise}.
    \end{cases}
    \]
    The relation $\delta(p, \tau, \sigma, q, d) = \alpha$ can be interpreted as follows: if the machine is in state $p$ and the symbol on the tape head is $\tau$, then with amplitude $\alpha$ the machine writes the symbol $\sigma$ on the tape head, changes its state to $q$ and moves to the direction $d$. It is reasonable to require that $M$ is \textit{well-formed} in the sense that its time evolution operator $U$ is unitary, i.e. $U^\dag U = U U^\dag = I$. A deterministic Turing machine (DTM) can be regarded as a QTM with transition function $\delta: Q \times \Sigma \times \Sigma \times Q \times \{L, R\} \to \{0,1\}$ such that for every $p \in Q$ and $\sigma \in \Sigma$, there is a unique triple $(\tau,q,d) \in \Sigma\times Q \times \{L, R\}$ with $\delta(p, \sigma, \tau, q, d) = 1$. Moreover, a reversible Turing machine (RTM) is a well-formed DTM.

    The computation of $M$ begins at time $t = 0$. The initial configuration is prepared to be
        $\ket{c_0} = \ket{q_0, \mathcal{T}_0, 0}.$
    At each step, $M$ performs $U$ on the current configuration $\ket c$ and makes a measurement to determine whether the configuration is in the final state $q_f$. The measurement result is ``yes'' with probability $p = \Abs{P_FU\ket{c}}^2$ and ``no'' with probability $1-p$, where $P_F = \ket{q_f}_Q \bra{q_f}$. \begin{itemize}\item If the result is ``yes'', then $M$ halts with configuration $\frac 1 {\sqrt{p}} P_FU\ket{c}$; \item If the result is ``no'', then $M$ continues running with configuration $\frac 1 {\sqrt{1-p}} P_F^\perp U \ket{c}$, where $P_F^\perp = I-P_F$.\end{itemize}The probability that $M$ halts on $\ket{c_0}$ exactly at time $t$ is
    $
        p(t) = \Abs{\ket{c_t}}^2,
    $
    where $\ket{c_t} = P_F U (P_F^\perp U)^{t-1} \ket{c_0}$. $M$ is said to halt on $\ket{c_0}$ within time $T$ if
        $\sum_{t=1}^T p(t) = 1;$ in this case, the configuration after $M$ halts is a mixed state
 $\rho = \sum_{t=1}^T \ket{c_t} \bra{c_t}.$
    Especially, $M$ is said to halt on $\ket{c_0}$ exactly at time $T$, denoted $\ket{c_0} \xrightarrow[T]{M} \ket{c_T}$, if $p(T) = 1$.

    \begin{remark} It should be pointed out that the above halting scheme of QTMs is different from the traditional ones (e.g. \cite{Ber97, Oza98, Nis02}). Traditional QTMs are supposed to be either in a final state or in a non-final state (see Definition 3.11 of \cite{Ber97}, for example), and therefore always halt exactly at some time, which results in that the final configuration is always a pure quantum state. In our halting scheme, this restriction is removed and QTMs are allowed to be in a superposition of final and non-final states. Therefore, our halting scheme allows QTMs to halt at any time with a non-zero probability, and as a result, the final configuration is an ensemble of pure configurations, which is a mixed quantum state.
    \end{remark}

    \subsection{Multi-Track QTMs}

    The alphabet $\Sigma$ of a $k$-track QTM is regarded as the Cartesian product $\Sigma = \Sigma_1 \times \Sigma_2 \times \dots \times \Sigma_k$. In particular, let $\#_i$ be the blank symbol in $\Sigma_i$ for $1 \leq i \leq k$. For convenience, for every $x \in \Sigma^*$, we write $x$ to indicate the tape $\mathcal{X}$:
    \[
        \mathcal{X}(m) = \begin{cases}
            x(m) & 0 \leq m < \abs{x}, \\
            \# & \text{otherwise},
        \end{cases}
    \]
    where $x(m)$ is the $m$-th symbol of $x$. The joint of $k$ tapes $\mathcal{X}_1, \mathcal{X}_2, \dots, \mathcal{X}_k$ is the tape of $k$ tracks:
    \[
        (\mathcal{X}_1;\mathcal{X}_2;\dots;\mathcal{X}_k)(m) = (\mathcal{X}_1(m), \mathcal{X}_2(m), \dots, \mathcal{X}_k(m))
    \]
    for every $m \in \mathbb{Z}$. On input $x \in \{0, 1\}^*$, we put $x$ on the first track and leave the other tracks empty; that is, the initial tape $\mathcal{T}_x = x;\epsilon;\dots;\epsilon$, where $\epsilon$ denotes the empty string.
    Let $T: \mathbb{N} \to \mathbb{N}$. A $k$-track QTM $M$ is said to be within time $T(n)$, if for every $x \in \{0, 1\}^*$, $M$ halts on $\ket{q_0, \mathcal{T}_x, 0}$ within time $T(\abs{x})$, where $\abs{x}$ is the length of $x$. Especially, $M$ is said to be with exact time $T(n)$, if for every $x \in \{0, 1\}^*$, $M$ halts on $\ket{q_0, \mathcal{T}_x, 0}$ exactly at some time $\tau_x \leq T(\abs{x})$, where $\tau_x$ depends on $x$.

    Let $M$ be a QTM within time $T(n)$, and $\rho_x$ the configuration after $M$ halts on input $x$. The tape contents of $\rho_x$ are obtained by performing a measurement on each position $-T(\abs{x}) \leq m \leq T(\abs{x})$ (the head position $\xi$ never goes beyond $\abs{\xi} > T(\abs{x})$). We define function $\mathcal{T}_M: \{0, 1\}^* \times \Sigma^\# \to [0, 1]$ by letting for every $x \in \{0,1\}^*$, $\mathcal{T}_M(x, \mathcal{Z})$ be the probability that on input $x$, $M$ halts with tape $\mathcal{Z}$ after measurement. Formally, $\mathcal{T}_M(x, \mathcal{Z}) = \operatorname{tr}\left(M_{\mathcal{Z}} \rho_x M_{\mathcal{Z}}^\dag\right)$, where $M_{\mathcal{Z}} = \ket{\mathcal{Z}}_{\Sigma^\#}\bra{\mathcal{Z}}$. We note that if $\mathcal{Z}(m) \neq \#$ for some position $m$ with $\abs{m} > T(\abs{x})$, then $\mathcal{T}_M(x, \mathcal{Z}) = 0$.

    We design the output track to be the second track. Suppose $\mathcal{Z}$ is the tape after measurement. For $-T(\abs{x}) \leq m \leq T(\abs{x})$, let $y_m \in \Sigma_2$ be the symbol of $\mathcal{Z}(m)$ in the second track.
    Then the output $y$ is defined to be the concatenation of $y_m$'s for $-T(\abs{x}) \leq m \leq T(\abs{x})$ after ignoring blank symbols $\#_2$. We write $y = \text{extract}(\mathcal{Z})$ if the contents of tape $\mathcal{Z}$ implies output $y$. We assume that $y$ consists of only symbols $0$ and $1$, i.e. $y \in \{0, 1\}^*$. In this way, QTM $M$ defines a function $M: \{0, 1\}^* \times \{0, 1\}^* \to [0, 1]$ so that $M$ on input $x$ outputs $y$ with probability $M(x, y)$, i.e.
    \[
        M(x, y) = \sum_{\text{extract}(\mathcal{Z}) = y} \mathcal{T}_M(x, \mathcal{Z}).
    \]
    In case $M(x, y) = 1$ for some $x, y \in \{0,1\}^*$, we may write $M(x) = y$, indicating that $M$ on input $x$ outputs $y$ with certainty.

    \subsection{Standard QTMs}

    Note that the above definition of QTM is different from that in \cite{Ber97}, where QTMs are prevented from reaching a superposition in which some configurations are in state $q_f$ but others are not, and therefore intermediate measurements on the state of the QTM (i.e. to see whether the state is in $q_f$ or not) will not modify the configurations during the computation. However, QTMs defined above are allowed to reach a superposition of the final state and the other states. For our purpose, it is natural to assume that each time the configuration is measured (and therefore changes), the configuration will collapse to one of the configurations either in $q_f$ or not with probability according to the amplitudes.
In this subsection, we establish a connection between these two kinds of QTMs. Let us first introduce several terminologies.

    \begin{definition}[Stationary QTMs]
        A QTM $M$ is said to be stationary, if it halts on every $x \in \{0,1\}^*$ and
        \[
            \operatorname{tr}(P_0\rho_x) = 1,
        \]
        where $P_0 = \ket 0_\mathbb{Z} \bra 0$, and $\rho_x$ is the configuration after $M$ halts on $x$.
    \end{definition}

    \begin{definition} [Normal Form QTMs]
        A QTM $M$ is said to be in normal form, if for every $\tau, \sigma \in \Sigma$, $q \in Q$ and $d \in \{L,R\}$,
        \[
        \delta(q_f, \tau, \sigma, q, d) = \begin{cases}
            1 & (\sigma, q, d) = (\tau, q_0, R), \\
            0 & \text{otherwise}.
        \end{cases}
        \]
    \end{definition}

    \begin{definition} [Unidirectional QTMs]
        A QTM $M$ is said to be unidirectional, if for every $q \in Q$, there is a direction $d_q \in \{L, R\}$ such that
        \[
            \delta(p, \tau, \sigma, q, \bar d_q) = 0
        \]
        for every $p \in Q$ and $\tau, \sigma \in \Sigma$, where $\bar d$ denotes the reverse direction of $d$.
    \end{definition}

    Intuitively, a stationary QTM always halts with tape head at position $0$, i.e. the starting position. In a normal form QTM, the transitions from $q_f$ are technically specified for convenience, since every QTM always halts before any transition out of $q_f$. In a unidirectional QTM, any state can be entered from only one direction.

    Let $M$ be a normal form QTM and $\ket{c_0} = \ket{q_0, \mathcal{T}_0, \xi_0}$ its initial configuration. If the configuration $\ket{c}$ becomes $\ket{c'}$ after $t$ steps, i.e. $U (P_F^\perp U)^{t-1} \ket{c} = \ket{c'}$, we write
    $
        \ket{c} \xrightarrow[t]{M} \ket{c'}.
    $
    If $M$ halts on $\ket{c_0}$ exactly at time $T$ with the final state $\ket{c_f} = \ket{q_f, \mathcal{T}_f, \xi_f}$, then we write
    $
        \ket{\mathcal{T}_0, \xi_0} \xrightarrow[T]{M} \ket{\mathcal{T}_f, \xi_f}.
    $
    If $M$ is stationary (and thus) $\xi_0 = \xi_f = 0$, we simply write
    $
        \ket{\mathcal{T}_0} \xrightarrow[T]{M} \ket{\mathcal{T}_f}.
    $
    Moreover, if both $\ket{\mathcal{T}_0}$ and $\ket{\mathcal{T}_f}$ are in the computational basis, we often write
    $
        \mathcal{T}_0 \xrightarrow[T]{M} \mathcal{T}_f.
    $

    \begin{definition} [Standard QTM]
        A QTM $M$ is standard, if it is well-formed, normal form, stationary and unidirectional and there is a function $T: \mathbb{N} \to \mathbb{N}$ such that for every $x \in \{0,1\}^*$, $M$ on input $x$ halts exactly at time $T(\abs{x})$.
    \end{definition}

    For comparing different QTM models, we need the notion of time constructible function.

    \begin{definition} [Time Constructible Functions] \label{def-time-qtm}
        Let $T: \mathbb{N} \to \mathbb{N}$ and $T(n) \geq n$ for every $n \in \mathbb{N}$. $T(n)$ is said to be time constructible, if there is a standard QTM $M$ with exact time $O(T(n))$ such that for every $x \in \{0,1\}^*$,
        \begin{equation}\label{time-cons}
            x;\epsilon \xrightarrow[O(T(\abs{x}))]{M} x;T(\abs{x}).
        \end{equation}
    \end{definition}

    Note that in (\ref{time-cons}), a natural number $n \in \mathbb{N}$ written on a tape or track indicates a binary string $a = a_0a_1\dots a_{k-1} \in \{0,1\}^k$ such that $k$ is the smallest positive integer that $2^k > n$ and $n = \sum_{i=0}^{k-1} 2^{k-i-1} a_{i}.$

    Now we are able to show our first result that every well-formed QTM can be efficiently simulated by a standard QTM.  For any QTM $M$, we write $\mathbb{C}(M)$ for the set of its transition coefficients of a QTM $M$; that is,
    \[
      \mathbb{C}(M)   = \{ \delta(p, \sigma, \tau, q, d): p, q \in Q, \sigma, \tau \in \Sigma, d \in \{L, R\} \}.
    \]

    \begin{theorem}[Standardisation] \label{thm-qtm}
        Let $T: \mathbb{N} \to \mathbb{N}$ be a function time constructible by QTM. For every well-formed and normal form QTM $M$ within time $T(n)$, there is a standard QTM $M'$ with exact time $O(T(n)\log^2 T(n))$ such that
        \begin{enumerate}
          \item $M(x, y) = M'(x, y)$ for every $x, y \in \{0,1\}^*$.
          \item $\mathbb{C}(M') \subseteq \mathbb{C}(M) \cup \{0, 1\}$.
        \end{enumerate}
    \end{theorem}

   The proof of Theorem \ref{thm-qtm} is given in Section \ref{sec:standard-qtm}.

    \section{QRAMs and QRASPs} \label{sec:qram-qrasp}

    It is a common sense in the existing literature that a practical quantum computer (of the first generation at least) consists of a classical computer with access to quantum registers, where the classical part performs classical computations and controls the evolution of quantum registers, and the quantum part can be initialised in certain states (e.g. basis state $\ket0$), perform elementary unitary operations (e.g. Hadamard, $\pi/8$ and CNOT gates), and be measured with the outcomes sent to the classical machine. The models of QRAMs and QRASPs are defined based on this intuition.

    \subsection{QRAMs}

Formally, a quantum random access machine (QRAM) is a program $P$, i.e. a finite sequence of QRAM instructions operating on an infinite sequence of both classical and quantum registers. Each classical register holds an arbitrary integer (positive, negative, or zero), while each quantum register holds a qubit (in state $\ket 0$, $\ket 1$ or their superposition). The contents of the $i$-th ($i \geq 0$) classical (resp. quantum) register is denoted by $X_i$ (resp. $Q_i$).

Associated with the machine is a cost function $l(n)$, which denotes the memory required to store, or the time required to load the number $n$. Two forms of $l(n)$ commonly used in studying classical RAMs are:
    \begin{enumerate}
      \item $l(n)$ is a constant, i.e. $l(n) = O(1)$; and
      \item $l(n)$ is logarithmic, i.e. $l(n) = O(\log \abs{n})$, where $\abs{n}$ is the absolute value of $n$.
    \end{enumerate}

\begin{definition} The instructions for QRAM and their execution times are given in Table \ref{tab1}.
    \begin{table}[!htp]
        \centering
        \caption{QRAM instructions}
        \begin{tabular}{lll}
        \hline
        Type & Instruction & Execution time \\
        \hline
        Classical & $X_i \gets C$, $C$ any integer & $1$ \\
        Classical & $X_i \gets X_j+X_k$ & $l(X_j)+l(X_k)$ \\
        Classical & $X_i \gets X_j-X_k$ & $l(X_j)+l(X_k)$ \\
        Classical & $X_i \gets X_{X_j}$ & $l(X_j)+l(X_{X_j})$ \\
        Classical & $X_{X_i} \gets X_j$ & $l(X_i)+l(X_j)$ \\
        Classical & TRA $m$ if $X_j > 0$ & $l(X_j)$ \\
        Classical & READ $X_i$ & $l(\text{input})$ \\
        Classical & WRITE $X_i$ & $l(X_i)$ \\
        \hline
        Quantum & $\text{CNOT}[Q_{X_i}, Q_{X_j}]$ & $l(X_i)+l(X_j)$ \\
        Quantum & $H[Q_{X_i}]$ & $l(X_i)$ \\
        Quantum & $T[Q_{X_i}]$ & $l(X_i)$ \\
        \hline
        Measurement & $X_i \gets M[Q_{X_j}]$ & $l(X_j)$ \\
        \hline
        \end{tabular}
        \label{tab1}
    \end{table}\end{definition}

The QRAM instructions in Table \ref{tab1} are divided into two types. The classical-type instructions  are the same as those adopted in \cite{Coo73}. Here, $i, j, k$ are any nonnegative integers, and $m$ are integers between $0$ and $L$ (inclusive), where $L$ is the length of the QRAM program and also denotes termination (see Section \ref{sec:qram-operation} for details about $m$). The effect of most of the instructions are obvious. For example, $X_i \gets C$ causes $X_i$ to hold $C$, while $X_i \gets X_j \pm X_k$ causes $X_i$ to hold the calculation result of $X_j \pm X_k$. The instruction TRA $m$ if $X_j > 0$ causes the $m$-th instruction to be the next instruction to execute if $X_j > 0$. READ $X_i$ causes $X_i$ to hold the next input number on the input tape, while WRITE $X_i$ causes $X_i$ to be printed on the output tape. The indirect instruction $X_i \gets X_{X_j}$ causes $X_i$ to hold $X_{X_j}$, provided $X_j \geq 0$, while $X_{X_i} \gets X_j$ causes $X_{X_i}$ to hold $X_j$, provided $X_i \geq 0$. The indirect instructions allow a fixed program to access unbounded registers. It should be noted that classical registers are needed as classical address to indirectly access quantum registers.

 The quantum-type instructions include quantum gates and measurements. For simplicity of presentation, we choose to use a minimal but universal set of quantum gates: CNOT (the Controlled-NOT gate), $H$ (the Hadamard gate) and $T$ (the $\pi/8$ gate). Indeed, any finite universal set of quantum gates is acceptable. The measurement instruction $X_i \gets M[Q_{X_j}]$ is a bridge between classical and quantum registers, which causes $X_i$ to hold the measurement result of $Q_{X_j}$ in the computational basis, provided $X_j \geq 0$.

    \subsection{QRASPs}

    A quantum random access stored-program machine (QRASP) is a program $P$, i.e. a finite sequence of QRASP instructions operating on infinite sequences of both classical and quantum registers.

    \begin{definition} The instructions for QRASP are given in Table \ref{tab2}.
    \begin{table}[!htp]
        \centering
        \caption{QRASP instructions}
        \begin{tabular}{cclcll}
        \hline
        Type & Operation & Mnemonic & code & Description & Execution time \\
        \hline
        \specialrule{0em}{1.5pt}{1.5pt}
        Classical & load constant & LOD, $j$ & 1 & \tabincell{l}{$\text{AC} \gets j$; \\ $\text{IC} \gets \text{IC}+2$ } & $l(\text{IC})+l(j)$ \\
        \specialrule{0em}{3pt}{3pt}
        Classical & add & ADD, $j$ & 2 & \tabincell{l}{$\text{AC} \gets \text{AC}+X_j$; \\ $\text{IC} \gets \text{IC}+2$ } & \tabincell{l}{$l(\text{IC})+l(j)+$\\$l(\text{AC})+l(X_j)$} \\
        \specialrule{0em}{3pt}{3pt}
        Classical & subtract & SUB, $j$ & 3 & \tabincell{l}{$\text{AC} \gets \text{AC}-X_j$; \\ $\text{IC} \gets \text{IC}+2$ } & \tabincell{l}{$l(\text{IC})+l(j)+$\\$l(\text{AC})+l(X_j)$} \\
        \specialrule{0em}{3pt}{3pt}
        Classical & store & STO, $j$ & 4 & \tabincell{l}{$X_j \gets \text{AC}$; \\ $\text{IC} \gets \text{IC}+2$ } & \tabincell{l}{$l(\text{IC})+l(j)+$\\$l(\text{AC})$} \\
        \specialrule{0em}{3pt}{3pt}
        Classical & \tabincell{c}{branch on\\positive\\accumulator} & BPA, $j$ & 5 & \tabincell{l}{if $\text{AC} > 0$ then \\ $\text{IC} \gets j$; otherwise \\$\text{IC} \gets \text{IC}$+2 } & \tabincell{l}{$l(\text{IC})+l(j)+$\\$l(\text{AC})$} \\
        \specialrule{0em}{3pt}{3pt}
        Classical & read & RD, $j$ & 6 & \tabincell{l}{$X_j \gets \text{next input}$; \\ $\text{IC} \gets \text{IC}+2$ } & \tabincell{l}{$l(\text{IC})+l(j)+$\\$l(\text{input})$} \\
        \specialrule{0em}{3pt}{3pt}
        Classical & print & PRI, $j$ & 7 & \tabincell{l}{output $X_j$; \\ $\text{IC} \gets \text{IC}+2$ } & \tabincell{l}{$l(\text{IC})+l(j)+$\\$l(X_j)$} \\
        \specialrule{0em}{1.5pt}{1.5pt}
        \hline
        \specialrule{0em}{1.5pt}{1.5pt}
        Quantum & CNOT & CNOT, $j$, $k$ & 8 & \tabincell{l}{$\text{CNOT}[Q_{j}, Q_{k}]$; \\ $\text{IC} \gets \text{IC}+3$ } & \tabincell{l}{$l(\text{IC})+l(j)+$\\$l(k)$} \\
        \specialrule{0em}{3pt}{3pt}
        Quantum & H & H, $j$ & 9 & \tabincell{l}{$H[Q_j]$; \\ $\text{IC} \gets \text{IC}+2$ } & $l(\text{IC})+l(j)$ \\
        \specialrule{0em}{3pt}{3pt}
        Quantum & T & T, $j$ & 10 & \tabincell{l}{$T[Q_j]$; \\ $\text{IC} \gets \text{IC}+2$ } & $l(\text{IC})+l(j)$ \\
        \specialrule{0em}{1.5pt}{1.5pt}
        \hline
        \specialrule{0em}{1.5pt}{1.5pt}
        Measurement & measure & MEA, $j$ & 11 & \tabincell{l}{$\text{AC} \gets M[Q_j]$; \\ $\text{IC} \gets \text{IC}+2$ } & $l(\text{IC})+l(j)$ \\
        \specialrule{0em}{1.5pt}{1.5pt}
        \hline
        \specialrule{0em}{1.5pt}{1.5pt}
        Termination & halt & HLT & - & stop & $l(\text{IC})+l(X_{\text{IC}})$ \\
        \specialrule{0em}{1.5pt}{1.5pt}
        \hline
        \end{tabular}
        \label{tab2}
    \end{table}\end{definition}

  Note that the classical-type QRASP instructions are the same as RASP instructions defined in \cite{Coo73}, and the quantum-type QRASP instructions are the same as in QRAMs.

 Strictly speaking, a QRASP is a finite sequence of integers that are to be interpreted into QRASP instructions during the execution rather than an explicit program. The reason is that QRASP may modify itself during the execution and causes unpredictable interpreted QRASP instructions. Our machine has an accumulator (AC), which holds an arbitrary integer, an instruction counter (IC), and two infinite sequences of both classical and quantum registers. Each classical register $X_i$ holds an arbitrary integer, while each quantum register $Q_i$ holds a qubit. An instruction is stored in two or three consecutive classical registers depending on its operation code. The first classical register contains an operation code (shown in Table \ref{tab2}). In case that the operation code is beyond the range 1 to 11, the execution immediately terminates. The second (and the third if needed) classical register contains the parameter of the instruction. In fact, only the CNOT operation needs two parameters while other operations do not. It is noted that indirect addressing is not allowed in QRASP, the programs need to modify themselves in order to access unbounded number of (both classical and quantum) registers.

    \section{Main Results} \label{sec:main-results}

    In this section, we state our main results, which clarify the relationships between QTMs, QRAMs and QRASPs.

    The relationship between QRAMs and QRASPs is simple. We prove that QRAMs and QRASPs can simulate each other with constant slowdown. For a QRAM (or QRASP) $P$ and $x, y \in \{0,1\}^*$, let $P(x, y)$ denote the probability that $P$ on input $x$ outputs $y$ (see Section \ref{sec:def-qram} and Section \ref{sec:def-qrasp} for its formal definition).

    \begin{theorem}\label{thm-qram-qrasp}
        Let $T: \mathbb{N} \to \mathbb{N}$ with $T(n) \geq n$.
        \begin{enumerate}
          \item For every $T(n)$-time QRAM $P$, there is a $O(T(n))$-time QRASP $P'$ such that for every $x, y \in \{0, 1\}^*$, $P(x, y) = P'(x, y)$.
          \item For every $T(n)$-time QRASP $P$, there is a $O(T(n))$-time QRAM $P'$ such that for every $x, y \in \{0, 1\}^*$, $P(x, y) = P'(x, y)$.
        \end{enumerate}
    \end{theorem}

   To further compare QRAMs (and thus QRASPs) with QTMs, we need a QRAM variant of time constructible functions.

    \begin{definition} [QRAM-Time Constructible Functions]
        Let $T: \mathbb{N} \to \mathbb{N}$ and $T(n) \geq n$ for every $n \in \mathbb{N}$. $T(n)$ is said to be QRAM-time constructible, if there is an $O(T(n))$-time QRAM $P$ such that for every $x \in \{0,1\}^*$,
        \[
            P(x, T(\abs{x})) = 1,
        \]
        where $T(\abs{x})$ denotes its binary form as in Definition \ref{def-time-qtm}.
    \end{definition}

    The relationship between QRAMs and QTMs is then established in the following:

    \begin{theorem} \label{thm-qram-qtm}
        \begin{enumerate}
          \item Let $T: \mathbb{N} \to \mathbb{N}$. Suppose $P$ is a $T(n)$-time QRAM. Then there is a well-formed and normal form QTM $M$ within time $T'(n)$ such that
          \begin{enumerate}
            \item $P(x, y) = M(x, y)$ for every $x, y \in \{0,1\}^*$.
            \item $\mathbb{C}(M) = \{ 0, 1, \frac 1 {\sqrt2}, -\frac 1 {\sqrt2}, \exp(i\pi/4) \}$.
          \end{enumerate}
          Moreover,
          \begin{enumerate}
            \item If $l(n)$ is logarithmic, then $T'(n) = O(T(n)^4)$.
            \item If $l(n)$ is constant, then $T'(n) = O(T(n)^8)$.
          \end{enumerate}
          \item Let $T: \mathbb{N} \to \mathbb{N}$ be a QRAM-time constructible function, and $\lambda: \mathbb{N} \to \mathbb{N}$ with $\lambda(n) \geq n$. For every standard QTM $M$ with exact time $T(n)$ and $\mathbb{C}(M) \subseteq \mathbb{C}(\lambda(n))$, there is a constant $c > 0$ such that for every $0 < \varepsilon < 1$, there is a $O(T(n)^2 (\lambda(\log (T(n)/\varepsilon)))^c)$-time QRAM $P$ such that $\abs{M(x, y)-P(x, y)} < \varepsilon$ for every $x, y \in \{0,1\}^*$.
        \end{enumerate}
    \end{theorem}

 It is well-known \cite{Bre75} that square roots are computable in polynomial time in $n$ with precision $n$, and thus $\sqrt2 \in \tilde{\mathbb{C}}$. So, it holds that $\mathbb{C}(M) \subseteq \tilde{\mathbb{C}}$ in the first part of the above theorem.

 A combination of the above two theorems and Theorem \ref{thm-qtm} indicates that QTMs, QRAMs and QRASPs can simulate each other with polynomial slowdown.

The above results have some simple corollaries on quantum complexity classes. To present them, let us first recall the definitions of \textbf{EQP} and \textbf{BQP} from \cite{Ber97,Adl97}.

    \begin{definition}
        Let $L \subseteq \{0,1\}^*$. The language $L$ is said to be in $\textbf{EQP}$, if there is a well-formed, normal form and stationary multi-track QTM $M$ with exact time $T(n)$, satisfying:
        \begin{enumerate}
          \item $x \in L$ if and only if $M(x, 1) = 1$;
          \item $x \notin L$ if and only if $M(x, 1) = 0$;
          \item $T(n)$ is a polynomial in $n$.
        \end{enumerate}
        The language $L$ is said to be in $\textbf{BQP}$, if there is a well-formed, normal form, stationary, multi-track QTM $M$ with exact time $T(n)$,
        \begin{enumerate}
          \item $x \in L$ if and only if $M(x, 1) \geq \frac 2 3$;
          \item $x \notin L$ if and only if $M(x, 1) \leq \frac 1 3$;
          \item $T(n)$ is a polynomial in $n$.
        \end{enumerate}
    \end{definition}

    The complexity classes \textbf{EQP'} and \textbf{BQP'} are defined by removing the stationary condition and allowing QTMs to be within time $T(n)$ in the above definition. Immediately from Theorem \ref{thm-qtm}, we have:
    \begin{proposition} \label{prop-eqp-bqp}
        $\textbf{EQP} = \textbf{EQP'}$ and $\textbf{BQP} = \textbf{BQP'}$.
    \end{proposition}

    The QTM in the first part of Theorem \ref{thm-qram-qtm} is only guaranteed to be well-formed and normal form, but not to halt at an exact time. But Theorem \ref{thm-qtm} can be employed to strengthen it by a standard QTM with exact time $O(T(n)^4 \log^2 T(n))$ for logarithmic $l(n)$, and $O(T(n)^8 \log^2 T(n))$ for constant $l(n)$, provided $T(n)$ is time constructible.

    Let \textbf{EQRAMP} and \textbf{BQRAMP} denote the classes of languages that are computable by exact and bounded-error quantum random access machine in polynomial time, respectively (see Section \ref{sec:qram-computing} for their formal definitions). Then by Theorem \ref{thm-qram-qtm} and Proposition \ref{prop-eqp-bqp}, immediately we have:
    \begin{theorem} \label{thm-relation-bqp}
        $\textbf{P} \subseteq \textbf{EQRAMP} \subseteq \textbf{EQP}$ and $\textbf{BQP} = \textbf{BQRAMP}$.
    \end{theorem}

 Note that \textbf{BQP} and \textbf{BQRAMP} coincide, but it seems that \textbf{EQP} and \textbf{EQRAMP} do not. This is because a QRAM only has a finite number of quantum gate operations, while a QTM can have an infinite (but countable) number of quantum gate operations. It is almost impossible to simulate an infinite set of quantum gates by a finite set of quantum gates with no errors. On the other hand, in the definitions of \textbf{EQP} and \textbf{EQRAMP}, the probabilities are restricted to $0$ and $1$, which makes it unlikely that the two complexity classes coincide.


    \section{Computations of QRAMs} \label{sec:def-qram}

For this section on, we provide all of the details for proving our main results. In this section, we carefully describe the computations of QRAMs. It is presented in the forms of operational and denotational semantics of QRAMs in Subsections \ref{sec:qram-operation} and \ref{sec:qram-denotational}, respectively. QRAM computation and the notion of two complexity classes \textbf{EQRAMP} and \textbf{BQRAMP} are formally defined in Section \ref{sec:qram-computing}.

Section \ref{sec:qram-address-shifting} gives a useful method for shifting addresses, based on which we show that every QRAM can be simulated by an address-safe QRAM in the sense that it never accesses invalid addresses. Postponing measurements is a widely used technique in quantum computing. In Section \ref{sec:qram-measurement-postponed}, we introduce the notion of measurement-postponed QRAMs.

    \subsection{Operational semantics} \label{sec:qram-operation}

    Formally, a QRAM is represented by a sequence $P = P_0, P_1, P_2, \dots, P_{L-1}$ of instructions with $L = \abs{P}$ being the length of $P$, i.e. the number of instructions in this QRAM. In the execution of a QRAM, there is an instruction counter (IC) indicating which instruction to be executed. A configuration of the QRAM is a tuple $(\xi, \mu, \ket\psi, x, y)$, where:
    \begin{enumerate}
      \item $\xi \in \mathbb{N} \cup \{\downarrow\}$ denotes the current IC, with $\downarrow$ indicating the end of execution;
      \item $\mu: \mathbb{N} \to \mathbb{Z}$ is the description of all contents of classical registers;
      \item $\ket\psi \in \mathcal{H}= \bigotimes_{i=0}^\infty \mathcal{H}_i$ is the state of quantum registers, with $\mathcal{H}_i = \operatorname{span} \{ \ket 0_i, \ket 1_i \}$;
      \item $x \in \mathbb{Z}^\omega$ is a sequence of integers to read on the input tape;
      \item $y \in \mathbb{Z}^*$ is a sequence of printed integers on the output tape.
    \end{enumerate}
We write $\mathcal{C} = (\mathbb{N} \cup \{\downarrow\}) \times \mathbb{Z}^\mathbb{N} \times \mathcal{H} \times \mathbb{Z}^\omega \times \mathbb{Z}^*$ for the set of all configurations. A configuration $c = (\xi, \mu, \ket\psi, x, y) \in \mathcal{C}$ is called a terminal configuration if $\xi = \downarrow$. Let $\mathcal{C}_f \subseteq \mathcal{C}$ denote the set of all terminal configurations.

    The execution transition is a function $\rightarrow: \mathcal{C} \times \mathcal{C} \to [0, 1] \times \mathbb{N}$.
    For two configurations $c$ and $c'$, $\rightarrow(c, c') = (p, T)$
    means that configuration $c$ is changed to configuration $c'$ in time $T$ with probability $p$ after executing one instruction. For readability, we write
    \[
        c \xrightarrow[T]{p} c',
    \]
    wherein $p$ can be ignored if $p = 1$. Then the operational semantics of QRAMs is defined by the following transitional rules:
    \begin{enumerate}
      \item If $\xi$ is out of range $[0, L)$,
      \[
      (\xi, \mu, \ket\psi, x, y) \xrightarrow[1]{} (\downarrow, \mu, \ket\psi, x, y).
      \]
      \item If $P_\xi$ has the form $X_i \gets C$,
      \[
      (\xi, \mu, \ket\psi, x, y) \xrightarrow[1]{} (\xi+1, \mu_{i}^C, \ket\psi, x, y),
      \]
      where
      \[
        \mu_a^b = \begin{cases}
            \mu(j) & j \neq a, \\
            b & j = a.
        \end{cases}
      \]
      \item If $P_\xi$ has the form $X_i \gets X_j+X_k$,
      \[
      (\xi, \mu, \ket\psi, x, y) \xrightarrow[l(\mu(j))+l(\mu(k))]{} (\xi+1, \mu_{i}^{\mu(j)+\mu(k)}, \ket\psi, x, y).
      \]
      \item If $P_\xi$ has the form $X_i \gets X_j-X_k$,
      \[
      (\xi, \mu, \ket\psi, x, y) \xrightarrow[l(\mu(j))+l(\mu(k))]{} (\xi+1, \mu_{i}^{\mu(j)-\mu(k)}, \ket\psi, x, y).
      \]
      \item If $P_\xi$ has the form $X_i \gets X_{X_j}$, whenever $\mu(j) \geq 0$, then
      \[
      (\xi, \mu, \ket\psi, x, y) \xrightarrow[l(\mu(j))+l(\mu(\mu(j)))]{} (\xi+1, \mu_{i}^{\mu(\mu(j))}, \ket\psi, x, y);
      \]
      otherwise,
      \[
      (\xi, \mu, \ket\psi, x, y) \xrightarrow[l(\mu(j))]{} (\downarrow, \mu, \ket\psi, x, y).
      \]
      \item If $P_\xi$ has the form $X_{X_i} \gets X_j$, whenever $\mu(i) \geq 0$, then
      \[
      (\xi, \mu, \ket\psi, x, y) \xrightarrow[l(\mu(i))+l(\mu(j))]{} (\xi+1, \mu_{\mu(i)}^{\mu(j)}, \ket\psi, x, y);
      \]
      otherwise,
      \[
      (\xi, \mu, \ket\psi, x, y) \xrightarrow[l(\mu(i))]{} (\downarrow, \mu, \ket\psi, x, y).
      \]
      \item If $P_\xi$ has the form TRA $m$ if $X_j > 0$, whenever $\mu(j) > 0$, then
      \[
      (\xi, \mu, \ket\psi, x, y) \xrightarrow[l(\mu(j))]{} (m, \mu, \ket\psi, x, y);
      \]
      otherwise,
      \[
      (\xi, \mu, \ket\psi, x, y) \xrightarrow[l(\mu(j))]{} (\xi+1, \mu, \ket\psi, x, y).
      \]
      \item If $P_\xi$ has the form READ $X_i$, let $a$ be the first integer in $x$, then
      \[
      (\xi, \mu, \ket\psi, x, y) \xrightarrow[{l(a)}]{} (\xi+1, \mu_i^{a}, \ket\psi, x', y),
      \]
      where $x'$ denotes the string obtained by deleting the first integer $x$.
      \item If $P_\xi$ has the form WRITE $X_i$,
      \[
      (\xi, \mu, \ket\psi, x, y) \xrightarrow[{l(\mu(i))}]{} (\xi+1, \mu, \ket\psi, x, y'),
      \]
      where $y'$ denotes the string obtained by appending an integer $\mu(i)$ to $y$.
      \item If $P_\xi$ has the form $\text{CNOT}[Q_{X_i}, Q_{X_j}]$, whenever $\mu(i) \geq 0$ and $\mu(j) \geq 0$, then
      \[
      (\xi, \mu, \ket\psi, x, y) \xrightarrow[{l(\mu(i))+l(\mu(j))}]{} (\xi+1, \mu, \text{CNOT}_{\mu(i), \mu(j)}\ket\psi, x, y),
      \]
      where for $a_0, a_1, a_2, \dots \in \{0, 1\}$, $\text{CNOT}_{i,j} \ket{a_0, a_1, a_2, \dots} = \ket{b_0, b_1, b_2, \dots}$ with $$b_k =\begin{cases} a_k &{\rm if}\ k \neq j,\\ a_i \oplus a_j &{\rm otherwise},\end{cases}$$ and $\oplus$ denotes modulo-2 addition;
      otherwise,
      \[
      (\xi, \mu, \ket\psi, x, y) \xrightarrow[{l(\mu(i))+l(\mu(j))}]{} (\downarrow, \mu, \ket\psi, x, y).
      \]
      \item If $P_\xi$ has the form $A[Q_{X_i}]$ with $A=H$ or $T$, whenever $\mu(i) \geq 0$, then
      \[
      (\xi, \mu, \ket\psi, x, y) \xrightarrow[{l(\mu(i))}]{} (\xi+1, \mu, A_{\mu(i)}\ket\psi, x, y),
      \]
      where
      \[
      A_i \bigotimes_{j=0}^\infty \ket{a_j} = \bigotimes_{j=0}^{i-1} \ket{a_j} \otimes A\ket{a_i} \otimes \bigotimes_{j=i+1}^{\infty} \ket{a_j}
      \]
      for $a_0, a_1, a_2, \dots \in \{0, 1\}$;
      otherwise,
      \[
      (\xi, \mu, \ket\psi, x, y) \xrightarrow[{l(\mu(i))}]{} (\downarrow, \mu, \ket\psi, x, y).
      \]
          \item If $P_\xi$ has the form $X_i \gets M[Q_{X_j}]$, whenever $\mu(j) \geq 0$, then
      \[
      (\xi, \mu, \ket\psi, x, y) \xrightarrow[{l(\mu(j))}]{\Abs{P_{j}\ket\psi}^2} (\xi+1, \mu_i^0, \frac{P_{j}\ket\psi} {\Abs{P_{j}\ket\psi}^2}, x, y),
      \]
      \[
      (\xi, \mu, \ket\psi, x, y) \xrightarrow[{l(\mu(j))}]{1-\Abs{P_{j}\ket\psi}^2} (\xi+1, \mu_i^1, \frac{(I-P_{j})\ket\psi} {1-\Abs{P_{j}\ket\psi}^2}, x, y),
      \]
      where
      \[
      P_j = \bigotimes_{i=0}^{j-1} I_i \otimes \ket{0}_j\bra{0} \otimes \bigotimes_{i=j+1}^{\infty} I_i;
      \]
      otherwise,
      \[
      (\xi, \mu, \ket\psi, x, y) \xrightarrow[{l(\mu(i))}]{} (\downarrow, \mu, \ket\psi, x, y).
      \]
    \end{enumerate}

    \subsection{Denotational semantics} \label{sec:qram-denotational}

    An execution path $\pi = c_0, c_1, \dots, c_n$ is a non-empty sequence of configurations, i.e. $\pi \in \mathcal{C}^+$. We define the length of $\pi$ as $\abs{\pi} = n$. A path is called terminal, if $c_n$ is a terminal configuration. For readability, an execution path $\pi$ is usually written as
    \[
        \pi: c_0 \xrightarrow[T_1]{p_1} c_1 \xrightarrow[T_2]{p_2} c_2 \xrightarrow[T_3]{p_3} \dots \xrightarrow[T_{n-1}]{p_{n-1}} c_{n-1} \xrightarrow[T_{n}]{p_{n}} c_n,
    \]
    where for every $1 \leq i \leq n$, $c_{i-1} \xrightarrow[T_i]{p_i} c_i$ is a transition defined in the above subsection.
    We may simply write:
    \[
        \pi: c_0 \xrightarrow[T]{p}^n c_n,
    \]
    where
    \[
        T = \sum_{i=1}^n T_i, \qquad p = \prod_{i=1}^n p_i.
    \]
    It should be noted that there could be multiple paths of the form $c_0 \xrightarrow[T]{p}^n c_n$ with different pairs of $p$ and $T$.

    For simplicity of presentation, let us introduce several abbreviations:
    \begin{itemize}\item $\pi.p_1$ (resp. $\pi.T_1$) denotes the transition probability (resp. time) in the first step of $\pi$.
    \item $\pi.p$ (resp. $\pi.T$) denotes the transition probability (resp. time) of $\pi$.
    \item $\pi.c_{\abs{\pi}}.\xi =\ \downarrow$ means that the last configuration of $\pi$ is a terminal configuration; that is, $\pi.c_{\abs{\pi}} \in \mathcal{C}_f$, where $\abs{\pi}$ is the length of $\pi$.
\end{itemize}
    Moreover, for a QRAM $P$, we use the following notations: \begin{itemize}\item $\mathcal{P}$ denotes the set of all execution paths of $P$.
    \item $\mathcal{P}(c)$ denotes the set of all execution paths starting from $c$ (with positive probabilities), i.e. $$\mathcal{P}(c) = \{ \pi \in \mathcal{P}: \pi.c_0 = c \text{ and } \pi.p > 0 \}.$$
    \item $\mathcal{P}_f(c)$ denotes the set of all terminal execution paths starting from $c$, i.e. $$\mathcal{P}_f(c) = \{ \pi \in \mathcal{P}(c): \pi.c_{\abs{\pi}} \in \mathcal{C}_f \}.$$
    \item $\mathcal{P}^{=n}(c)$ denotes the set of all execution paths of length $n$, starting from $c$, i.e. $$\mathcal{P}^{=n}(c) = \{ \pi \in \mathcal{P}(c): \abs{\pi} = n \}.$$
    \item $\mathcal{P}^{=n}_f(c)$ denotes the set of all terminal execution paths of length $n$, starting from $c$, i.e. $$\mathcal{P}^{=n}_f(c) = \{ \pi \in \mathcal{P}^{=n}(c): \pi.c_n \in \mathcal{C}_f \}.$$
    \item $\mathcal{P}^n(c)$ denotes the set of all execution paths starting from $c$ within $n$ steps, i.e.
    \[
    \mathcal{P}^n(c) = \bigcup_{m=0}^{n} \mathcal{P}^{=m}_f(c).
    \]\end{itemize}

    \begin{definition}\begin{enumerate}\item The $n$-step semantics function $\llbracket P \rrbracket^n: \mathcal{C} \to (\mathcal{C} \to [0, 1])$ is defined by
    \[
        \llbracket P \rrbracket^n (c) = \sum_{\pi \in \mathcal{P}^n(c)} \pi.p \cdot \mathbb{I}_{\pi.c_{\abs{\pi}}},
    \]
    where
    \[
        \mathbb{I}_c(c') = \begin{cases}
            1 & c = c', \\
            0 & \text{otherwise}.
        \end{cases}
    \]
    \item The semantic function $\llbracket P \rrbracket: \mathcal{C} \to (\mathcal{C} \to [0, 1])$ is defined by
    \[
        \llbracket P \rrbracket(c) = \lim_{n \to \infty} \llbracket P \rrbracket^n(c).
    \]\end{enumerate}\end{definition}

    It should be noted that $\llbracket P \rrbracket(c)$ may not exist.

    \begin{definition} The worst case running time $\tau_P: \mathcal{C} \to \mathbb{N} \cup \{ \infty \}$
    is defined by
    \[
        \tau_P (c) = \sup \{ \pi.T: \pi \in \mathcal{P}(c) \}.
    \]\end{definition}

    The following lemma is straightforward:

    \begin{lemma}\label{lem-worst-time}
       For any $c \in \mathcal{C}$, we have:\begin{enumerate}
        \item if $\tau_P(c) < \infty$, then $\llbracket P \rrbracket(c) = \llbracket P \rrbracket^{\tau_P(c)}(c)$;
        \item if $\tau_P(c) < \infty$ and $\llbracket P \rrbracket(c)(c') > 0$ for some $c' \in \mathcal{C}$, then $c' \in \mathcal{C}_f$.
   \end{enumerate} \end{lemma}

    \subsection{QRAM computations} \label{sec:qram-computing}

    We are interested in the time required for a QRAM to recognize a language on a finite alphabet $\Sigma = \{\sigma_0, \sigma_1, \sigma_2, \dots, \sigma_{m-1}\}$. An input string $x = \sigma_{i_0} \sigma_{i_1} \dots \sigma_{i_{n-1}}$ is represented in the machine as the sequence of integers $i_0, i_1, \dots, i_{n-1}, (-1)^\omega$, where the infinite occurrences of $-1$ at the end of the sequence indicate the end of the string. This input convention guarantees that whenever an instruction of the form READ $X_i$ is executed, $X_i$ always obtain $-1$ if the string has been read up. We use $\mathit{in}: \Sigma^* \to \mathbb{Z}^\omega$ to denote this conversion from an input string $x$ to the contents on input tape $\mathit{in}(x)$.

    After the execution of the QRAM, a finite sequence $y$ of integers is obtained on the output tape. In order to extract the output string on $\Sigma$ from the contents on the output tape, we define $\mathit{out}: \mathbb{Z} \to \Sigma$ by
    \[
        \mathit{out}(n) = \begin{cases}
            \sigma_n & 0 \leq n < m-1, \\
            \sigma_{m-1} & \text{otherwise}.
        \end{cases}
    \]
    This function can be extended to $\mathit{out}: \mathbb{Z}^* \to \Sigma^*$ by concatenation of each single conversion.

    \begin{example} Consider the simplest alphabet $\Sigma = \{0, 1\}$. The input string $x = 0101$ is converted to $\mathit{in}(x) = 0, 1, 0, 1, -1, -1, -1, \dots$. An output string $\mathit{out}(y) = 0111$ is extracted from the contents $y = 0, 1, 2, 3$ on the output tape.
    \end{example}

    With the above input/output conventions, we can now describe the notion of QRAM computation. Before the computation, $\text{IC}$ is set to $0$ initially, with all classical registers being zero and all quantum registers being $\ket 0$; that is, $\mu_0 = 0$ and
    \[
    \ket{\psi_0} = \bigotimes_{i=0}^\infty \ket 0_i.
    \]
    Suppose the input string is $x$, then the initial configuration is $c_x = (0, \mu_0, \ket{\psi_0}, \mathit{in}(x), \epsilon)$. The distribution $\mathcal{D}: \Sigma^* \to (\mathcal{C} \to [0, 1])$ of terminal configurations is:
    \[
        \mathcal{D}(x) = \llbracket P \rrbracket (c_x).
    \]
    The computational result $P: \Sigma^* \times \Sigma^* \to [0, 1]$ of QRAM $P$ is defined by
    \[
        P(x, y) = \sum_{c \in \mathcal{C}_f: \mathit{out}(c.y) = y} \mathcal{D}(x)(c),
    \]
    and the worst case running time $T_P: \Sigma^* \to \mathbb{N} \cup \{\infty\}$ is defined by
    \[
        T_P(x) = \tau_P (c_x).
    \]

 \begin{definition} Let $T: \mathbb{N} \to \mathbb{N}$. $P$ is said to be a $T(n)$-time QRAM, if for every $x \in \Sigma^*$, $T_P(x) \leq T(\abs{x})$.\end{definition}

 In particular, $P$ is said to be a polynomial-time QRAM, if it is a $p(n)$-time QRAM for some polynomial $p$. Furthermore, two complexity classes are defined as follows:
 \begin{itemize}\item \textbf{EQRAMP} stands for Exact Quantum Random Access Machine Polynomial-time. More precisely, a language $L \subseteq \{0,1\}^*$ is said to be in \textbf{EQRAMP}, if there is a polynomial-time QRAM $P$ such that for every $x \in \{0,1\}^*$,
    \begin{enumerate}
      \item $x \in L \Longleftrightarrow P(x, 1) = 1$,
      \item $x \notin L \Longleftrightarrow P(x, 1) = 0$.
    \end{enumerate}
\item \textbf{BQRAMP} stands for Bounded-error Quantum Random Access Machine Polynomial-time. More precisely, a language $L \subseteq \{0,1\}^*$ is said to be in \textbf{BQRAMP}, if there is a polynomial-time QRAM $P$ such that for every $x \in \{0,1\}^*$,
    \begin{enumerate}
      \item $x \in L \Longleftrightarrow P(x, 1) \geq \frac 2 3$,
      \item $x \notin L \Longleftrightarrow P(x, 1) \leq \frac 1 3$.
    \end{enumerate}\end{itemize}

    \subsection{Address Shifting and Address-safe QRAMs} \label{sec:qram-address-shifting}

    In order to describe algorithms more conveniently, we introduce the technique of address shifting, which enables us to flexibly deal with free variables.

    \begin{lemma} [Address Shifting] \label{lemma-qram-address-shifting}
        Let $T: \mathbb{N} \to \mathbb{N}$. For every $T(n)$-time QRAM $P$ and every integer $k > 0$, there is an $O(T(n))$-time QRAM $P'$ such that
        \begin{enumerate}
          \item $P(x, y) = P'(x, y)$ for every $x, y \in \Sigma^*$.
          \item $P'$ never accesses to the classical registers $X_1, X_2, \dots, X_k$.
        \end{enumerate}
    \end{lemma}


    With the help of Lemma \ref{lemma-qram-address-shifting}, we show that address-safety can be enforced for QRAMs.

    \begin{lemma} \label{lemma-qram-address}
        Let $T: \mathbb{N} \to \mathbb{N}$. For every $T(n)$-time QRAM $P$, there is a $O(T(n))$-time QRAM $P'$ such that
        \begin{enumerate}
          \item $P(x, y) = P'(x, y)$ for every $x, y \in \Sigma^*$.
          \item $P'$ never accesses to an invalid address in the execution. We call such a QRAM $P'$ address-safe.
        \end{enumerate}
    \end{lemma}

    The proofs of Lemma \ref{lemma-qram-address-shifting} and Lemma \ref{lemma-qram-address} are given in Appendix \ref{sec:proofs-of-qram-address}.

    \subsection{Measurement-postponed QRAMs} \label{sec:qram-measurement-postponed}

    In this subsection, we generalise the technique of postponing measurements, which has been widely used in quantum computing, to QRAMs.

    \begin{definition}
        A QRAM is said to be measurement-postponed, if no further operations are performed on the quantum registers once they are measured.
    \end{definition}

    \begin{theorem} \label{lemma-qram-measurement}
        Let $T: \mathbb{N} \to \mathbb{N}$. For every $T(n)$-time QRAM $P$, there is a $O(T(n))$-time QRAM $P'$ such that
        \begin{enumerate}
          \item $P(x, y) = P'(x, y)$ for every $x, y \in \Sigma^*$.
          \item $P'$ is measurement-postponed.
        \end{enumerate}
    \end{theorem}

    \begin{proof}
        Suppose $P$ consists of $L$ instructions $P_0, P_1, \dots, P_{L-1}$. By Lemma \ref{lemma-qram-address}, we can assume that $P$ is address-safe.
 In order to postpone measurements, we recall the technique shown in Figure \ref{fig1}.
        \begin{figure}[!htp]
        \caption{Quantum circuits for postponing measurements}
        \label{fig1} \centering
        \[
        \Qcircuit @C=1em @R=.7em {
            & \qw & \meter & \gate{U} & \qw \\
        }
        =
        \Qcircuit @C=1em @R=.7em {
            &        & & \qw & \ctrl{1} & \gate{U} & \qw \\
            & \ket 0 & & \qw & \targ & \meter & \qw \\
        }
        \]
        \end{figure}
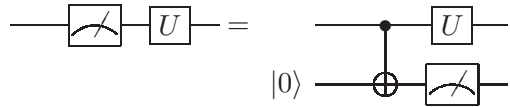
        Inspired by this, we use a special variable $\mathit{mea}$ to count how many measurements are performed. We split quantum registers into two disjoint parts, one of which is of even addresses and the other is of odd addresses. The quantum registers of even addresses are used for quantum gate operations while the rest (i.e. those of odd addresses) are used only for measurements. For a better understanding, we first give an intuition behind our construction. We use two functions $f(x) = 2x$ and $g(x) = 2x+1$. For a quantum gate (CNOT, Hadamard and $\pi/8$ gates), say $\text{CNOT}[Q_a, Q_b]$, we can perform $\text{CNOT}[Q_{f(a)}, Q_{f(b)}]$. For a measurement, say $M[Q_a]$, let $\mathit{mea}$ be the current number of measurements that have been performed, then we can perform $\text{CNOT}[Q_{f(a)}, Q_{g(\mathit{mea})}]$ and then measure $Q_{g(\mathit{mea})}$, i.e. $M[Q_{g(\mathit{mea})}]$.

        In order to precisely describe how to postpone measurements, we first list the lengths needed for all instructions in Table \ref{tab-postpone-measurement}. For $0 \leq l < L$, we write $\mathit{length}(l)$ for the length needed for postponing measurements according to Table \ref{tab-postpone-measurement}. To label the instructions in $P'$, we define:
        \[
            \mathit{label}(l) = \sum_{i=0}^{l-1} \mathit{length}(i)
        \]
        for $0 \leq l \leq L$. Especially, the length of $P'$ is defined to be $L' = \mathit{label}(L)$.

        \begin{table}[!htp]
            \centering
            \caption{Lengths of QRAM instructions for postponing measurements}
            \begin{tabular}{llc}
            \hline
            Type & Instruction & length \\
            \hline
            Classical & $X_i \gets C$, $C$ any integer & $1$ \\
            Classical & $X_i \gets X_j+X_k$ & $1$ \\
            Classical & $X_i \gets X_j-X_k$ & $1$ \\
            Classical & $X_i \gets X_{X_j}$ & $1$ \\
            Classical & $X_{X_i} \gets X_j$ & $1$ \\
            Classical & TRA $m$ if $X_j > 0$ & $1$ \\
            Classical & READ $X_i$ & $1$ \\
            Classical & WRITE $X_i$ & $1$ \\
            \hline
            Quantum & $\text{CNOT}[Q_{X_i}, Q_{X_j}]$ & $7$ \\
            Quantum & $H[Q_{X_i}]$ & $4$ \\
            Quantum & $T[Q_{X_i}]$ & $4$ \\
            \hline
            Measurement & $X_i \gets M[Q_{X_j}]$ & $12$ \\
            \hline
            \end{tabular}
            \label{tab-postpone-measurement}
        \end{table}

        Now we are ready to describe our construction of $P'$. For every $0 \leq l < L$, we convert $P_l$ to one or more instructions in $P'$. Note that we need three extra variables $\mathit{mea}, a$ and $b$.

        {\vskip 3pt}

        \textbf{Case 1}. If $P_l$ is $\text{CNOT}[Q_{X_i}, Q_{X_j}]$, the instructions for $P'$ are as follows.
        \begin{align*}
          \mathit{label}(l): & a \gets 0 \\
          & a \gets a+X_i \\
          & a \gets a+a \\
          & b \gets 0 \\
          & b \gets b+X_j \\
          & b \gets b+b \\
          & \text{CNOT}[Q_a, Q_b] \\
        \end{align*}

        \textbf{Case 2}. If $P_l$ is $H[Q_{X_i}]$ (resp. $T[X_{X_i}]$), the instructions for $P'$ are as follows.
        \begin{align*}
          \mathit{label}(l): & a \gets 0 \\
          & a \gets a+X_i \\
          & a \gets a+a \\
          & H[Q_a] (\text{resp. } T[Q_a]) \\
        \end{align*}

        \textbf{Case 3}. If $P_l$ is $X_i \gets M[Q_{X_j}]$, the instructions for $P'$ are as follows.
        \begin{align*}
          \mathit{label}(l): & a \gets 1 \\
          & \mathit{mea} \gets \mathit{mea}+a \\
          & b \gets 0 \\
          & b \gets b+\mathit{mea} \\
          & b \gets b+b \\
          & a \gets 1 \\
          & b \gets b+a \\
          & a \gets 0 \\
          & a \gets a+X_j \\
          & a \gets a+a \\
          & \text{CNOT}[Q_a, Q_b] \\
          & X_i \gets M[Q_b] \\
        \end{align*}

        \textbf{Case 4}. If $P_l$ is jumping, i.e. TRA $m$ if $X_j > 0$, we use a single modified instruction
        \begin{align*}
          \mathit{label}(l): & \text{TRA } m' \text{ if } X_{j+\delta} > 0
        \end{align*}
        with $m' = \mathit{label}(m)$.

        {\vskip 3pt}

        \textbf{Case 5}. For other cases, use the same instructions as in $P$.

        {\vskip 3pt}

        It can be seen that the constructed QRAM $P'$ simulates QRAM $P$ with a constant factor slowdown. Since $\mathit{mea}$ always increases, no quantum register will be measured more than once. To the end, according to Lemma \ref{lemma-qram-address-shifting}, the variables $\mathit{mea}, a$ and $b$ are involved by shifting the address to the right by $\delta = 4$.
    \end{proof}

    \begin{remark}
        Since the CNOT gate is considered to be more expensive than a single-qubit gate in the current generation of quantum hardware (for more details about the cost of quantum gates, see \cite{Lee06}), one may be concerned about how many CNOT gates are used in the measurement-postponed simulation. In the proof of Theorem \ref{lemma-qram-measurement}, we see that every time a measurement is performed in the simulated QRAM, an additional qubit is introduced and an extra CNOT gate is performed in the simulation. Therefore, if the simulated QRAM performs measurements for $m$ times, the measurement-postponed simulation will need $m$ additional qubits and will require $m$ extra CNOT gates.
    \end{remark}

    \section{Computations of QRASPs} \label{sec:def-qrasp}

    In this section, we define the computations of QRASPs in terms of operational and denotational semantics, in parallel with what we did for QRAMs in the last section.

    \subsection{Operational semantics} \label{sec:qrasp-operational}

    A configuration of a QRASP $P$ is a tuple $(\xi, \zeta, \mu, \ket\psi, x, y)$, where:
    \begin{enumerate}
      \item $\xi \in \mathbb{N} \cup \{\downarrow\}$ denotes the current IC, with $\downarrow$ indicating the end of execution.
      \item $\zeta \in \mathbb{Z}$ denotes the current AC.
      \item $\mu: \mathbb{N} \to \mathbb{Z}$ is the description of all contents of classical registers;
      \item $\ket\psi \in \mathcal{H}= \bigotimes_{i=0}^\infty \mathcal{H}_i$ is the state of all quantum registers, with $\mathcal{H}_i = \operatorname{span} \{ \ket 0_i, \ket 1_i \}$;
      \item $x \in \mathbb{Z}^\omega$ is a sequence of integers to read on the input tape;
      \item $y \in \mathbb{Z}^*$ is a sequence of printed integers on the output tape.
    \end{enumerate}

We write $\mathcal{C} = (\mathbb{N} \cup \{\downarrow\}) \times \mathbb{Z} \times \mathbb{Z}^\mathbb{N} \times \mathcal{H} \times \mathbb{Z}^\omega \times \mathbb{Z}^*$ for the set of all configurations. A configuration $c = (\xi, \zeta, \mu, \ket\psi, x, y) \in \mathcal{C}$ is a terminal configuration if $\xi =\ \downarrow$. Let $\mathcal{C}_f \subseteq \mathcal{C}$ denote the set of all terminal configurations.

    Similar to the case of QRAMs, the execution transition of a QRASP $P$ is a function $\rightarrow: \mathcal{C} \times \mathcal{C} \to [0, 1] \times \mathbb{N}$ defined by the following rules:
    \begin{enumerate}
      \item When $\mu(\xi)$ is beyond $[1, 11]$,
      \[
        (\xi, \zeta, \mu, \ket\psi, x, y) \xrightarrow[l(\xi)+l(\mu(\xi))]{} (\downarrow, \zeta, \mu, \ket\psi, x, y).
      \]
      \item When $\mu(\xi) = 1$,
      \[
        (\xi, \zeta, \mu, \ket\psi, x, y) \xrightarrow[l(\xi)+l(\mu(\xi+1))]{} (\xi+2, \mu(\xi+1), \mu, \ket\psi, x, y).
      \]
      \item When $\mu(\xi) = 2$, if $\mu(\xi+1) \geq 0$, then
      \[
        (\xi, \zeta, \mu, \ket\psi, x, y) \xrightarrow[l(\xi)+l(\mu(\xi+1))+l(\zeta)+l(\mu(\mu(\xi+1)))]{} (\xi+2, \zeta+\mu(\mu(\xi+1)), \mu, \ket\psi, x, y);
      \]
      otherwise,
      \[
        (\xi, \zeta, \mu, \ket\psi, x, y) \xrightarrow[l(\xi)+l(\mu(\xi+1))]{} (\downarrow, \zeta, \mu, \ket\psi, x, y).
      \]
      \item When $\mu(\xi) = 3$, if $\mu(\xi+1) \geq 0$, then
      \[
        (\xi, \zeta, \mu, \ket\psi, x, y) \xrightarrow[l(\xi)+l(\mu(\xi+1))+l(\zeta)+l(\mu(\mu(\xi+1)))]{} (\xi+2, \zeta-\mu(\mu(\xi+1)), \mu, \ket\psi, x, y);
      \]
      otherwise,
      \[
        (\xi, \zeta, \mu, \ket\psi, x, y) \xrightarrow[l(\xi)+l(\mu(\xi+1))]{} (\downarrow, \zeta, \mu, \ket\psi, x, y).
      \]
      \item When $\mu(\xi) = 4$, if $\mu(\xi+1) \geq 0$, then
      \[
        (\xi, \zeta, \mu, \ket\psi, x, y) \xrightarrow[l(\xi)+l(\mu(\xi+1))+l(\zeta)]{} (\xi+2, \zeta, \mu_{\mu(\xi+1)}^\zeta, \ket\psi, x, y);
      \]
      otherwise,
      \[
        (\xi, \zeta, \mu, \ket\psi, x, y) \xrightarrow[l(\xi)+l(\mu(\xi+1))]{} (\downarrow, \zeta, \mu, \ket\psi, x, y).
      \]
      \item When $\mu(\xi) = 5$, if $\zeta > 0$, then:
      \begin{enumerate}
        \item if $\mu(\xi+1) \geq 0$,
          \[
            (\xi, \zeta, \mu, \ket\psi, x, y) \xrightarrow[l(\xi)+l(\mu(\xi+1))+l(\zeta)]{} (\mu(\xi+1), \zeta, \mu, \ket\psi, x, y),
          \]
        \item if $\mu(\xi+1) < 0$,
          \[
            (\xi, \zeta, \mu, \ket\psi, x, y) \xrightarrow[l(\xi)+l(\mu(\xi+1))+l(\zeta)]{} (\downarrow, \zeta, \mu, \ket\psi, x, y);
          \]
      \end{enumerate}
      otherwise,
      \[
        (\xi, \zeta, \mu, \ket\psi, x, y) \xrightarrow[l(\xi)+l(\zeta)]{} (\xi+2, \zeta, \mu, \ket\psi, x, y).
      \]
      \item When $\mu(\xi) = 6$, if $\mu(\xi+1) \geq 0$, then let $a$ be the first integer in $x$, and
      \[
        (\xi, \zeta, \mu, \ket\psi, x, y) \xrightarrow[l(\xi)+l(\mu(\xi+1))+l(a)]{} (\xi+2, \zeta, \mu_{\mu(\xi+1)}^a, \ket\psi, x', y),
      \]
      where $x'$ denotes the string obtained by deleting the first integer from $x$;
      otherwise,
      \[
        (\xi, \zeta, \mu, \ket\psi, x, y) \xrightarrow[l(\xi)+l(\mu(\xi+1))]{} (\downarrow, \zeta, \mu, \ket\psi, x, y).
      \]
      \item When $\mu(\xi) = 7$, if $\mu(\xi+1) \geq 0$, then
      \[
        (\xi, \zeta, \mu, \ket\psi, x, y) \xrightarrow[l(\xi)+l(\mu(\xi+1))+l(\mu(\mu(\xi+1)))]{} (\xi+2, \zeta, \mu, \ket\psi, x, y'),
      \]
      where $y'$ denotes the string obtained by appending $\mu(\mu(\xi+1))$ to $y$;
      otherwise,
      \[
        (\xi, \zeta, \mu, \ket\psi, x, y) \xrightarrow[l(\xi)+l(\mu(\xi+1))]{} (\downarrow, \zeta, \mu, \ket\psi, x, y).
      \]
      \item When $\mu(\xi) = 8$, if $\mu(\xi+1) \geq 0$ and $\mu(\xi+2) \geq 0$, then
      \[
        (\xi, \zeta, \mu, \ket\psi, x, y) \xrightarrow[l(\xi)+l(\mu(\xi+1))+l(\mu(\xi+2))]{} (\xi+3, \zeta, \mu, \text{CNOT}_{\mu(\xi+1), \mu(\xi+2)}\ket\psi, x, y);
      \]
      otherwise,
      \[
        (\xi, \zeta, \mu, \ket\psi, x, y) \xrightarrow[l(\xi)+l(\mu(\xi+1))]{} (\downarrow, \zeta, \mu, \ket\psi, x, y).
      \]
      \item When $\mu(\xi) = 9$, if $\mu(\xi+1) \geq 0$, then
      \[
        (\xi, \zeta, \mu, \ket\psi, x, y) \xrightarrow[l(\xi)+l(\mu(\xi+1))]{} (\xi+2, \zeta, \mu, H_{\mu(\xi+1)}\ket\psi, x, y);
      \]
      otherwise,
      \[
        (\xi, \zeta, \mu, \ket\psi, x, y) \xrightarrow[l(\xi)+l(\mu(\xi+1))]{} (\downarrow, \zeta, \mu, \ket\psi, x, y).
      \]
      \item When $\mu(\xi) = 10$, if $\mu(\xi+1) \geq 0$, then
      \[
        (\xi, \zeta, \mu, \ket\psi, x, y) \xrightarrow[l(\xi)+l(\mu(\xi+1))]{} (\xi+2, \zeta, \mu, T_{\mu(\xi+1)}\ket\psi, x, y);
      \]
      otherwise,
      \[
        (\xi, \zeta, \mu, \ket\psi, x, y) \xrightarrow[l(\xi)+l(\mu(\xi+1))]{} (\downarrow, \zeta, \mu, \ket\psi, x, y).
      \]
      \item When $\mu(\xi) = 11$, if $\mu(\xi+1) \geq 0$, then
      \[
        (\xi, \zeta, \mu, \ket\psi, x, y) \xrightarrow[l(\xi)+l(\mu(\xi+1))]{\Abs{P_{\mu(\xi+1)}\ket\psi}^2} (\xi+2, 0, \mu, \frac{P_{\mu(\xi+1)}\ket\psi}{\Abs{P_{\mu(\xi+1)}\ket\psi}^2}, x, y),
      \]
      \[
        (\xi, \zeta, \mu, \ket\psi, x, y) \xrightarrow[l(\xi)+l(\mu(\xi+1))]{1-\Abs{P_{\mu(\xi+1)}\ket\psi}^2} (\xi+2, 1, \mu, \frac{(I-P_{\mu(\xi+1)})\ket\psi}{1-\Abs{P_{\mu(\xi+1)}\ket\psi}^2}, x, y);
      \]
      otherwise,
      \[
        (\xi, \zeta, \mu, \ket\psi, x, y) \xrightarrow[l(\xi)+l(\mu(\xi+1))]{} (\downarrow, \zeta, \mu, \ket\psi, x, y).
      \]
    \end{enumerate}

    \subsection{Denotational semantics} \label{sec:qrasp-denotational}

The notions of execution path, semantic function and  worst case running time for a QRASP can be defined in the same way as those for a QRAM in Subsection \ref{sec:qram-denotational}. Moreover, it is easy to show that Lemma \ref{lem-worst-time} holds for QRASPs too.

    \subsection{QRASP computations} \label{sec:qrasp-computing}

    A QRASP is a sequence $P = P_0, P_1, \dots, P_{L-1}$ of integers with $L = \abs{P}$ the length of $P$, which is initially stored in the classical registers. More precisely, the initial contents of classical registers are described by
    \[
        \mu_0(\xi) = \begin{cases}
            P_\xi & 0 \leq \xi < L, \\
            0 & \text{otherwise}.
        \end{cases}
    \]
    Moreover, initially IC and AC are set to $0$ and all quantum registers are $\ket 0$, i.e.
    \[
        \ket{\psi_0} = \bigotimes_{i=0}^\infty \ket 0_i.
    \]
    Similar to the case of QRAMs, the computation of a QRASP $P$ is defined on finite strings over a finite alphabet $\Sigma$. Suppose the input string is $x \in \Sigma^*$, then the computation starts from the configuration
    \[
    c_x = (0, 0, \mu_0, \ket{\psi_0}, \mathit{in}(x), \epsilon).
    \]
    The computational result $P: \Sigma^* \times \Sigma^* \to [0, 1]$ is then defined by
    \[
        P(x, y) = \sum_{c \in \mathcal{C}_f: \mathit{out}(c.y) = y} \llbracket P \rrbracket (c_x)(c),
    \]
    and the worst case running time $T_P: \Sigma^* \to \mathbb{N} \cup \{\infty\}$ is defined by
    \[
        T_P(x) = \tau_P(c_x).
    \]
    Furthermore, let $T: \mathbb{N} \to \mathbb{N}$. Then $P$ is said to be a $T(n)$-time QRASP, if for every $x \in \Sigma^*$, $T_P(x) \leq T(\abs{x})$.

    \section{Comparison of QRAMs and QRASPs} \label{sec:sim-qram-qrasp}

With the precise definitions of QRAMs and QRASPs given in the previous sections, we are now ready to prove Theorem \ref{thm-qram-qrasp}. Subsection \ref{sec:qrasp-by-qram} shows how QRAMs can simulate QRASPs, and Subsection \ref{sec:qram-by-qrasp} shows how QRASPs can simulate QRAMs.

    \subsection{QRAMs simulate QRASPs} \label{sec:qrasp-by-qram}

    \subsubsection{Simulation construction}

    Let QRASP $P$ be given as a sequence $P_0, P_1, P_2, \dots, P_{L-1}$ of integers.
    The idea of the simulation is to hardcode $P$ into the classical registers of a QRAM $P'$, and then simulate the execution of $P$. The details of the simulation are presented in Algorithm \ref{algo1}.
   For readability, we only present $P'$ as pseudo-codes. (The translation from the pseudo-codes to QRAM instructions can be done in a familiar way, and the details are provided in Appendix \ref{qram-to-qrasp} for completeness).
 \begin{algorithm}[!htp]
        \caption{QRAM pseudo-code for simulating QRASP.}
        \label{algo1}
        \begin{algorithmic}[1]
        \Require The input of the QRASP to be simulated.
        \Ensure The intended output of the QRASP to be simulated.

        \State integer array $\mathit{memory}$;
        \State integer $\text{IC}, \text{AC}, \mathit{flag}, \mathit{op}, j, k$;

        \State $\mathit{memory}[0] \gets P_0; \mathit{memory}[1] \gets P_1; \dots \mathit{memory}[L-1] \gets P_{L-1};$
        \While {$\mathit{flag} = 0$}
            \State $\mathit{op} \gets \mathit{memory}[\text{IC}]$;
            \If {$\mathit{op} = 1$}
                \State $j \gets \mathit{memory}[\text{IC}+1]$; $\text{AC} \gets j$; $\text{IC} \gets \text{IC}+2$;
            \ElsIf {$\mathit{op} = 2$}
                \State $j \gets \mathit{memory}[\text{IC}+1]$; $\text{AC} \gets \text{AC}+\mathit{memory}[j]$; $\text{IC} \gets \text{IC}+2$;
            \ElsIf {$\mathit{op} = 3$}
                \State $j \gets \mathit{memory}[\text{IC}+1]$; $\text{AC} \gets \text{AC}-\mathit{memory}[j]$; $\text{IC} \gets \text{IC}+2$;
            \ElsIf {$\mathit{op} = 4$}
                \State $j \gets \mathit{memory}[\text{IC}+1]$; $\mathit{memory}[j] \gets \text{AC}$; $\text{IC} \gets \text{IC}+2$;
            \ElsIf {$\mathit{op} = 5$}
                \If {$\text{AC} > 0$}
                    \State $j \gets \mathit{memory}[\text{IC}+1]$; $\text{IC} \gets j$;
                \Else
                    \State $\text{IC} \gets \text{IC}+2$;
                \EndIf
            \ElsIf {$\mathit{op} = 6$}
                \State $j \gets \mathit{memory}[\text{IC}+1]$; READ $\mathit{memory}[j]$; $\text{IC} \gets \text{IC}+2$;
            \ElsIf {$\mathit{op} = 7$}
                \State $j \gets \mathit{memory}[\text{IC}+1]$; WRITE $\mathit{memory}[j]$; $\text{IC} \gets \text{IC}+2$;
            \ElsIf {$\mathit{op} = 8$}
                \State $j \gets \mathit{memory}[\text{IC}+1]$; $k \gets \mathit{memory}[\text{IC}+2]$; CNOT$[Q_j, Q_k]$; $\text{IC} \gets \text{IC}+3$;
            \ElsIf {$\mathit{op} = 9$}
                \State $j \gets \mathit{memory}[\text{IC}+1]$; $H[Q_j]$; $\text{IC} \gets \text{IC}+2$;
            \ElsIf {$\mathit{op} = 10$}
                \State $j \gets \mathit{memory}[\text{IC}+1]$; $T[Q_j]$; $\text{IC} \gets \text{IC}+2$;
            \ElsIf {$\mathit{op} = 11$}
                \State $j \gets \mathit{memory}[\text{IC}+1]$; $\text{AC} \gets M[Q_j]$; $\text{IC} \gets \text{IC}+2$;
            \Else
                \State $\mathit{flag} \gets 1$;
            \EndIf
        \EndWhile

        \end{algorithmic}
    \end{algorithm}

    \subsubsection{Correctness proof}\label{proof-qram}

    The remaining part of this subsection is devoted to prove correctness of Algorithm \ref{algo1}; that is, for any QRASP $P$, the QRAM $P'$ constructed by Algorithm \ref{algo1} can simulate $P$ with a suitable time complexity.

    Let $(\xi, \zeta, \mu, \ket\psi, x, y)$ be a configuration of $P$ and $(\xi', \mu', \ket{\psi'}, x', y')$ a configuration of $P'$. We use $\mu'(\mathit{var})$ to denote the value of variable $\mathit{var}$ stored in $P'$ according to $\mu'$.

 \begin{definition}   We say that a QRAM configuration $(\xi', \mu', \ket{\psi'}, x', y')$ agrees with a QRASP configuration $(\xi, \zeta, \mu, \ket\psi, x, y)$, written $(\xi', \mu', \ket{\psi'}, x', y') \models (\xi, \zeta, \mu, \ket\psi, x, y)$, if
    \begin{enumerate}
      \item $\mu'(\text{AC}) = \zeta$, ($\mu'(\mathit{flag}) = 1$ or $\xi' = \downarrow$), $\mu'(\mathit{memory}[j]) = \mu(j)$ for every $j \in \mathbb{N}$, $\ket{\psi'} = \ket\psi$, $x'=x$ and $y' = y$ in the case $\xi = \downarrow$; or
      \item $\mu'(\text{AC}) = \zeta$, $\mu'(\text{IC}) = \xi$, $\mu'(\mathit{memory}[j]) = \mu(j)$ for every $j \in \mathbb{N}$, $\ket{\psi'} = \ket\psi$, $x'=x$ and $y' = y$ in the case $\xi \in \mathbb{N}$.
    \end{enumerate}\end{definition}

    \begin{lemma} \label{lemma-qrasp-by-qram-unique-config}
        For every $c' \in \mathcal{C}'$, there is a unique $c \in \mathcal{C}$ such that $c' \models c$.
    \end{lemma}
    \begin{proof} We only need to observe: (1) if $c' \models c_1$ and $c' \models c_2$, then $c_1 = c_2$; and (2) for every $c'$, there is a $c$ such that $c' \models c$.
    \end{proof}

    Let $\mathcal{C}$ and $\mathcal{C}'$ be the set of configurations of $P$ and $P'$, respectively, and let $c_0 \in \mathcal{C}$ and $c_0' \in \mathcal{C}'$ be their initial configurations.
    We write $\mathcal{C}'_{\text{L}5} = \{ c' \in \mathcal{C}': c'.\xi = 5 \}$ for the set of configurations of $P'$ that reaches Line $5$ in Algorithm \ref{algo1} (Here, we use the line number to indicate the current IC).

    \begin{lemma} \label{lemma-qrasp-by-qram-simulate-single-step}
        Let $c' \in \mathcal{C}'_{\text{L}5}$ and $c, d \in \mathcal{C}$. If $c' \models c$, and $c \xrightarrow[T]{p} d$,
        then there is a $d' \in \mathcal{C}'_{\text{L}5}$ such that $d' \models d$ and $c' \xrightarrow[\Theta(T)]{p} d'$.
    \end{lemma}
    \begin{proof}
        Direct from the operational semantics.
    \end{proof}

    We use $\mathcal{P}$ and $\mathcal{P}'$ to denote the sets of all possible execution paths of $P$ and $P'$, respectively.
    Let $\pi' \in \mathcal{P}'_f(c_0')$ be a path of length $\abs{\pi'} = k$:
    \[
        \pi': c_0' \xrightarrow[T_1']{p_1'} c_1' \xrightarrow[T_2']{p_2'} \dots \xrightarrow[T_{k-1}']{p_{k-1}'} c_{k-1}' \xrightarrow[T_k']{p_k'} c_k',
    \]
    and let $0 < i_0 < i_1 < \dots < i_{m-1} < k$ be all indices that $c'_{(j)} = c'_{i_j} \in \mathcal{C}'_{\text{L}5}$ for $0 \leq j < m$. Then it can be written as
    \[
        \pi': c_0' \xrightarrow[T'_{(0)}]{p'_{(0)}}^* c'_{(0)} \xrightarrow[T'_{(1)}]{p'_{(1)}}^* c'_{(1)} \xrightarrow[T'_{(2)}]{p'_{(2)}}^* \dots \xrightarrow[T'_{(m-1)}]{p'_{(m-1)}}^* c'_{(m-1)} \xrightarrow[T'_{(m)}]{p'_{(m)}}^* c'_{(m)} = c'_k,
    \]
    where
    \[
    p'_{(j)} = \prod_{l=i_{j-1}+1}^{i_j} p'_{l},
    \]
    and
    \[
    T'_{(j)} = \sum_{l=i_{j-1}+1}^{i_j} T'_{l}
    \]
    for $0 \leq j \leq m$, and $i_{-1} = 0, i_m = k$. We define $\Abs{\pi'} = m$.

    \begin{lemma}
        $c'_{(0)} \models c_0$ and $c'_0 \xrightarrow[O(1)]{1}^* c'_{(0)}$.
    \end{lemma}

     \begin{proof} Obvious.\end{proof}

    \begin{definition}
        Let $\pi' \in \mathcal{P}'_f(c'_{(0)})$ and $\pi \in \mathcal{P}_f(c_0)$.
        Then we say that $\pi'$ agrees with $\pi$, denoted $\pi' \models \pi$, if
        \begin{enumerate}
          \item $\Abs{\pi'} = \abs{\pi}$.
          \item $c'_{(j)} \models c_j$ for $0 \leq j \leq \Abs{\pi'}$.
          \item $p'_{(j)} = p_j$ and $T'_{(j)} = \Theta(T_j)$ for $1 \leq j \leq \Abs{\pi'}$.
        \end{enumerate}
    \end{definition}

    \begin{lemma} \label{lemma-qrasp-by-qram-time-bound}
        For every $\pi' \in \mathcal{P}'_f(c'_{(0)})$, there is a unique $\pi \in \mathcal{P}_f(c_0)$ such that $\pi' \models \pi$.
    \end{lemma}

    \begin{proof} Obvious.\end{proof}

    Lemma \ref{lemma-qrasp-by-qram-time-bound} implies that $P'$ is time bounded by $O(T(n))$. Let $T': \mathbb{N} \to \mathbb{N}$ be the worst case running time of $P'$.

    \begin{lemma}
        For every $\pi \in \mathcal{P}_f(c_0)$,
        there is a unique $\pi' \in \mathcal{P}'_f(c'_{(0)})$ such that $\pi' \models \pi$.
        We use $h: \mathcal{P}(c_0) \to \mathcal{P}'(c'_{(0)})$ to denote this bijection.
    \end{lemma}
    \begin{proof}
        (\textbf{Existence}) Directly by Lemma \ref{lemma-qrasp-by-qram-simulate-single-step}.

        (\textbf{Uniqueness}) For every $\pi \in \mathcal{P}(c_0)$, we choose an arbitrary $\pi' \in \mathcal{P}'(c'_{(0)})$ such that $\pi' \models \pi$ and write $h(\pi) = \pi'$. We note that
        \[
            1 = \sum_{\pi \in \mathcal{P}(c_0)} \pi.p = \sum_{\pi \in \mathcal{P}(c_0)} h(\pi).p \leq \sum_{\pi' \in \mathcal{P}'(c'_{(0)})} \pi'.p = 1,
        \]
        the uniqueness of $h(\pi)$ follows immediately.
    \end{proof}

  Finally, we are ready to show that $P'$ actually simulates $P$. Let $x \in \Sigma^*$ be the input string and the initial configuration of $P$ be $c_0 = (0, 0, \mu_0, \ket{\psi_0}, \mathit{in}(x), \epsilon)$. Since $P$ is a $T(n)$-time QRASP, $\abs{\pi} \leq T(\abs{x})$ is finite for every $\pi \in \mathcal{P}(c_0)$. Now that each transition leads to at most two branches, $\abs{\mathcal{P}(c_0)} \leq 2^{T(\abs{x})}$ must be finite too. Thus, for every $y \in \Sigma^*$, we have:

    \begin{align*}
        P(x, y)
        & = \sum_{c \in \mathcal{C}_f: \mathit{out}(c.y) = y} \llbracket P \rrbracket (c_0) (c) \\
        & = \sum_{c \in \mathcal{C}_f: \mathit{out}(c.y) = y} \llbracket P \rrbracket^{T(\abs{x})} (c_0) (c) \\
        & = \sum_{\pi \in \mathcal{P}^{T(\abs{x})}(c_0): \mathit{out}(\pi.c_{\abs{\pi}}.y) = y} \pi.p \\
        & = \sum_{\pi \in \mathcal{P}^{T(\abs{x})}(c_0): \mathit{out}(h(\pi).c_{\abs{f(\pi)}}.y) = y} h(\pi).p \\
        & = \sum_{\pi' \in {\mathcal{P}'}^{T'(\abs{x})}(c'_{(0)}): \mathit{out}(\pi'.c_{\abs{\pi'}}.y) = y} \pi'.p \\
        & = \sum_{\pi' \in {\mathcal{P}'}^{T'(\abs{x})}(c'_{0}): \mathit{out}(\pi'.c_{\abs{\pi'}}.y) = y} \pi'.p \\
        & = \sum_{c' \in \mathcal{C}'_f: \mathit{out}(c'.y) = y} \llbracket P' \rrbracket^{T'(\abs{x})} (c'_0) (c) \\
        & = \sum_{c' \in \mathcal{C}'_f: \mathit{out}(c'.y) = y} \llbracket P' \rrbracket (c'_0) (c) \\
        & = P'(x, y).
    \end{align*}

    \subsection{QRASPs simulate QRAMs} \label{sec:qram-by-qrasp}

    \subsubsection{Simulation construction}

    Let QRAM $P$ be a sequence $P_0, P_1, \dots, P_{L-1}$ of QRAM instructions. By Lemma \ref{lemma-qram-address}, we may assume that $P$ is address-safe without any loss of generality.
    The QRASP $P'$ that simulates QRAM $P$ is defined as follows. Let $\delta$ be an integer greater than the length of $P'$, i.e. $\delta > \abs{P'}$. It will be shown later that $\delta = 20L$ (a finite number) is enough. Define the simulating length $\mathit{simulate}(P_i)$ of $P_i$ being the length of QRASP code intended to simulate the QRAM instruction $P_i$. The intended value for $\mathit{simulate}(P_i)$ are shown in Table \ref{tab3} according to the instruction type of $P_i$.
    \begin{table}[!htp]
        \centering
        \caption{Simulating length of QRAM instructions by QRASP}
        \begin{tabular}{llc}
        \hline
        Type & Instruction & Simulating length \\
        \hline
        Classical & $X_i \gets C$, $C$ any integer & $4$ \\
        Classical & $X_i \gets X_j+X_k$ & $8$ \\
        Classical & $X_i \gets X_j-X_k$ & $8$ \\
        Classical & $X_i \gets X_{X_j}$ & $12$ \\
        Classical & $X_{X_i} \gets X_j$ & $12$ \\
        Classical & TRA $m$ if $X_j > 0$ & $6$ \\
        Classical & READ $X_i$ & $2$ \\
        Classical & WRITE $X_i$ & $2$ \\
        \hline
        Quantum & $\text{CNOT}[Q_{X_i}, Q_{X_j}]$ & $15$ \\
        Quantum & $H[Q_{X_i}]$ & $8$ \\
        Quantum & $T[Q_{X_i}]$ & $8$ \\
        \hline
        Measurement & $X_i \gets M[Q_{X_j}]$ & $10$ \\
        \hline
        \end{tabular}
        \label{tab3}
    \end{table}
    In order to deal with the jump instruction ``TRA $m$ if $X_j > 0$'', $\mathit{label}(m)$ is needed, which is defined to be the jump address in QRASP corresponding to the jump address $m$ in QRAM. More precisely,
    \[
        \mathit{label}(m) = \sum_{i=0}^{m-1} \mathit{simulate}(P_i).
    \]
The length of our QRASP $P'$ is designed to be $L' = \abs{P'} = \mathit{label}(L)$. For every $0 \leq i < L$, the instruction $P_i$ is interpreted into QRASP instructions as $\mathit{simulate}(P_i)$ integers starting from $\mathit{label}(i)$. In other words, the QRASP instructions $P'_{\mathit{label}(i)}, P'_{\mathit{label}(i)+1}, \dots, P'_{\mathit{label}(i+1)-1}$ are corresponding to QRAM instruction $P_i$.

Now for $0 \leq l < L$, we present the QRASP code for simulating $P_l$. Here we only display those for quantum instructions (The simulations of classical instructions are standard and thus omitted here; they are provided in Appendix \ref{qrasp-to-qram} for completeness). For readability, the QRASP codes are written by means of QRASP mnemonics.
    \begin{enumerate}
       \item $P_l$ is of the form CNOT$[Q_{X_i}, Q_{X_j}]$. The QRASP code is
          \begin{align*}
              \mathit{label}(l): & \text{LOD}, \delta \\
              & \text{ADD}, i+\delta \\
              & \text{STO}, a+1 \\
              & \text{LOD}, \delta \\
              & \text{ADD}, j+\delta \\
              & \text{STO}, a+2 \\
              a: & \text{CNOT}, 0, 0 \\
          \end{align*}
          Note that $a = \mathit{label}(l)+12$.
      \item $P_l$ is of the form $A[Q_{X_i}]$ with $A=H$ or $T$. The QRASP code is
          \begin{align*}
              \mathit{label}(l): & \text{LOD}, \delta \\
              & \text{ADD}, i+\delta \\
              & \text{STO}, a+1 \\
              a: & \text{A}, 0 \\
          \end{align*}
          Note that $a = \mathit{label}(l)+6$.
      \item $P_l$ is of the form $X_i \gets M[Q_{X_j}]$. The QRASP code is
          \begin{align*}
              \mathit{label}(l): & \text{LOD}, \delta \\
              & \text{ADD}, j+\delta \\
              & \text{STO}, a+1 \\
              a: & \text{MEA}, 0 \\
              & \text{STO}, i+\delta \\
          \end{align*}
          Note that $a = \mathit{label}(l)+6$.
    \end{enumerate}

    \subsubsection{Correctness proof}

    The proof is similar to that given in Subsection \ref{proof-qram}. We put it in Appendix \ref{sec:correct-simulate-qrams-by-qrasps} for completeness.

    \section{Comparison of QRAMs and QTMs} \label{sec:sim-qram-qtm}

    In this section, we prove Theorem \ref{thm-qram-qtm}. Subsection \ref{sec:qram-by-qtm} shows how QTMs can simulate QRAMs, and Subsection \ref{sec:qtm-by-qram} describes how QRAMs can simulate QTMs.

    \subsection{QTMs simulate QRAMs} \label{sec:qram-by-qtm}

 Our simulation of QRAMs by QTMs is carried out in two steps. First, we introduce the notion of Turing machines with a quantum device ($\text{TM}^\text{Q}$s) and prove in Subsection \ref{sec:qram-by-tmq} that every QRAM can be simulated by a measurement-postponed $\text{TM}^\text{Q}$. The main technique here is based on the idea of simulating RAMs by TMs given in \cite{Coo73}. Then we show in Subsection \ref{sec:tmq-by-qtm} that a measurement-postponed $\text{TM}^\text{Q}$ can be simulated by a well-formed and normal form QTM. The main technique here is based on the idea of simulating TMs by RTMs given in \cite{Ber97}.

    \subsubsection{TMs with quantum devices simulate QRAMs} \label{sec:qram-by-tmq}

  Let us first define the notion of TM with a quantum device.

    \begin{definition}
        A TM with a quantum device ($\text{TM}^\text{Q}$) is a $8$-tuple
            $M = (Q, Q_s, Q_t, \Sigma, \delta, \lambda, q_s, q_f),$
        where:
        \begin{enumerate}
          \item $Q$ is a finite set of states;
          \item $Q_s \subseteq Q$ and $Q_t \subseteq Q$ are two disjoint sets of states as interactor for the quantum device;
          \item $\delta: (Q \setminus Q_s) \times \Sigma \to \Sigma \times (Q \setminus Q_t) \times \{ L, R \}$ is the transition function;
          \item $\lambda: (Q_s \times \Sigma^\# \times \mathcal{H}) \times (Q_t \times \mathcal{H}) \to [0, 1]$ is the transition function for quantum device, where
             $\mathcal{H} = \bigotimes_{i=0}^\infty \mathcal{H}_i,$
              and $\mathcal{H}_i = \operatorname{span}\{\ket 0_i, \ket 1_i\}$. It is required that for every $p \in Q_s, \mathcal{T} \in \Sigma^\#, \ket\psi \in \mathcal{H}$,
              \[
                \sum_{q \in Q_t, \ket\phi \in \mathcal{H}} \lambda((p,\mathcal{T},\ket\psi),(q, \ket\phi)) = 1;
              \]
          \item $q_s \in Q \setminus Q_s \setminus Q_t$ is the initial state;
          \item $q_f \in Q \setminus Q_s \setminus Q_t$ is the final state.
        \end{enumerate}
    \end{definition}

    A configuration of $\text{TM}^\text{Q}$ is a tuple $c = (q, \mathcal{T}, \xi, \ket{\psi}) \in Q \times \Sigma^\# \times \mathbb{Z} \times \mathcal{H}$. Let $\mathcal{C} = Q \times \Sigma^\# \times \mathbb{Z} \times \mathcal{H}$ be the set of configurations. The one step execution transition of $\text{TM}^\text{Q}$ is a function $\to: \mathcal{C} \times \mathcal{C} \to [0, 1]$ defined by the following rules: let $c = (p, \mathcal{T}, \xi, \ket{\psi})$,
    \begin{enumerate}
      \item if $p = q_f$, then the execution terminates.
      \item if $p \in Q \setminus Q_s \setminus \{q_f\}$ and $\delta(p, \mathcal{T}(\xi)) = (q, \sigma, d)$, then after one step, the configuration will become $c' = (q, \mathcal{T}_\xi^\sigma, \xi+d, \ket\psi)$, i.e.
          \[
            (p, \mathcal{T}, \xi, \ket{\psi}) \xrightarrow[]{1} (q, \mathcal{T}_\xi^\sigma, \xi+d, \ket{\psi}),
          \]
      \item if $p \in Q_s$ and $\lambda((p, \mathcal{T}, \ket\psi), (q, \ket\phi)) = a$, then after one step, the configuration will become $c' = (q, \mathcal{T}, \xi, \ket{\phi})$ with probability $a$, i.e.
          \[
            (p, \mathcal{T}, \xi, \ket{\psi}) \xrightarrow[]{a} (q, \mathcal{T}, \xi, \ket{\phi}).
          \]
    \end{enumerate}
An execution path is then a non-empty sequence of configurations associated with probabilities:
    \[
        \pi: c_0 \xrightarrow[]{a_1} c_1 \xrightarrow[]{a_2} c_2 \dots c_{n-1} \xrightarrow[]{a_{n}} c_n.
    \]
    The length of $\pi$ is $\abs{\pi} = n$, and the probability of path $\pi$ is
    $
        a = \prod_{i=1}^n a_i.
    $
    In this case, we can simply denote $\pi: c_0 \xrightarrow[]{a}^n c_n$.
    Let $T: \mathbb{N} \to \mathbb{N}$ and
    $
        \ket{\psi_0} = \bigotimes_{i=0}^\infty \ket{0}_i.
    $
    A $\text{TM}^\text{Q}$ is called $T(n)$-time, if for every $x \in \{0,1\}^*$, and every execution path $(q_0, \mathcal{T}_x, 0, \ket{\psi_0}) \xrightarrow[]{a}^t c$, if $a > 0$, then $t \leq T(\abs{x})$.

    Now we can explain the basic idea of our simulation of a QRAM by a $\text{TM}^\text{Q}$.
    The $\text{TM}^\text{Q}$ used to simulate a QRAM needs the following (a finite number of) tracks:
    \begin{enumerate}
      \item ``input'': This track initially contains the input.
      \item ``output'': This track contains the output after the machine halts.
      \item ``creg'': This track contains contents of classical registers. The format is designed to be $\# \text{L} a_1 \text{M} b_1 \text{R} \text{L} a_2 \text{M} b_2 \text{R} \dots \text{L} a_r \text{M} b_r \text{R} \#$, indicating that the content of classical register $a_i$ is $b_i$ for $1 \leq i \leq s$. In particular, if the number $x$ is not found among $a_1, a_2, \dots, a_r$, then the content of classical register is $0$.
      \item ``qcnt'' --- the quantum register counter: This track contains a single non-negative number indicating the number of used quantum registers.
      \item ``qreg'': This track contains a correspondence list from virtual addresses to physical addresses. Similar to track ``creg'', the format is designed to be $\# \text{L} u_1 \text{M} v_1 \text{R} \text{L} u_2 \text{M} v_2 \text{R} \dots \text{L} u_s \text{M} v_s \text{R} \#$, indicating that the virtual address $u_i$ corresponds to physical address $v_i$ for $1 \leq i \leq s$. To avoid too large addresses of quantum registers, the machine re-numbers every used address of quantum registers (virtual address) to a small number (physical address). Each time a quantum register $a$ is accessed in the QRAM, the machine checks whether the virtual address $a$ is collected in ``qreg''. If so, convert $a$ to its corresponding physical address; and if not, increment the quantum register counter and assign $a$ with the value of the current quantum register counter (``qcnt'') as its physical address (and add this assignment to the correspondence list).
      \item ``qdev'' --- a track for interactions to the quantum device: This track is used for quantum operation calls to quantum device. In our case, there are four kinds of quantum operations, i.e. CNOT, Hadamard and $\pi/8$ gates and measurements. The format of this track will be defined shortly.
      \item ``qret'' --- a track for calling back after quantum operations: This track contains a single symbol denoting the returning state, i.e. the next state it should be, after the quantum operations. This track is used as a system stack in the computers.
      \item ``work$i$'' for $i \geq 1$ --- work tracks: The contents of these tracks are ignorant and they will be cleaned to an empty track after use as a hub.
    \end{enumerate}

    The output of $\text{TM}^\text{Q}$ is defined to be the contents in track ``output''. Moreover,
    $\text{TM}^\text{Q}$ $M$ defines a function $M: \{0, 1\}^* \times \{0, 1\}^* \to [0, 1]$ with $M(x, y)$ being the probability that $M$ on input $x$ outputs $y$. In our model, there are only four kinds of quantum operations: CNOT, Hadamard and $\pi/8$ gates, and measurements. So, we use $q_s^C, q_s^H, q_s^T, q_s^M$ to denote their initial states and $q_t^C, q_t^H, q_t^T, q_{t0}^M, q_{t1}^M$ to denote their terminating states, and put $Q_s = \{q_s^C, q_s^H, q_s^T, q_s^M\}$ and $Q_t = \{q_t^C, q_t^H, q_t^T, q_{t0}^M, q_{t1}^M\}$. Furthermore, the format of track ``qdev''  is defined as follows:
    \begin{enumerate}
      \item For a CNOT gate, the machine reads the contents $s$ of track ``qdev''. The string $s \in \{0,1,2\}^*$ is assumed to consist of one $1$, one $2$ and some $0$s. Let $1$ and $2$ be the $a$-th and the $b$-th elements of $s$ (0-indexed), respectively, and let $\ket\psi \in \mathcal{H}$ be the quantum state before the application of the gate. When the gate is applied, the state of $\text{TM}^\text{Q}$ is changed from $q_s^C$ to $q_t^C$ with the quantum state becoming $\text{CNOT}[a, b] \ket\psi$; that is, the $a$-th qubit acts as the control qubit and the $b$-th qubit as the target qubit (0-indexed).
      \item For a Hadamard (resp. $\pi/8$) gate, the machine reads the contents $s$ of track ``qdev''. The string $s \in \{0,1\}^*$ is assumed to consist of one $1$ and several $0$s. Let $1$ be the $a$-th element of $s$ (0-indexed), and let $\ket\psi \in \mathcal{H}$ be the quantum state before application of the gate. When the gate operation is applied, the state of $\text{TM}^\text{Q}$ is changed from $q_s^H$ (res. $q_s^T$) to $q_t^H$ (resp. $q_t^T$); that is, the gate is performed on the $a$-th qubit (0-indexed) with the quantum state becoming $H[a] \ket\psi$ (res. $T[a] \ket\psi$).

      \item For a measurement, the machine reads the contents $s$ of track ``qdev''. The string $s \in \{0,1\}^*$ is assumed to consist of one $1$ and several $0$s. Let $1$ be the $a$-th element of $s$ (0-indexed), and let $\ket\psi \in \mathcal{H}$ be the quantum state before the measurement. When the measurement is performed, the state of $\text{TM}^\text{Q}$ is changed from $q_s^M$ to $q_{t0}^M$ such that the quantum state becomes $\ket{\phi_0}$ with probability $p_0$, and to $q_{t1}^M$ such that  the quantum state becomes $\ket{\phi_1}$ with probability $p_1$, where: $$p_0 = \Abs{M_0[a]\ket{\psi}}^2,\quad \ket{\phi_0} = \frac{M_0[a] \ket\psi}{\Abs{M_0[a] \ket\psi}},\quad p_1 = \Abs{M_1[a]\ket{\psi}}^2,\quad \ket{\phi_1} = \frac{M_1[a] \ket\psi}{\Abs{M_1[a] \ket\psi}}$$ and $M_0[a] = \ket 0_a \bra 0$, $M_1[a] = I-M_0[a]$.
    \end{enumerate}

    Now let $P = P_0, P_1, \dots, P_{L-1}$ be a QRAM to be simulated by a $\text{TM}^\text{Q}$. By Lemma \ref{lemma-qram-address} and Lemma \ref{lemma-qram-measurement}, we can assume that $P$ is address-safe and measurement-postponed without any loss of generality. For every $0 \leq l < L$, we use a bunch of states $(p_l, 0), (p_l, 1), \dots, (p_l, k_l)$ to simulate instruction $P_l$, where for every $l$, $k_l$ is an appropriate integer. The state $(p_l, 0)$ indicates the beginning of the simulation of $P_l$. During the simulation of $P_l$, the intermediate states $(p_l, 1), \dots, (p_l, k_l)$ may be visited. In particular, $(p_L, 0)$ indicates the termination of the simulation, and is going to become $q_f$, which indicates the termination of the execution of $\text{TM}^\text{Q}$.

    Before constructing $\text{TM}^\text{Q}$, we first show how integers are stored in our machine. For every integer $n \in \mathbb{Z}$, we use $\mathit{bin}(n) \in \{0,1\}^*$ to denote its binary form. The first symbol of $\mathit{bin}(n)$ is $0$ if $n \geq 0$ and $1$ otherwise, which is followed by a binary representation of $\abs{n}$. Conversely, we use $\mathit{dec}(x)$ to denote the decimal value of the binary string $x$ if it is valid. We also need some TMs that perform arithmetic and other basic operations:
    \begin{itemize}
      \item $M_\text{inc}$ --- a TM for increment by one: for every $a \in \mathbb{Z}$,
        \[
          M_\text{inc}(\mathit{bin}(a)) = \mathit{bin}(a+1).
        \]
        The time of $M_\text{inc}$ is $O(\log \abs{a})$.
      \item $M_\text{add}$ --- a TM for addition: for every $a, b \in \mathbb{Z}$,
        \[
          M_\text{add}(\mathit{bin}(a);\mathit{bin}(b);\epsilon) = M_\text{add}(\mathit{bin}(a);\mathit{bin}(b);\mathit{bin}(a+b)).
        \]
        The time of $M_\text{add}$ is $O(\log \abs{a} + \log \abs{b})$.
      \item $M_\text{sub}$ --- a TM for subtraction. Formally, for every $a, b \in \mathbb{Z}$,
        \[
          M_\text{sub}(\mathit{bin}(a);\mathit{bin}(b);\epsilon) = M_\text{add}(\mathit{bin}(a);\mathit{bin}(b);\mathit{bin}(a-b)).
        \]
        The time of $M_\text{sub}$ is $O(\log \abs{a} + \log \abs{b})$.
      \item $M_\text{gtz}$ --- a TM for checking positivity: for every $a \in \mathbb{Z}$,
        \[
          M_\text{gtz}(\mathit{bin}(a)) = \begin{cases}
            1 & a > 0, \\
            0 & \text{otherwise}.
          \end{cases}
        \]
        The time of $M_\text{gtz}$ is $O(\log \abs{a})$.

  \item $M_\text{clean}$ --- a TM that cleans a track: for every $x \in \{0, 1\}^*$, $M_\text{clean}(x) = \epsilon$. The time of $M_\text{clean}$ on input $x$ is $O(\abs{x})$.
      \item $M_\text{read}$ --- a TM that reads a symbol from a track: for every $x = \sigma_0 \sigma_1 \dots \sigma_{k-1} \in \{0, 1\}^*$,
          \[
            M_\text{read}(x;\epsilon) = \begin{cases}
                \sigma_1\dots\sigma_{k-1};\mathit{bin}(0) & k \geq 1 \text{ and } \sigma_0 = 0, \\
                \sigma_1\dots\sigma_{k-1};\mathit{bin}(1) & k \geq 1 \text{ and } \sigma_0 = 1, \\
                \epsilon;\mathit{bin}(-1) & k = 0.
            \end{cases}
          \]
          The time of $M_\text{read}$ is $O(\abs{x})$.
      \item $M_\text{write($a$)}$ --- a TM that writes a specific (pre-determined) content $a$ to the end of a track: for every $x \in \{0, 1\}^*$, $M_\text{write($a$)}(x) = xa$. The time of $M_\text{write($a$)}$ is $O(\abs{x})$.
      \item $M_\text{append}$ --- a TM that appends the contents in the second track to the end of the first track: for every $x, y \in \{0,1\}^*$, $M_\text{append}(x;y) = xy; y$. The time of $M_\text{append}$ is $O(\abs{x}\abs{y})$.
      \item $M_\text{fetch}$ --- a TM that fetches the contents of registers: suppose the contents in the first track is $z = \text{L} a_1 \text{M} b_1 \text{R} \text{L} a_2 \text{M} b_2 \text{R} \dots \text{L} a_r \text{M} b_r \text{R}$. For every $x \in \{0,1\}^*$,
          \[
              M_\text{fetch}(z; x; \epsilon) = \begin{cases}
                z;x;b_i & a_i = x, \\
                z;x;\mathit{bin}(0) & \text{otherwise}.
              \end{cases}
          \]
          The time of $M_\text{fetch}$ is $O(\abs{z}(\abs{x}+\abs{y}))$, where $y$ is the contents in the third track after execution.
      \item $M_\text{update}$ --- a TM that updates the contents of registers: suppose the contents in the first track is $z = \text{L} a_1 \text{M} b_1 \text{R} \text{L} a_2 \text{M} b_2 \text{R} \dots \text{L} a_r \text{M} b_r \text{R}$. For every $x \in \{0,1\}^*$,
          \[
              M_\text{update}(z; x; y) = \begin{cases}
                \text{L} a_1 \text{M} b_1 \text{R} \dots \text{L} a_{i} \text{M} y \text{R} \dots \text{L} a_r \text{M} b_r \text{R};x;y & a_i = x, \\
                z\text{L}x\text{M}y\text{R};x & \text{otherwise}.
              \end{cases}
          \]
          The time of $M_\text{update}$ is $O(\abs{z}(\abs{x}+\abs{y}))$.
      \item $M_\text{qget}$ --- a TM that gets the physical address of a virtual address: suppose the contents in the first track is $z = \text{L} a_1 \text{M} b_1 \text{R} \text{L} a_2 \text{M} b_2 \text{R} \dots \text{L} a_r \text{M} b_r \text{R}$. For every $c \in \mathbb{N}$ and $x \in \{0, 1\}^+$,
          \[
              M_\text{qget}(z; \mathit{bin}(c); x; \epsilon) = \begin{cases}
                z;\mathit{bin}(c);x;b_i & x = a_i, \\
                z\text{L}x\text{M}\mathit{bin}(c+1)\text{R};\mathit{bin}(c+1);x;\mathit{bin}(c+1) & \text{otherwise}.
              \end{cases}
          \]
          The time of $M_\text{qget}$ is $O(\abs{z}(\abs{x}+\abs{y}+\log \abs{c}))$, where $y$ is the contents in the fourth track after execution.
      \item $M_\text{untary}$ --- a TM that converts a non-negative integer $c$ to a string $0^c1$: for every $c \in \mathbb{N}$,
          \[
            M_\text{untary}(\mathit{bin}(c)) = 0^c1.
          \]
          The time of $M_\text{untary}$ is $O(c\log c)$. This TM is used to produce contents in track ``qdev'' for $q_s^H, q_s^T$ and $q_s^M$ calls.
      \item $M_\text{pair}$ --- a TM that converts a two non-negative integers $a$ and $b$ ($a \neq b$) to a string $s_{ab} \in \{0,1,2\}^*$ of length $\abs{s_{ab}} = \max\{a,b\}+1$ that
          \[
            s_{ab}(c) = \begin{cases}
                1 & c = a, \\
                2 & c = b, \\
                0 & \text{otherwise},
            \end{cases}
          \]
          where $s(c)$ denotes the $c$-th symbol of $s$ (0-indexed). Formally, for every $a, b \in \mathbb{N}$ with $a \neq b$, then
          \[
            M_\text{pair}(\mathit{bin}(a);\mathit{bin}(b);\epsilon) = \mathit{bin}(a);\mathit{bin}(b);s_{ab}.
          \]
          The time of $M_\text{pair}$ is $O((a+b)(\log a+\log b))$. This TM is used to produce contents in track ``qdev'' for $q_s^C$ calls.
    \end{itemize}

    Now we are ready to construct the $\text{TM}^\text{Q}$. We assume that all of the TMs introduced above are stationary.
    Before simulation, we should initialize the ``qcnt'' track to be decimal value of zero, i.e.
    \begin{align*}
        q_0: & M_{\text{write($\mathit{bin}(0)$)}}[\text{qcnt}] \\
        q_1: & \text{transition to } (p_{0}, 0) \\
    \end{align*}
    For every $0 \leq l < L$, if $P_l$ is a classical instruction, it can be simulated in a standard way, and the details are omitted here but provided in Appendix \ref{qtm-in-sim}. The following are the simulation of quantum instruction $P_l$:
    \begin{enumerate}
       \item If $P_l$ has the form $\text{CNOT}[Q_{X_i}, Q_{X_j}]$, then we use:
          \begin{align*}
            (p_l, 0): & M_{\text{write($i$)}}[\text{work1}] \\
            (p_l, 1): & M_{\text{fetch}}[\text{creg}, \text{work1}, \text{work2}] \\
            (p_l, 2): & M_{\text{qget}}[\text{qreg}, \text{qcnt}, \text{work2}, \text{work3}] \\
            (p_l, 3): & M_{\text{write($j$)}}[\text{work4}] \\
            (p_l, 4): & M_{\text{fetch}}[\text{creg}, \text{work4}, \text{work5}] \\
            (p_l, 5): & M_{\text{qget}}[\text{qreg}, \text{qcnt}, \text{work5}, \text{work6}] \\
            (p_l, 6): & M_{\text{pair}}[\text{work3}, \text{work6}, \text{qdev}] \\
            (p_l, 7): & M_{\text{write($(p_l, 9)$)}}[\text{qret}] \\
            (p_l, 8): & \text{transition to } q_s^C \\
            q_t^C, & (p_l, 9)_\text{qret} \to \#_\text{qret}, (p_l, 9), L \\
            (p_l, 9), & \#_\text{qret} \to \#_\text{qret}, (p_l, 10), R \\
            (p_l, 10): & M_{\text{clean}}[\text{work1}] \\
            (p_l, 11): & M_{\text{clean}}[\text{work2}] \\
            (p_l, 12): & M_{\text{clean}}[\text{work3}] \\
            (p_l, 13): & M_{\text{clean}}[\text{work4}] \\
            (p_l, 14): & M_{\text{clean}}[\text{work5}] \\
            (p_l, 15): & M_{\text{clean}}[\text{work6}] \\
            (p_l, 16): & M_{\text{clean}}[\text{qdev}] \\
            (p_l, 17): & \text{transition to } (p_{l+1}, 0) \\
          \end{align*}
      \item If $P_l$ has the form $A[Q_{X_i}]$ with $A=H$ or $T$, then we use:
          \begin{align*}
            (p_l, 0): & M_{\text{write($i$)}}[\text{work1}] \\
            (p_l, 1): & M_{\text{fetch}}[\text{creg}, \text{work1}, \text{work2}] \\
            (p_l, 2): & M_{\text{qget}}[\text{qreg}, \text{qcnt}, \text{work2}, \text{qdev}] \\
            (p_l, 3): & M_{\text{untary}}[\text{qdev}] \\
            (p_l, 4): & M_{\text{write($(p_l, 6)$)}}[\text{qret}] \\
            (p_l, 5): & \text{transition to } q_s^A \\
            q_t^A, & (p_l, 6)_\text{qret} \to \#_\text{qret}, (p_l, 6), L \\
            (p_l, 6), & \#_\text{qret} \to \#_\text{qret}, (p_l, 7), R \\
            (p_l, 7): & M_{\text{clean}}[\text{work1}] \\
            (p_l, 8): & M_{\text{clean}}[\text{work2}] \\
            (p_l, 9): & M_{\text{clean}}[\text{qdev}] \\
            (p_l, 10): & \text{transition to } (p_{l+1}, 0) \\
          \end{align*}
       \item If $P_l$ has the form $X_i \gets M[Q_{X_j}]$, then we use:
          \begin{align*}
            (p_l, 0): & M_{\text{write($j$)}}[\text{work1}] \\
            (p_l, 1): & M_{\text{fetch}}[\text{creg}, \text{work1}, \text{work2}] \\
            (p_l, 2): & M_{\text{qget}}[\text{qreg}, \text{qcnt}, \text{work2}, \text{qdev}] \\
            (p_l, 3): & M_{\text{untary}}[\text{qdev}] \\
            (p_l, 4): & M_{\text{write($(p_l, 6)$)}}[\text{qret}] \\
            (p_l, 5): & \text{transition to } q_s^M \\
            q_{t0}^M, & (p_l, 6)_\text{qret} \to \#_\text{qret}, (p_l, 6), L \\
            q_{t1}^M, & (p_l, 6)_\text{qret} \to \#_\text{qret}, (p_l, 9), L \\
            (p_l, 6), & \#_\text{qret} \to \#_\text{qret}, (p_l, 7), R \\
            (p_l, 7): & M_{\text{write($\mathit{bin}(0)$)}}[\text{work3}] \\
            (p_l, 8): & \text{transition to } (p_l, 11) \\
            (p_l, 9), & \#_\text{qret} \to \#_\text{qret}, (p_l, 10), R \\
            (p_l, 10): & M_{\text{write($\mathit{bin}(1)$)}}[\text{work3}] \\
            (p_l, 11): & M_{\text{write($i$)}}[\text{work4}] \\
            (p_l, 12): & M_{\text{update}}[\text{creg}, \text{work4}, \text{work3}] \\
            (p_l, 13): & M_{\text{clean}}[\text{work1}] \\
            (p_l, 14): & M_{\text{clean}}[\text{work2}] \\
            (p_l, 15): & M_{\text{clean}}[\text{work3}] \\
            (p_l, 16): & M_{\text{clean}}[\text{work4}] \\
            (p_l, 17): & M_{\text{clean}}[\text{qdev}] \\
            (p_l, 18): & \text{transition to } (p_{l+1}, 0) \\
          \end{align*}
    \end{enumerate}

   To conclude this subsection, we prove the correctness of our simulation. Let $\mathcal{M}^\text{Q}$ denote the $\text{TM}^\text{Q}$ constructed according to the above description.

    \begin{lemma}
        For every $x, y \in \{0,1\}^*$, $P(x, y) = \mathcal{M}^\text{Q}(x, y)$.
    \end{lemma}
    \begin{proof}
        Clear by the construction.
    \end{proof}

    \begin{lemma}
        Suppose $P$ is a $T(n)$-time QRAM. Then:
        \begin{enumerate}
          \item If $l(n)$ is logarithmic, then the number of non-empty symbols in every track of $\mathcal{M}^\text{Q}$ is $O(T(n))$.
          \item If $l(n)$ is constant, then the number of non-empty symbols in every track of $\mathcal{M}^\text{Q}$ is $O(T(n)^2)$.
        \end{enumerate}
    \end{lemma}
    \begin{proof}
        We need to only focus on the tracks ``creg'' and ``qreg''.

        {\vskip 3pt}

        \textbf{Case 1}. $l(n)$ is logarithmic: For each executed classical-type QRAM instruction, at most one classical register is altered. Let $t_i$ denote the execution time of the $i$-th executed instruction. After the execution of this instruction, the number of non-empty symbols in track ``creg'' increases at most $O(t_i)$. Therefore, the number of non-empty symbols in track ``creg'' is
        \[
            \sum_i O(t_i) = O(T(n)).
        \]
        For each executed quantum-type QRAM instruction, at most one virtual address is assigned a physical address. Similarly, after the $i$-th executed instruction, the number of non-empty symbols in track ``qreg'' increases at most $O(t_i)$. Therefore, the number of non-empty symbols in track ``qreg'' is also $O(T(n))$.

        {\vskip 3pt}

        \textbf{Case 2}. $l(n)$ is constant: The analysis is similar. The only difference is that after each executed instruction, the number of non-empty symbols in track ``creg'' and ``qreg'' increase at most $O(T(n) t_i)$. This is because after executing $T(n)$ QRAM instructions, the largest possible address could be $2^{T(n)}$ (by repeatedly executing the instruction $X_i \gets X_j + X_k$ with $i = j = k$), which is of length $T(n)$ in its binary representation. Therefore, the number of non-empty symbols in track ``qreg'' is $O(T(n)^2)$.
    \end{proof}

    \begin{lemma}
        Suppose $P$ is a $T(n)$-time QRAM. Then,
        \begin{enumerate}
          \item If $l(n)$ is logarithmic, then $\mathcal{M}^\text{Q}$ is a $O(T(n)^2)$-time $\text{TM}^\text{Q}$.
          \item If $l(n)$ is constant, then $\mathcal{M}^\text{Q}$ is a $O(T(n)^4)$-time $\text{TM}^\text{Q}$.
        \end{enumerate}
    \end{lemma}
    \begin{proof} \textbf{Case 1}. $l(n)$ is logarithmic: Let $t_i$ denote the execution time of the $i$-th executed instruction, then the lengths of addresses accessed are bounded by $O(t_i)$. $\mathcal{M}^\text{Q}$ simulates this instruction with $O(t_i T(n))$ time, which comes from the usage of those basic TMs, e.g. $M_\text{append}$, $M_\text{fetch}$, $M_\text{update}$, $M_\text{qget}$, $M_\text{untary}$ and $M_\text{pair}$. Therefore, $\mathcal{M}^\text{Q}$ is $O(T(n)^2)$-time.

   {\vskip 3pt}

     \textbf{Case 2}. $l(n)$ is constant: The lengths of addresses accessed are bounded by $O(T(n))$ for each executed instruction. Hence, $\mathcal{M}^\text{Q}$ simulates each instruction with $O(T(n)^3)$ time. Consequently, $\mathcal{M}^\text{Q}$ is $O(T(n)^4)$-time.
    \end{proof}

    \subsubsection{QTMs simulate $\text{TM}^\text{Q}$s} \label{sec:tmq-by-qtm}

    Our strategy of simulating $\mathcal{M}^\text{Q}$ - the $\text{TM}^\text{Q}$ constructed above - by a QTM is presented in the following lemma and its proof.

    \begin{lemma}
        There is a well-formed, normal form, stationary and unidirectional QTM $M$ within time $O(T(n)^2)$ such that $M(x, y) = \mathcal{M}^\text{Q}(x, y)$ for every $x, y \in \{0,1\}^*$.
    \end{lemma}

    \begin{proof} Let $\mathcal{M}^\text{Q} = (Q, Q_s, Q_t, \Sigma, \delta, \lambda, q_0, q_f)$. We recall that $\mathcal{M}^\text{Q}$ is measurement-postponed and stationary. Our basic idea is to simulate $\mathcal{M}^\text{Q}$ by maintaining a history to make the simulation reversible. The technique we use here is partly borrowed from \cite{Ber97}.

The TM $M$ used to simulate $\mathcal{M}^\text{Q}$ has five tracks: \begin{itemize}\item The first track, with alphabet $\Sigma_1 = \Sigma$, is used to simulate the tape of $\mathcal{M}^\text{Q}$.

\item The second track, with track $\Sigma_2 = \{ \#, @ \}$, is used to store a $@$ indicating the position of tape head of $\mathcal{M}^\text{Q}$.

\item The third track, with alphabet $\Sigma_3 = \{ \#, \$ \} \cup ((Q \setminus Q_s) \times \Sigma)$ is used to write down a list of the transitions taken by $\mathcal{M}^\text{Q}$, starting with the end marker $\$$.

\item The fourth track is a ``quantum'' track and will be defined later.

\item The fifth track is an ``extra'' quantum track for measurements.\end{itemize}

Now we can elaborate the construction of $M$. We use $\forall_i$ to denote any symbol in $\Sigma_i$ and $\forall_3'$ to denote any symbol in $\Sigma_3 \setminus \{ \#, \$ \}$. Since the fourth and the fifth tracks are not used in classical operations, we write only the first three tracks in the transitions unless the fourth or the fifth track are needed. The first stage of simulation needs the state set $$Q_1 = Q \cup ((Q \setminus Q_t) \times (Q \setminus Q_s) \times \Sigma \times \{ 1, 2, 3, 4 \}) \cup ((Q \setminus Q_t) \times \{ 5, 6, 7 \}) \cup \{ q_a, q_b \}.$$ The initial state is $q_a$ and the final state is $q_f$. The transitions are defined as follows:
\begin{enumerate}
  \item
    At the beginning, we write an end marker $@$ on the second track and end marker $\$$ on the third track, and then come back to the initial position with state $q_0$. Include the instructions:
    \[
    \begin{matrix}
        q_a, & (\forall_1, \#_2, \#_3) & \to & (\forall_1, @, \#), & q_b, & R; & 1 \\
        q_b, & (\forall_1, \#_2, \#_3) & \to & (\forall_1, \#, \$), & q_c, & R; & 1 \\
        q_c, & (\forall_1, \forall_2, \forall_3) & \to & (\forall_1, \forall_2, \forall_3), & q_d, & L; & 1 \\
        q_d, & (\forall_1, \forall_2, \forall_3) & \to & (\forall_1, \forall_2, \forall_3), & q_e, & L; & 1 \\
        q_e, & (\forall_1, \forall_2, \forall_3) & \to & (\forall_1, \forall_2, \forall_3), & q_g, & L; & 1 \\
        q_g, & (\forall_1, \forall_2, \forall_3) & \to & (\forall_1, \forall_2, \forall_3), & q_0, & R; & 1 \\
    \end{matrix}
    \]
  \item For $p \in Q \setminus Q_s, \sigma \in \Sigma$ with $p \neq q_f$ and with transition $\delta(p, \sigma) = (\tau, q, d)$ in $M$, we make transitions to go from $p$ to $q$ updating the first track, i.e. the simulated tape of $M$, and adding $(p, \sigma)$ to the end of the history. Include the instructions:
    \[
    \begin{matrix}
        p, & (\sigma, @, \forall_3) & \to & (\tau, \#_2, \forall_3), & (q, p, \sigma, 1), & d; & 1 \\
        (q, p, \sigma, 1), & (\forall_1, \#_2, \$) & \to & (\forall_1, @, \$), & (q, p, \sigma, 3), & R; & 1 \\
        (q, p, \sigma, 1), & (\forall_1, \#_2, \forall_3') & \to & (\forall_1, @, \forall_3'), & (q, p, \sigma, 3), & R; & 1 \\
        (q, p, \sigma, 1), & (\forall_1, \#_2, \#_3) & \to & (\forall_1, @, \#_3), & (q, p, \sigma, 2), & R; & 1 \\
        (q, p, \sigma, 2), & (\forall_1, \#_2, \#_3) & \to & (\forall_1, \#_2, \#_3), & (q, p, \sigma, 2), & R; & 1 \\
        (q, p, \sigma, 2), & (\forall_1, \#_2, \$) & \to & (\forall_1, \#_2, \$), & (q, p, \sigma, 3), & R; & 1 \\
        (q, p, \sigma, 3), & (\forall_1, \#_2, \forall_3') & \to & (\forall_1, \#_2, \forall_3'), & (q, p, \sigma, 3), & R; & 1 \\
        (q, p, \sigma, 3), & (\forall_1, \#_2, \#_3) & \to & (\forall_1, \#_2, (p, \sigma)), & (q, 4), & R; & 1 \\
    \end{matrix}
    \]
    Whenever $(q, 4)$ is reached, the tape head is on the first blank after the end of the history (on the third track). Now we move the tape head back to the position of tape head of $M$ by including the instructions:
    \[
    \begin{matrix}
        (q, 4), & (\forall_1, \#_2, \#_3) & \to & (\forall_1, \#_2, \#_3), & (q, 5), & L; & 1 \\
        (q, 5), & (\forall_1, \#_2, \forall_3') & \to & (\forall_1, \#_2, \forall_3'), & (q, 5), & L; & 1 \\
        (q, 5), & (\forall_1, \#_2, \$) & \to & (\forall_1, \#_2, \$), & (q, 6), & L; & 1 \\
        (q, 5), & (\forall_1, @, \forall_3') & \to & (\forall_1, @, \forall_3'), & (q, 7), & L; & 1 \\
        (q, 5), & (\forall_1, @, \$) & \to & (\forall_1, @, \$), & (q, 7), & L; & 1 \\
        (q, 6), & (\forall_1, \#_2, \#_3) & \to & (\forall_1, \#_2, \#_3), & (q, 6), & L; & 1 \\
        (q, 6), & (\forall_1, @, \#_3) & \to & (\forall_1, @, \#_3), & (q, 7), & L; & 1 \\
        (q, 7), & (\forall_1, \forall_2, \forall_3) & \to & (\forall_1, \forall_2, \forall_3), & q, & R; & 1 \\
    \end{matrix}
    \]
    \item For $p \in Q_s$, we need the ``quantum'' track (namely, the fourth track) to simulate quantum operations, but the second and third tracks are useless now. Thus in the following discussion, we only write the first track and the ``quantum'' track in transitions. Let $\Sigma_q = \{ 0, 1 \}$ denote the alphabet of the ``quantum'' track, whose contents are initially set to $0$ (in other words, $0$ is used to be the empty symbol in the ``quantum'' track, for readability). The four quantum-type operations are implemented as follows.

        \textbf{Case 1}. $p = q_s^H$. Include the instructions:
        \[
        \begin{matrix}
            q_s^H, & (\forall_1, \forall_q) & \to & (\forall_1, \forall_q), & q_1^H, & L; & 1 \\
            q_1^H, & (\#_1, \forall_q) & \to & (\#_1, \forall_q), & q_2^H, & R; & 1 \\
            q_2^H, & (0, \forall_q) & \to & (0, \forall_q), & q_2^H, & R; & 1 \\
            q_2^H, & (1, 0) & \to & (1, 0), & q_2^H, & R; & 1/\sqrt2 \\
            q_2^H, & (1, 0) & \to & (1, 1), & q_2^H, & R; & 1/\sqrt2 \\
            q_2^H, & (1, 1) & \to & (1, 0), & q_2^H, & R; & 1/\sqrt2 \\
            q_2^H, & (1, 1) & \to & (1, 1), & q_2^H, & R; & -1/\sqrt2 \\
            q_2^H, & (\#_1, 0) & \to & (\#_1, 0), & q_3^H, & L; & 1 \\
            q_3^H, & (0, \forall_q) & \to & (0, \forall_q), & q_3^H, & L; & 1 \\
            q_3^H, & (1, \forall_q) & \to & (1, \forall_q), & q_3^H, & L; & 1 \\
            q_3^H, & (\#_1, 0) & \to & (\#_1, 0), & q_t^H, & R; & 1 \\
        \end{matrix}
        \]

        \textbf{Case 2}. $p = q_s^T$. Include the instructions:
        \[
        \begin{matrix}
            q_s^T, & (\forall_1, \forall_q) & \to & (\forall_1, \forall_q), & q_1^T, & L; & 1 \\
            q_1^T, & (\#_1, \forall_q) & \to & (\#_1, \forall_q), & q_2^T, & R; & 1 \\
            q_2^T, & (0, \forall_q) & \to & (0, \forall_q), & q_2^T, & R; & 1 \\
            q_2^T, & (1, 0) & \to & (1, 0), & q_2^T, & R; & 1 \\
            q_2^T, & (1, 1) & \to & (1, 1), & q_2^T, & R; & \exp (i \pi/4) \\
            q_2^T, & (\#_1, 0) & \to & (\#_1, 0), & q_3^T, & L; & 1 \\
            q_3^T, & (0, \forall_q) & \to & (0, \forall_q), & q_3^T, & L; & 1 \\
            q_3^T, & (1, \forall_q) & \to & (1, \forall_q), & q_3^T, & L; & 1 \\
            q_3^T, & (\#_1, 0) & \to & (\#_1, 0), & q_t^T, & R; & 1 \\
        \end{matrix}
        \]

        \textbf{Case 3}. $p = q_s^C$. Include the instructions:
        \[
        \begin{matrix}
            q_s^C, & (\forall_1, \forall_q) & \to & (\forall_1, \forall_q), & q_1^C, & L; & 1 \\
            q_1^C, & (\#_1, \forall_q) & \to & (\#_1, \forall_q), & (q_2^C, 0), & R; & 1 \\
            (q_2^C, 0), & (0, \forall_q) & \to & (0, \forall_q), & (q_2^C, 0), & R; & 1 \\
            (q_2^C, 0), & (1, 0) & \to & (1, 0), & (q_2^C, 0), & R; & 1 \\
            (q_2^C, 0), & (1, 1) & \to & (1, 1), & (q_2^C, 1), & R; & 1 \\
            (q_2^C, 0), & (2, \forall_q) & \to & (2, \forall_q), & (q_2^C, 0), & R; & 1 \\
            (q_2^C, 0), & (\#_1, \forall_q) & \to & (\#_1, \forall_q), & (q_3^C, 0), & L; & 1 \\
            (q_2^C, 1), & (0, \forall_q) & \to & (0, \forall_q), & (q_2^C, 1), & R; & 1 \\
            (q_2^C, 1), & (1, 0) & \to & (1, 0), & (q_2^C, 1), & R; & 1 \\
            (q_2^C, 1), & (1, 1) & \to & (1, 1), & (q_2^C, 0), & R; & 1 \\
            (q_2^C, 1), & (2, \forall_q) & \to & (2, \forall_q), & (q_2^C, 1), & R; & 1 \\
            (q_2^C, 1), & (\#_1, \forall_q) & \to & (\#_1, \forall_q), & (q_3^C, 1), & L; & 1 \\
            (q_3^C, 0), & (0, \forall_q) & \to & (0, \forall_q), & (q_3^C, 0), & L; & 1 \\
            (q_3^C, 0), & (1, \forall_q) & \to & (1, \forall_q), & (q_3^C, 0), & L; & 1 \\
            (q_3^C, 0), & (2, 0) & \to & (2, 0), & (q_3^C, 0), & L; & 1 \\
            (q_3^C, 0), & (2, 1) & \to & (2, 1), & (q_3^C, 0), & L; & 1 \\
            (q_3^C, 0), & (\#_1, \forall_q) & \to & (\#_1, \forall_q), & (q_4^C, 0), & R; & 1 \\
            (q_3^C, 1), & (0, \forall_q) & \to & (0, \forall_q), & (q_3^C, 1), & L; & 1 \\
            (q_3^C, 1), & (1, \forall_q) & \to & (1, \forall_q), & (q_3^C, 1), & L; & 1 \\
            (q_3^C, 1), & (2, 0) & \to & (2, 1), & (q_3^C, 1), & L; & 1 \\
            (q_3^C, 1), & (2, 1) & \to & (2, 0), & (q_3^C, 1), & L; & 1 \\
            (q_3^C, 1), & (\#_1, \forall_q) & \to & (\#_1, \forall_q), & (q_4^C, 1), & R; & 1 \\
            (q_4^C, 0), & (0, \forall_q) & \to & (0, \forall_q), & (q_4^C, 0), & R; & 1 \\
            (q_4^C, 0), & (1, 0) & \to & (1, 0), & (q_4^C, 0), & R; & 1 \\
            (q_4^C, 0), & (1, 1) & \to & (1, 1), & (q_4^C, 1), & R; & 1 \\
            (q_4^C, 0), & (2, \forall_q) & \to & (2, \forall_q), & (q_4^C, 0), & R; & 1 \\
            (q_4^C, 0), & (\#_1, \forall_q) & \to & (\#_1, \forall_q), & q_5^C, & L; & 1 \\
            (q_4^C, 1), & (0, \forall_q) & \to & (0, \forall_q), & (q_4^C, 1), & R; & 1 \\
            (q_4^C, 1), & (1, 0) & \to & (1, 0), & (q_4^C, 1), & R; & 1 \\
            (q_4^C, 1), & (1, 1) & \to & (1, 1), & (q_4^C, 0), & R; & 1 \\
            (q_4^C, 1), & (2, \forall_q) & \to & (2, \forall_q), & (q_4^C, 1), & R; & 1 \\
            q_5^C, & (0, \forall_q) & \to & (0, \forall_q), & q_5^C, & L; & 1 \\
            q_5^C, & (1, \forall_q) & \to & (1, \forall_q), & q_5^C, & L; & 1 \\
            q_5^C, & (2, \forall_q) & \to & (2, \forall_q), & q_5^C, & L; & 1 \\
            q_5^C, & (\#_1, \forall_q) & \to & (\#_1, \forall_q), & q_t^C, & R; & 1 \\
        \end{matrix}
        \]

        \textbf{Case 4}. $p = q_s^M$. In order to make the two possible measurement outcomes distinguishable in the successive configurations, we need an extra track with alphabet $\Sigma_e = \{\#, 0, 1\}$. Include the instructions:
        \[
        \begin{matrix}
            q_s^M, & (\forall_1, \forall_q, \forall_e) & \to & (\forall_1, \forall_q, \forall_e), & q_1^M, & L; & 1 \\
            q_1^M, & (\#_1, \forall_q, \forall_e) & \to & (\#_1, \forall_q, \forall_e), & q_2^M, & R; & 1 \\
            q_2^M, & (0, \forall_q, \forall_e) & \to & (0, \forall_q, \forall_e), & q_2^M, & R; & 1 \\
            q_2^M, & (1, 0, \#_e) & \to & (1, 0, 0), & (q_3^M, 0), & R; & 1 \\
            q_2^M, & (1, 1, \#_e) & \to & (1, 1, 1), & (q_3^M, 1), & R; & 1 \\
            (q_3^M, x), & (0, \forall_q, \forall_e) & \to & (0, \forall_q, \forall_e), & (q_3^M, x), & R; & 1 \\
            (q_3^M, x), & (\#_1, \forall_q, \forall_e) & \to & (\#_1, \forall_q, \forall_e), & (q_4^M, x), & L; & 1 \\
            (q_4^M, x), & (0, \forall_q, \forall_e) & \to & (0, \forall_q, \forall_e), & (q_4^M, x), & L; & 1 \\
            (q_4^M, x), & (1, \forall_q, \forall_e) & \to & (1, \forall_q, \forall_e), & (q_4^M, x), & L; & 1 \\
            (q_4^M, x), & (\#_1, \forall_q, \forall_e) & \to & (\#_1, \forall_q, \forall_e), & q_{tx}^M, & R; & 1 \\
        \end{matrix}
        \]
        In the above construction, the ``extra'' track is used to describe the measurement results at each position in the ``quantum'' track, which guarantees that no interference happens between the branches with different measurement results. This construction heavily depends on the condition that the simulated $\text{TM}^Q$ is measurement-postponed because each position of the extra track is allowed to be altered during the execution only once.

    \item Finally, to make $M$ in normal form, we add the transition:
    \[
    \begin{matrix}
        q_f, & (\forall_1, \forall_2, \forall_3) & \to & (\forall_1, \forall_2, \forall_3), & q_a, & R; & 1 \\
    \end{matrix}
    \]
\end{enumerate}

    It can be easily verified that the QTM $M$ constructed above is well-formed, normal form, stationary, unidirectional, and within time $O(T(n)^2)$.
    \end{proof}

    \subsection{QRAMs simulate QTMs} \label{sec:qtm-by-qram}

    Now we turn to consider how to simulate QTMs by QRAMs. Our simulation strategy is divided into the following three steps:
    \begin{enumerate}\item Simulate a QTM by a family of quantum circuits with the technique developed in \cite{Yao93} and \cite{Nis02} --- Subsection \ref{circuit-qtm}.

    \item  Use the Solovay-Kitaev algorithm \cite{Daw05} to decompose the gates used in these quantum circuits into basic gates $H,T$ and CNOT, within bounded errors --- Subsection \ref{Kitaev}.

    \item Translate the family of quantum circuits with the basic gates into a QRAM --- Subsection \ref{sec-concate}.
    \end{enumerate}

    \subsubsection{Quantum Circuit Families simulate QTMs}\label{circuit-qtm}

    We write $[n] = \{1, 2, \dots, n\}$ and let $\mathcal{G}$ be a finite set of gates with their qubits indexed from $1$ to $n$. A $k$-qubit quantum gate ($1 \leq k \leq n$) can be written as $G=U[q_1,...,q_k]$, indicating a $2^k\times 2^k$-unitary operator $U$ on qubits $q_1,\dots,q_k \in [n]$, where $q_1, \dots, q_k$ are pairwise distinct.

   \begin{definition} An $n$-qubit $a$-input $b$-output quantum circuit $C$ over $\mathcal{G}$ is a $4$-tuple
    $C = (\mathcal{U}, A, B, f),$
    where:
    \begin{enumerate}
      \item $\mathcal{U}$ is a finite sequence of gates from $\mathcal{G}$;
      \item $A \subseteq [n]$ with $\abs{A} = a$ is the set of input qubits;
      \item $B \subseteq [n]$ with $\abs{B} = b$ is the set of output qubits;
      \item $f: [n] \setminus A \to \{0,1\}$ is the initial setting for non-input qubits.
    \end{enumerate}\end{definition}

    Now we describe the computation of circuit $C$. Suppose $A = \{ i_1, i_2, \dots, i_a \}$ and $\mathcal{U}=G_1 G_2... G_t$. For every $x = x_1x_2\dots x_a \in \{0, 1\}^a$, the input state is set to $\ket{\psi_x} = \ket{u_1}\ket{u_2}\dots\ket{u_n}$, where
    \[
        u_k = \begin{cases}
            x_l & i_l = k, \\
            f(k) & \text{otherwise}.
        \end{cases}
    \]
    The final state of $\mathcal{U}$ on input $x$ is $\ket{\phi_x} = G_t \dots G_2 G_1 \ket{\psi_x}$. Then we measure it in the computational basis of the output qubits $q_i$ $(i\in B)$, and $C$ outputs $y = y_1y_2\dots y_b \in \{0,1\}^b$ with probability $\Abs{M_y \ket{\phi_x}}^2$,
    where measurement operator $M_y = \sum_{v \in S_y} \ket v \bra v$ and $$S_y = \{ v \in \{0,1\}^n: v_{j_l} = y_l \text{ for every } l \in [n] \}.$$ In this way, $C$ defines a function $C: \{0,1\}^a \times \{0,1\}^b \to [0, 1]$ that $$C(x, y) = \Abs{M_y \ket{\phi_x}}^2,$$ meaning that $C$ on input $x$ outputs $y$ with probability $C(x, y)$.

    \begin{lemma} \label{lemma-qcircuit-error}
        Let $C_1 = (\mathcal{U}_1, A, B, f)$ and $C_2 = (\mathcal{U}_2, A, B, f)$, and let $0 < \epsilon < 1$. If $\Abs{\mathcal{U}_1-\mathcal{U}_2}_2 < \varepsilon$, then for every $x \in \{0,1\}^{\abs{A}}$ and $y \in \{0,1\}^{\abs{B}}$: $$\abs{C_1(x, y)-C_2(x, y)} < 3\varepsilon.$$
    \end{lemma}
    \begin{proof}
        We can write $\mathcal{U}_1 = \mathcal{U}_2+J$ with $\Abs{J}_2 < \varepsilon$. Then:
        \begin{align*}
            \abs{C_1(x, y) - C_2(x, y)}
            & = \abs{ \Abs{M_y \mathcal{U}_1 \ket{\psi_x}}^2 - \Abs{M_y \mathcal{U}_2 \ket{\psi_x}}^2 } \\
            & = \abs{ \bra{\psi_x} \mathcal{U}_1^\dag M_y \mathcal{U}_1 \ket{\psi_x} - \bra{\psi_x} \mathcal{U}_2^\dag M_y \mathcal{U}_2 \ket{\psi_x} } \\
            & \leq \Abs{\mathcal{U}_1^\dag M_y \mathcal{U}_1-\mathcal{U}_2^\dag M_y \mathcal{U}_2}_2 \\
            & \leq \Abs{(\mathcal{U}_2^\dag+J^\dag) M_y (\mathcal{U}_2+J)-\mathcal{U}_2^\dag M_y \mathcal{U}_2}_2 \\
            & = \Abs{ J^\dag M_y \mathcal{U}_2 + \mathcal{U}_2^\dag M_y J + J^\dag M_y J }_2 \\
            & \leq \Abs{J}_2 + \Abs{J}_2 + \Abs{J}_2^2 \\
            & < 2\varepsilon + \varepsilon^2 < 3\varepsilon.
        \end{align*}
    \end{proof}

    Suppose the unitary operators appearing in $\mathcal{G}$ are $U_1, U_2, \dots, U_m$, and for $1\leq i\leq m$, $U_i$ is a $c_i$-qubit unitary operator, i.e. a $2^{c_i} \times 2^{c_i}$ unitary matrix. Then the description of circuit $C$ is a sequence of integers of the form
    \[
        \begin{matrix}
            g_1, & q_{1,1}, & \dots, & q_{1,c_{g_1}}, \\
            g_2, & q_{2,1}, & \dots, & q_{2,c_{g_2}}, \\
            \dots, \\
            g_t, & q_{t,1}, & \dots, & q_{t,c_{g_t}}, \\
            -1, \\
            i_1, & i_2, & \dots, & i_a, \\
            -1, \\
            j_1, & j_2, & \dots, & j_b, \\
            -1, \\
            f_1, & f_2, & \dots, & f_n, \\
            -1. \\
        \end{matrix}
    \]
    The sequence consists of four parts with $-1$ as the separator.
    \begin{enumerate}
      \item The first part describes $\mathcal{U} = G_1 G_2 \dots G_t$, where $G_i = U_{g_i} [q_{i,1}, \dots, q_{i,c_{g_i}}]$ for $1 \leq i \leq t$.
      \item The second part describes $A = \{ i_1, i_2, \dots, i_a \}$.
      \item The third part describes $B = \{ j_1, j_2, \dots, j_b \}$.
      \item The fourth part describes $f$ such that $f(k) = f_k$ for every $k \notin A$.
    \end{enumerate}
    The arguments $t, a, b$ and $n$ are obtained by counting the integers in their corresponding parts.

       \begin{definition}Let $M$ be a QTM and $\{C_n\}_{n=0}^\infty$ a family of quantum circuits, where $C_n$ is an $n$-input $b(n)$-output quantum circuit for every $n \in \mathbb{N}$ and $b(n) = (2t_n+1){\ceil{\log_2 \abs{\Sigma}}}$. We say that $\{C_n\}_{n=0}^\infty$ simulates $M$ if for every $x \in \{0,1\}^*$ and $y \in \{0,1\}^*$: $$M(x, y) = \sum_{\text{extract}(\text{tape}(z)) = y} C_{\abs{x}}(x, z),$$ where $\text{tape}(z)$ denotes the tape that $z$ represents. More precisely, if we write $z = z_1 z_2 \dots z_{2t+1}$ with $z_i \in \{0, 1\}^{\ceil{\log_2 \abs{\Sigma}}}$ for every $1 \leq i \leq 2t+1$ and regard $z_i$ as an integer of binary form $z_i$, then
              \[
                \text{tape}(z)(m) = \begin{cases}
                    \mathit{out}(z_{m+t+1}) & -t \leq m \leq t, \\
                    \# & \text{otherwise},
                \end{cases}
              \]
              where $\Sigma = \{\sigma_0, \sigma_1, \dots, \sigma_{\abs{\Sigma}-1}\}$.
     \end{definition}

    It was proven in \cite{Yao93, Nis02} that each QTM can be efficiently simulated by a family of quantum circuits. One of the main results in \cite{Yao93, Nis02} can be restated in a way convenient for our purpose as the following:

    \begin{theorem} \label{thm-qcf-qtm}
        Let $T: \mathbb{N} \to \mathbb{N}$ with $T(n) \geq n$ that is time-constructible (for example, by a RAM). For every standard QTM $M$ with exact time $T(n)$, one can find:
        \begin{itemize}
          \item three unitary matrices $U_1, U_2, U_3$ with their elements in $\mathbb{C}(M) \cup \{0,1\}$, each of which acts on at most $6\ell$ qubits, where $\ell = 2+\ceil{\log_2(\abs{Q}+1)}+\ceil{\log_2\abs{\Sigma}}$, and
          \item a classical algorithm $\mathbb{A}$ with time complexity $O(T(n)^2l(T(n)))$ (considered as a RAM with cost function $l(n)$ being constant or logarithmic),
        \end{itemize}
such that \begin{enumerate}
          \item for every $n \in \mathbb{N}$, on input $1^n$, $\mathbb{A}$ outputs the description of a $k(n)$-qubit $n$-input $b(n)$-output quantum circuit $C_n$ of size $O(T(n)^2)$, only using unitary matrices $U_1, U_2, U_3$,
     where: $$k(n)=(2T(n)+4)\ell,\quad b(n)=\ceil{\log_2\abs{\Sigma}}(2T(n)+1);$$
          \item $\{C_n\}_{n=0}^\infty$ simulates $M$.
             \end{enumerate}
    \end{theorem}

    Theorem \ref{thm-qcf-qtm} also holds for QRAMs instead of RAMs, if $T(n)$ is assumed to be QRAM-time constructible.

    \subsubsection{The Solovay-Kitaev Algorithm}\label{Kitaev}

    Our next step is to decompose the unitary matrices $U_1,U_2,U_3$ given in Theorem \ref{thm-qcf-qtm} into the basic gates $H,T$ and CNOT. Let us first briefly review the Solovay-Kitaev algorithm from \cite{Daw05}.

    \begin{definition}
        A set $\mathcal{W}$ of $d\times d$ matrices is called universal for $\mathit{SU}(d)$ if:
        \begin{enumerate}
          \item $\mathcal{W} \subseteq \mathit{SU}(d)$, i.e. for every $U \in \mathcal{W}$, $U^\dag U = U U^\dag = I$ and $\abs{U} = 1$.
          \item For every $U \in \mathcal{W}$, we also have $U^\dag \in \mathcal{W}$.
          \item for every $U \in \mathit{SU}(d)$ and $\varepsilon > 0$, there is a sequence $U_1, U_2, \dots, U_m \in \mathcal{W}$ such that
              \[
                \Abs{U-U_m\dots U_2U_1}_2 < \varepsilon.
              \]
        \end{enumerate}
    \end{definition}

    \begin{theorem} [Solovay-Kitaev Theorem \cite{Daw05}] \label{thm-sk}
        Let $\mathcal{W} = \{ U_1, U_2, \dots, U_k \}$ be universal for $\mathit{SU}(d)$. Then there is a classical algorithm with time complexity $O(\log^c(1/\varepsilon))$ (of which the constant factor depends on $d$) for some constant $c > 0$ that on input $\varepsilon > 0$ and $U \in \mathit{SU}(d)$, outputs a sequence $i_1, i_2, \dots, i_m \in [k]$ such that
        \begin{enumerate}
          \item $\Abs{U-U_{i_m}\dots U_{i_2}U_{i_1}}_2 < \varepsilon$.
          \item $m = O(\log^c(1/\varepsilon))$.
        \end{enumerate}
        More explicitly, $U$ is represented by a $d \times d$ unitary matrix each of whose elements is described as a floating point number within a high enough precision whose length is bounded by $O(\log^c(1/\varepsilon))$.
    \end{theorem}

    Now we can use the Solovay-Kitaev algorithm to find good approximations of $U_1, U_2, U_3$ in Theorem \ref{thm-qcf-qtm} by the basic gates. Since $U_1, U_2, U_3$ are unitary operators on $d = 6\ell$ qubits, we choose the set of basic gates:
    \[
        \mathcal{G} = \{ \text{CNOT}[a, b], H[a], T[a]: 1 \leq a, b \leq d, a \neq b \},
    \]
    where $\text{CNOT}[a,b]$ denotes a CNOT gate with the $a$-th qubit as its control qubit and the $b$-th qubit as its target qubit, and $H[a]$ and $T[a]$ denote Hadamard and $\pi/8$ gates on the $a$-th qubit. For every $\varepsilon > 0$, since the matrix elements of $U_1,U_2,U_3$ are in $\mathbb{C}(\lambda(n))$ (with $\lambda(n) \geq n$ polynomial), we can compute each element within a high enough precision in $O(\lambda(\log(1/\varepsilon))^c)$ time. Therefore, using the algorithm stated in Theorem \ref{thm-sk}, we can decompose $U_1, U_2, U_3$ into the basic gates $H$, $T$ and CNOT within precision $\varepsilon$ in $O(\lambda(\log(1/\varepsilon))^{c})$ time.

    \subsubsection{QRAMs simulate QTMs}\label{sec-concate}

    Now we can finish the construction of the QRAM $P$ that given $\varepsilon > 0$, simulates a standard QTM $M$ in the sense that $$\abs{P(x, y) - M(x, y)} < \varepsilon.$$ Suppose $M$ is a standard QTM with exact time $T(n)$ and $T(n)$ is QRAM-time constructible. Since a RAM can be seen as a special QRAM, we can turn the algorithm given in Theorem \ref{thm-qcf-qtm} to a $O(T(n)^2l(T(n)))$-time QRAM $P_1$ that on input $1^n$, outputs a description of quantum circuit $C_n$. By the Solovay-Kitaev algorithm, there is a $O(\lambda(\log(1/\epsilon))^{c})$-time QRAM $P_2$ that on input $\epsilon > 0$ and $U \in \mathit{SU}(d)$ with matrix elements in $\mathbb{C}(\lambda(n))$, outputs a sequence of basic gates $G_1, G_2, \dots, G_m$ of length $m = O(\log^c(1/\epsilon))$ such that $$\Abs{U-G_m\dots G_2G_1}_2 < \epsilon.$$
Then the QRAM can be described as follows:

    \textbf{Step 0}. We hardcode the three quantum gates $U_1, U_2, U_3$ in Theorem \ref{thm-qcf-qtm} into our QRAM $P$ for the later use.

    \textbf{Step 1}. Read the input string $x \in \{0, 1\}^*$ and count the length of $x$, i.e. $n = \abs{x}$.

    \textbf{Step 2}. Apply $P_1$ on input $1^n$ and obtain a description of quantum circuit $C_n$. According to Theorem \ref{thm-qcf-qtm}, there are $t = O(T(n)^2)$ gates in $C_n$, and each of them is an application of the unitary operator $U_1,U_2$ or $U_3$.

    \textbf{Step 3}. For each gate $G_i$ in $C_n$, apply $P_2$ to get an approximation $\tilde G_i$ of $G_i$ such that $\Abs{G_i-\tilde G_i} < \epsilon$, where $\epsilon = \varepsilon/3t$. This takes time $O(\lambda(\log(1/\epsilon))^c)$. Note that the size of $\tilde G_i$ is $O(\log^c(1/\epsilon))$. By replacing each $G_i$ in $C_n$ by $\tilde G_i$, we obtain a circuit $\tilde C_n$ consisting of only the basic gates $H,T$ and CNOT.
    By Lemma \ref{lemma-qcircuit-error}, we have: $$\abs{C_n(x, y) - \tilde C_n(x, y)} < \varepsilon.$$

    \textbf{Step 4}. Simulate $\tilde C_n$ with quantum-type QRAM instructions.

    Note that the size of $\tilde C_n$ is $O(t\log^c(1/\epsilon))$. Therefore, QRAM $P$ has running time $$O(t \operatorname{poly}(\lambda(\log(1/\epsilon)))) = O(T(n)^2 \operatorname{poly}(\lambda(\log (T(n)/\varepsilon)))).$$

\section{Standardisation of QTMs} \label{sec:standard-qtm}

 The aim of this section is to prove the Standardisation Theorem for QTMs (Theorem \ref{thm-qtm}).

\subsection{Properties of TMs and QTMs} \label{properties-of-tms-and-qtms}

We first present several lemmas about reversible TMs and QTMs needed in our proof of Theorem \ref{thm-qtm}. Some of them are generalised from \cite{Ber97}, and some are new.

\subsubsection{Several Lemmas for TMs}

    \begin{definition}
        A deterministic (classical) TM is said to be oblivious if its running time and head position at each time step depend only on the length of input. That is, there is a function $T: \mathbb{N} \to \mathbb{N}$ and a function $\mathit{pos}: \mathbb{N} \times \mathbb{N} \to \mathbb{Z}$ such that on input $x \in \{0,1\}^*$, the running time is $T(\abs{x})$ and the head position at time $t$ is $\mathit{pos}(\abs{x}, t)$.
    \end{definition}

    We note that a stationary, normal form, oblivious reversible TM is a standard QTM. We need the following lemmas, and their proofs are given in Appendix \ref{appendix-details-tms-and-qtms}.

    \begin{lemma} [Incrementing] \label{prop-inc}
        There is a stationary oblivious reversible TM $M$ that on input $x \in \{0,1\}^+$, produces $x^+ = (x+1) \bmod 2^{\abs{x}}$ in $O(\abs{x}^2)$ time, where $x+1$ denotes the arithmetic addition of $x$ and $1$, and $\abs{x}$ denotes the length of $x$ in binary. In other words,
        \[
            x \xrightarrow[T]{M} {x^+},
        \]
        where $T = O(\abs{x}^2)$ and depends only on $\abs{x}$.
    \end{lemma}

    \begin{lemma} [Equality Checking] \label{prop-eq}
        There is a stationary oblivious reversible TM $M$ that can check whether the contents in the first and second tracks are equal and puts the outcome in the third track. Formally, for $x, y \in \{0,1\}^+$ with $\abs{x} = \abs{y}$,
        \[
            x;y;0 \xrightarrow[T]{M} \begin{cases}
            x;y;1 & x = y \\
            x;y;0 & x \neq y
            \end{cases},
        \]
        where $T = O(\abs{x}^2)$ and depends only on $\abs{x}$.
    \end{lemma}

    \begin{lemma} \label{prop-shift}
        There are two stationary reversible TMs $M_{\operatorname{shl}}$ and $M_{\operatorname{shr}}$ such that for every $x \in \{0,1\}^+$,
        \[
            x \xrightarrow[T]{M_{\operatorname{shl}}} \operatorname{shl} x\ {\rm and}\
            x \xrightarrow[T]{M_{\operatorname{shr}}} \operatorname{shr} x,
        \]
        where $T = O(\abs{x}^2)$ and depends only on $\abs{x}$. Here, we write $\operatorname{shl}: \Sigma^\# \to \Sigma^\#$ for ``shift left'' and $\operatorname{shr}: \Sigma^\# \to \Sigma^\#$ for ``shift right''; that is, $(\operatorname{shl} \mathcal{T})(m) = \mathcal{T}(m+1)$ and $(\operatorname{shr} \mathcal{T})(m) = \mathcal{T}(m-1)$ for every $\mathcal{T} \in \Sigma^\#$.
    \end{lemma}

    \subsubsection{Several Lemmas for QTMs}

    \begin{lemma} [Dovetailing Lemma, Lemma 4.9 of \cite{Ber97}] \label{prop-dovetail}
        For any two well-formed, normal form and stationary QTMs $M_1$ and $M_2$, there is a well-formed, normal form and stationary QTM $M$ such that
        \[
            \ket{\mathcal{T}_0} \xrightarrow[T_1]{M_1} \ket{\mathcal{T}_1} \xrightarrow[T_2]{M_2} \ket{\mathcal{T}_2} \Longrightarrow \ket{\mathcal{T}_0} \xrightarrow[T_1+T_2]{M} \ket{\mathcal{T}_2}.
        \]
    \end{lemma}

    \begin{lemma}[Unidirection Lemma, Lemma 5.5 of \cite{Ber97}] \label{lemma-unidirectional}
        For every QTM $M = (Q, \Sigma, \delta, q_0, q_f)$ with time evolution operator $U$, there is a QTM $M' = (Q', \Sigma, \delta', q_0, q_f)$ with $Q \subseteq Q'$ and time evolution operator $U'$ such that for every $q \in Q \setminus \{q_f\}, \mathcal{T} \in \Sigma^\#$ and $\xi \in \mathbb{Z}$, we have
        \[
            U\ket{q, \mathcal{T}, \xi} = U' (P_F^\perp U')^4 \ket{q, \mathcal{T}, \xi},
        \]
        where $P_F^\perp = I - P_F$ and $P_F = \ket{q_f}_Q\bra{q_f}$. Moreover, if $M$ is well-formed, then so is $M'$.
    \end{lemma}

    Intuitively, the above lemma shows that any (well-formed) QTM can be converted to a (well-formed) unidirectional QTM with slowdown by a factor of $5$.

    \subsection{Proof of Theorem \ref{thm-qtm}}

Now we are ready to prove Theorem \ref{thm-qtm}. The proof is split into the following five steps:

{\vskip 3pt}

        \textbf{Step 1}. Let $\ell = \abs{T(\abs{x})}$ be the length of $T(\abs{x})$ in binary. By the definition of $T(n)$, there is a standard QTM $M_1$ such that
        \[
            x;\epsilon;\epsilon \xrightarrow[O(T(\abs{x}))]{M_1} x;{T(\abs{x})};{0^{\ell}}.
        \]
        It can be easily obtained by binding all non-blank symbols in the second track with a $0$ symbol in the third track.

        {\vskip 3pt}

        \textbf{Step 2}. After Step 1, suppose the state is $q_1$ (and the head position is $0$), we construct a standard QTM $M_2$ that adds a single symbol $0$ into the fourth track by the following transitions:
        \[
        \begin{matrix}
            q_1, & (\forall_1, \forall_2, \forall_3, \#_4) & \to & (\forall_1, \forall_2, \forall_3, 0), & q_2, & L \\
            q_2, & (\forall_1, \forall_2, \forall_3, \forall_4) & \to & (\forall_1, \forall_2, \forall_3, \forall_4), & q_0, & R \\
        \end{matrix}
        \]
      Then we have:
        \[
            x;{T(\abs{x})};{0^{\ell}};\epsilon \xrightarrow[2]{M_2} x;{T(\abs{x})};{0^{\ell}};0.
        \]
        Now the preparation is completed, and we will start from the configuration $\ket{q_0, x;{T(\abs{x})};{0^{\ell}};0, 0}$.

        {\vskip 3pt}

        \textbf{Step 3}. Let $M$ be a QTM to be standardised. By Lemma \ref{lemma-unidirectional}, we may assume that $M$ is unidirectional. Let $d_q$ be the direction of $q$, and let $M_\text{shl}$ and $M_\text{shr}$ be the reversible TMs constructed in Lemma \ref{prop-shift}. We construct $M_3$ as follows. For every $p \in Q \setminus \{q_f\}$, $q \in Q$ and $\tau, \sigma \in \Sigma$ with $\delta(p, \tau, \sigma, q, d_q) \neq 0$,

    \textbf{Case 1}. $d_q = R$. $M_3$ should include these instructions:
    \[
    \begin{matrix}
        p, & (\tau, \forall_2, \forall_3, \forall_4) & \to & (\sigma, \forall_2, \forall_3, \forall_4), & (q, 1), & R; & \delta(p, \tau, \sigma, q, R) \\
        (q, 1), & (\forall_1, \forall_2, \forall_3, \forall_4) & \to & (\forall_1, \forall_2, \forall_3, \forall_4), & (q, 2), & L; & 1 \\
        (q, 2), & (\forall_1, \forall_2, \forall_3, \forall_4) & \to & (\forall_1, \forall_2, \forall_3, \forall_4), & (q, 3), & L; & 1 \\
        (q, 3), & (\forall_1, \#_2, \#_3, \forall_4) & \to & (\forall_1, \#_2, \#_3, \forall_4), & (q, 4), & R; & 1 \\
        (q, 4) & & \to & & (q, 5): & & M_\text{shr}[2, 3, 4] \\
        (q, 5), & (\forall_1, \#_2, \#_3, \forall_4) & \to & (\forall_1, \#_2, \#_3, \forall_4), & (q, 6), & R; & 1
    \end{matrix}
    \]
    $(q, 4)$ and $(q, 5)$ are regarded as the initial state and final state of $M_\text{shr}$, respectively, that shifts the second, third and fourth tracks right by a cell.

    \textbf{Case 2}. $d_q = L$. $M_3$ should include these instructions:
    \[
    \begin{matrix}
        p, & (\tau, \forall_2, \forall_3, \forall_4) & \to & (\sigma, \forall_2, \forall_3, \forall_4), & (q, 1), & L; & \delta(p, \tau, \sigma, q, L) \\
        (q, 1), & (\forall_1, \forall_2, \forall_3, \forall_4) & \to & (\forall_1, \forall_2, \forall_3, \forall_4), & (q, 2), & R; & 1 \\
        (q, 2) & & \to & & (q, 3): & & M_\text{shl}[2, 3, 4] \\
        (q, 3), & (\forall_1, \forall_2, \forall_3, \forall_4) & \to & (\forall_1, \forall_2, \forall_3, \forall_4), & (q, 4), & L; & 1 \\
        (q, 4), & (\forall_1, \forall_2, \forall_3, \forall_4) & \to & (\forall_1, \forall_2, \forall_3, \forall_4), & (q, 5), & L; & 1 \\
        (q, 5), & (\forall_1, \#_2, \#_3, \forall_4) & \to & (\forall_1, \#_2, \#_3, \forall_4), & (q, 6), & R; & 1
    \end{matrix}
    \]
    $(q, 2)$ and $(q, 3)$ are regarded as the initial state and final state of $M_\text{shl}$, respectively, that shifts the second, third and fourth tracks left by a cell.

    Now both cases are in state $(q, 6)$. Let $M_{\text{inc}}$ and $M_{\text{eq}}$ be the RTMs constructed in Lemma \ref{prop-inc} and Lemma \ref{prop-eq}, respectively. We include the instructions:
    \[
    \begin{matrix}
        (q, 6) & \to & (q, 7): & M_{\text{inc}}[3] \\
        (q, 7) & \to & (q, 8): & M_{\text{eq}}[2, 3, 4]
    \end{matrix}
    \]
    The procedure from $(q, 6)$ and $(q, 7)$ performs incrementing on the third track according to $M_{\text{inc}}$. The procedure from $(q, 7)$ to $(q, 8)$ performs equality checking on the second and third tracks and puts the result on the fourth track according to $M_{\text{eq}}$.
    To the end of the simulation at this step, we include these instructions:
    \[
    \begin{matrix}
        (q, 8), & (\forall_1, \forall_2, \forall_3, \forall_4) & \to & (\forall_1, \forall_2, \forall_3, \forall_4), & (q, 9), & L; & 1 \\
        (q, 9), & (\forall_1, \#_2, \#_3, \#_4) & \to & (\forall_1, \#_2, \#_3, \#_4), & q, & R; & 1
    \end{matrix}
    \]

    Note that for every $p \in Q \setminus \{q_f\}$, $\mathcal{T} \in \Sigma^\#$ and $\xi \in \mathbb{Z}$,
    \[
        \ket{p, \mathcal{T}, \xi} \xrightarrow[1]{M} \sum_{\sigma, q} \delta(p, \mathcal{T}(\xi), \sigma, q, d_q) \ket{q, \mathcal{T}_\xi^\sigma, \xi+d_q}.
    \]
    We conclude that for every $T, t \in \mathbb{Z}$ with $0 \leq t < T-1$,
    \[
        \ket{p, \mathcal{T};\operatorname{shr}^{\xi}\left({T};t;0\right), \xi} \xrightarrow[\Delta(\ell)]{M_3} \sum_{\sigma, q} \delta(p, \mathcal{T}(\xi), \sigma, q, d_q) \ket{q, \mathcal{T}_\xi^\sigma; \operatorname{shr}^{\xi+d_q} \left({T};{t+1};0\right), \xi+d_q}.
    \]
    For the case $t = T-1$,
    \[
        \ket{p, \mathcal{T};\operatorname{shr}^{\xi}\left({T};{T-1};0\right), \xi} \xrightarrow[\Delta(\ell)]{M_3} \sum_{\sigma, q} \delta(p, \mathcal{T}(\xi), \sigma, q, d_q) \ket{q, \mathcal{T}_\xi^\sigma; \operatorname{shr}^{\xi+d_q} \left({T};{T};1\right), \xi+d_q},
    \]
    where $\Delta(\ell) = T_{\text{sh}}(\ell)+T_{\text{inc}}(\ell)+T_{\text{eq}}(\ell)+7 = O(\log^2 T)$. Here, $\operatorname{shr}^k \mathcal{T}$ denotes the tape that shifts $\mathcal{T}$ right by $k$ steps, i.e. $(\operatorname{shr}^k \mathcal{T})(m) = \mathcal{T}(m-k)$.

    {\vskip 3pt}

    \textbf{Step 4}. QTM $M_4$ is constructed as follows. We introduce a special state $q_a \notin Q$ and set $q_f' \notin Q$ to be the final state of $M'$. Moreover, we need a fifth track and mark $@$ on the fifth track to distinguish the usual simulation and the extending procedure. For the final state $q_f$, we include the instructions:
    \[
    \begin{matrix}
        q_f, & (\forall_1, \forall_2, \forall_3, \forall_4, \#_5) & \to & (\forall_1, \forall_2, \forall_3, \forall_4, @), & (q_f, 1), & L; & 1 \\
        (q_f, 1), & (\forall_1, \#_2, \#_3, \#_4, \#_5) & \to & (\forall_1, \#_2, \#_3, \#_4, \#_5), & q_a, & R; & 1
    \end{matrix}
    \]
    It takes two steps to transfer state $q_f$ to state $q_a$ with a marker $@$ on the fifth track, i.e.
    \[
        \ket{q_f, \mathcal{T};\operatorname{shr}^{\xi} \left({T};t;{z}\right);\epsilon, \xi} \xrightarrow[2]{M_4} \ket{q_a, \mathcal{T};\operatorname{shr}^{\xi} \left({T};t;{z};{@}\right), \xi}
    \]
    for $0 \leq t \leq T$ and $z \in \{0, 1\}$.

    Now include the instructions of $q_a$ as follows:
    \[
    \begin{matrix}
        q_a, & (\forall_1, \forall_2, \forall_3, 1, @) & \to & (\forall_1, \forall_2, \forall_3, 1, @), & q_f', & R; & 1 \\
        q_a, & (\forall_1, \forall_2, \forall_3, 0, @) & \to & (\forall_1, \forall_2, \forall_3, 0, @), & (q_a, 1), & R; & 1 \\
        (q_a, 1), & (\forall_1, \forall_2, \forall_3, \forall_4, \#_5) & \to & (\forall_1, \forall_2, \forall_3, \forall_4, @), & (q_a, 2), & L; & 1 \\
        (q_a, 2), & (\forall_1, \forall_2, \forall_3, \forall_4, \forall_5) & \to & (\forall_1, \forall_2, \forall_3, \forall_4, \forall_5), & (q_a, 3), & L; & 1 \\
        (q_a, 3), & (\forall_1, \#_2, \#_3, \#_4, \forall_5) & \to & (\forall_1, \#_2, \#_3, \#_4, \forall_5), & (q_a, 4), & R; & 1 \\
        (q_a, 4) & & \to & & (q_a, 5): & & M_\text{shr}[2, 3, 4] \\
        (q_a, 5), & (\forall_1, \#_2, \#_3, \#_4, \forall_5) & \to & (\forall_1, \#_2, \#_3, \#_4, \forall_5), & (q_a, 6), & R; & 1 \\
        (q_a, 6) & & \to & & (q_a, 7): & & M_\text{inc}[3] \\
        (q_a, 7) & & \to & & (q_a, 8): & & M_\text{eq}[2, 3, 4] \\
        (q_a, 8), & (\forall_1, \forall_2, \forall_3, \forall_4, \forall_5) & \to & (\forall_1, \forall_2, \forall_3, \forall_4, \forall_5), & (q_a, 9), & L; & 1 \\
        (q_a, 9), & (\forall_1, \forall_2, \forall_3, \forall_4, @) & \to & (\forall_1, \forall_2, \forall_3, \forall_4, @), & q_a, & R; & 1 \\
    \end{matrix}
    \]

    We conclude that for $\eta \leq \xi$ and $0 \leq t < T-1$,
    \[
    \ket{q_a, \mathcal{T};\operatorname{shr}^{\xi} \left({T};t;{0}\right);\operatorname{shr}^{\eta} {@^{\xi-\eta+1}}, \xi} \xrightarrow[\Delta(\ell)]{M_4} \ket{q_a, \mathcal{T};\operatorname{shr}^{\xi+1} \left({T};{t+1};{0}\right);\operatorname{shr}^{\eta} {@^{\xi-\eta+2}}, \xi+1}.
    \]
    For the case $t = T-1$, we have
    \[
    \ket{q_a, \mathcal{T};\operatorname{shr}^{\xi} \left({T};{T-1};{0}\right);\operatorname{shr}^{\eta} {@^{\xi-\eta+1}}, \xi} \xrightarrow[\Delta(\ell)]{M_4} \ket{q_a, \mathcal{T};\operatorname{shr}^{\xi+1} \left({T};{T};{1}\right);\operatorname{shr}^{\eta} {@^{\xi-\eta+2}}, \xi+1}.
    \]
    And the case $t = T$,
    \[
    \ket{q_a, \mathcal{T};\operatorname{shr}^{\xi} \left({T};{T};{1}\right);\operatorname{shr}^{\eta} {@^{\xi-\eta+1}}, \xi} \xrightarrow[1]{M_4} \ket{q_f', \mathcal{T};\operatorname{shr}^{\xi} \left({T};{T};{1}\right);\operatorname{shr}^{\eta} {@^{\xi-\eta+1}}, \xi+1}.
    \]

    By Lemma \ref{prop-dovetail}, dovetailing the four QTMs $M_1, M_2, M_3$ and $M_4$ will obtain a well-formed, normal form and unidirectional but not stationary QTM $M'$. Since the contents of fifth track allows distinguishing the result obtained at any time the state $\ket{q_f}$ is measured during the execution by the number of $@$ in the fifth track, it can be verified that $M'$ satisfies the condition claimed in the theorem statement (except for that $M'$ is not stationary).

    {\vskip 3pt}

    \textbf{Step 5}. This step fills meaningless instructions, which will not modify the contents of the first track, in order to make $M'$ stationary. We need three time stamps $T_1, T_2$ and $T_3$ with $T < T_1 < T_2 < T_3 = O(T)$. In the construction of Step 4, when a single symbol $1$ is found in the fourth track, we still move right (and print $@$ on the fifth track) until the time accumulator (i.e. the content of the third track) reaches $T_1$ instead of stopping at state $q_f'$. We may set $T_1$ sufficiently large, e.g. $T_1 = 4T$, in order to guarantee that the head position at time $T_1$ is positive. Let $\xi_1$ be the head position at time $T_1$. It should be noted that in order to obtain the current head position, we can use a new track to maintain the current head position in the simulation of each step by $M_\text{inc}$ and $M_\text{dec}$ constructed in Lemma \ref{prop-inc}.

    Our strategy is to move right until the time accumulator reaches $T_2$ and then move left until the time accumulator reaches $T_3$ such that the head position at time $T_3$ is $0$. Note that the head position will be $\xi_1+T_2-T_1$ at time $T_2$. In order to make the head position at time $T_3$ to be $0$, the condition $T_3-T_2 = \xi_1+T_2-T_1$ should hold, i.e. $T_2 = (T_1+T_3-\xi_1)/2$, which allows to compute $T_2$ when the time accumulator reaches $T_1$ (and then $\xi_1$ is known). In order to make the time evolution unitary, we need three more tracks to print symbol $@$ for $T_1, T_2$ and $T_3$ (similar to Step 4).

    We can set appropriate values for $T_1$ and $T_3$ to achieve these, for example $T_1 = 4T$ and $T_3 = 10T$.
    To see that the constructed QTM is stationary, we give an intuitive explanation here. Let $\tau_x$ be the running time of $M$ on input $x$ and $\xi_x$ be the head position of $M$ at time $\tau_x$, which is also the head position of $M'$ when the state $q_a$ is met for the first time. We have $\abs{\xi_x} \leq \tau_x \leq T = T(\abs{x})$. After that, our QTM $M'$ has four procedures:

    {\vskip 3pt}

    \textbf{Procedure 0}. In the simulation of $M$ for time stamp ranged from $\tau_x$ to $T$, the head position keeps going right. A symbol $@$ is printed on each position between $\xi_x$ and $\xi_0$ of the fifth track, where $\xi_0 - \xi_x = T-\tau_x$. We call the fifth track the 0th buffer track.

    {\vskip 3pt}

    \textbf{Procedure 1}. After Procedure 0, the head position keeps going right. Another (empty) track, called the 1st buffer track, is used to print a symbol $@$ on each position between $\xi_0$ and $\xi_1$, where $\xi_1 - \xi_0 = T_1 - T$.

    {\vskip 3pt}

    \textbf{Procedure 2}. After Procedure 1, the head position keeps going right. Another (empty) track, called the 2nd buffer track, is used to print a symbol $@$ on each position between $\xi_1$ and $\xi_2$, where $\xi_2 - \xi_1 = T_2 - T_1$ and $T_2 = (T_1+T_3-\xi_1)/2$.

    {\vskip 3pt}

    \textbf{Procedure 3}. After Procedure 2, the head position keeps going left. Another (empty) track, called the 3rd buffer track, is used to print a symbol $@$ on each position between $\xi_3$ and $\xi_2$, where $\xi_2 - \xi_3 = T_3 - T_2$.

    Note that $T_1$ and $\xi_1$ always have the same parity, and $T_3 = 10T$ is even, we conclude that $T_2$ is always integer. Instead, $T_2 \geq T_3/2 = 5T$ and $T_2 \leq (T_1+T_3)/2 = 7T$. Therefore, it holds that $$T < T_1 = 4T < 5T \leq  T_2 \leq 7T < T_3 = 10T.$$ On the other hand, we have: $$\xi_3 = \xi_2+T_2-T_3 = T_2-T_1+\xi_1+T_2-T_3 = 0,$$ which implies QTM $M'$ is stationary.
    In the simulation of each step of the four procedures, printing each symbol $@$ (except the first printed symbol at each procedure) takes exactly $\Delta(\ell)$ steps (see the construction in Step 4). Therefore, $M'$ halts exactly at time $T' = O(T_3 \Delta(\ell)) = O(T \log^2 T)$, and $M'$ is a standard QTM that simulates $M$.

     \section{Conclusion} \label{sec:conclusion}

    In this paper, we formally define the notions of QRAM and QRASP. The relationships between the computational powers of QRAMs, QRASPs and QTMs are established by overcoming the difficulty of mismatch between the halting scheme of QTMs and that of QRAMs and QRASPs through a technique for standardisation of QTMs. These results further help us to clarify the relationships between complexity classes \textbf{P}, \textbf{EQRAMP}, \textbf{EQP}, \textbf{BQRAMP} and \textbf{BQP}.

    The models of QRAMs and QRASPs defined in this paper can be further extended in several dimensions:
    \begin{itemize}
      \item The addressing adopted in our QRAM model is classical in the sense that an address indicating which quantum register (qubit) to perform a quantum gate is obtained from a classical register.
   This perfectly match the current architecture design of quantum computer, like a co-processor used together with a classical computer. But one can also conceive a fully quantum computer in the future which utilises only quantum registers but no classical registers. Such a machine should allow to access simultaneously the data of several different registers via a superposition of addresses.
        Indeed, such a notion of quantum addressing was already introduced in the quantum random access memory model \cite{Gio08}, and a possible quantum optical implementation is also proposed there.
        A model of QRAMs with quantum addressing is certainly an interesting topic for future research.

      \item In QRASPs considered in this paper, a program is stored in classical registers, and thus treated as classical data rather than quantum data. For a QRASP modelling a fully quantum computer, however, a program will be encoded as quantum data. Consequently, the quantum programming paradigm of superposition of programs proposed in \cite{Ying16} can be realised in such a generalised QRASP model.

      \item Several new parallel quantum algorithms or parallel implementation of existing quantum algorithms have been developed, e.g. \cite{Bra18, Cleve, Moore}. On the other hand, a parallel quantum programming language was defined in \cite{Yin18, YZLF22}.
       This motivates us to extend our QRAM and QRASP models to parallel quantum random access machines (PQRAMs), as a quantum generalisation of PRAMs \cite{Karp}.
    \end{itemize}

\appendix

\renewcommand{\thetheorem}{\Alph{section}.\arabic{theorem}}
\renewcommand{\thesection}{\Alph{section}}

\section{Proofs of Address Shifting and Address-Safe QRAMs} \label{sec:proofs-of-qram-address}

\subsection*{Proof of Lemma \ref{lemma-qram-address-shifting}}

        Suppose $P$ consists of $L$ instructions $P_0, P_1, \dots, P_{L-1}$.
Let $\delta = k+1$. In the following, we shift the address to the right by $\delta$ with the help of $X_0$. The modified instructions for address shifting are listed in Table \ref{tab-shift-address}. Most of them are directly obtained except indirect addressing and jumping.
        \begin{table}[!htp]
            \centering
            \caption{Modified QRAM instructions for address shifting by $\delta$}
            \begin{tabular}{lll}
            \hline
            Type & Instruction & Modified instruction \\
            \hline
            Classical & $X_i \gets C$, $C$ any integer & $X_{i+\delta} \gets C$ \\
            Classical & $X_i \gets X_j+X_k$ & $X_{i+\delta} \gets X_{j+\delta}+X_{k+\delta}$ \\
            Classical & $X_i \gets X_j-X_k$ & $X_{i+\delta} \gets X_{j+\delta}-X_{k+\delta}$ \\
            Classical & $X_i \gets X_{X_j}$ & $X_{i+\delta} \gets X_{X_{j+\delta}+\delta}$ \\
            Classical & $X_{X_i} \gets X_j$ & $X_{X_{i+\delta}+\delta} \gets X_{j+\delta}$ \\
            Classical & TRA $m$ if $X_j > 0$ & TRA $m'$ if $X_{j+\delta} > 0$ \\
            Classical & READ $X_i$ & READ $X_{i+\delta}$ \\
            Classical & WRITE $X_i$ & WRITE $X_{i+\delta}$ \\
            \hline
            Quantum & $\text{CNOT}[Q_{X_i}, Q_{X_j}]$ & $\text{CNOT}[Q_{X_{i+\delta}}, Q_{X_{j+\delta}}]$ \\
            Quantum & $H[Q_{X_i}]$ & $H[Q_{X_{i+\delta}}]$ \\
            Quantum & $T[Q_{X_i}]$ & $T[Q_{X_{i+\delta}}]$ \\
            \hline
            Measurement & $X_i \gets M[Q_{X_j}]$ & $X_{i+\delta} \gets M[Q_{X_{j+\delta}}]$ \\
            \hline
            \end{tabular}
            \label{tab-shift-address}
        \end{table}

        In order to precisely describe how to make this shifting, we first list the lengths needed for all instructions in Table \ref{tab-readdress}. For $0 \leq l < L$, we write $\mathit{length}(l)$ to denote the length needed for address shifting according to Table \ref{tab-readdress}. In order to label the instructions in $P'$, we define:
        \[
            \mathit{label}(l) = \sum_{i=0}^{l-1} \mathit{length}(i)
        \]
        for $0 \leq l \leq L$. Especially, the length of $P'$ is defined to be $L' = \mathit{label}(L)$.

        \begin{table}[!htp]
            \centering
            \caption{Lengths of QRAM instructions for address shifting}
            \begin{tabular}{llc}
            \hline
            Type & Instruction & length \\
            \hline
            Classical & $X_i \gets C$, $C$ any integer & $1$ \\
            Classical & $X_i \gets X_j+X_k$ & $1$ \\
            Classical & $X_i \gets X_j-X_k$ & $1$ \\
            Classical & $X_i \gets X_{X_j}$ & $6$ \\
            Classical & $X_{X_i} \gets X_j$ & $6$ \\
            Classical & TRA $m$ if $X_j > 0$ & $1$ \\
            Classical & READ $X_i$ & $1$ \\
            Classical & WRITE $X_i$ & $1$ \\
            \hline
            Quantum & $\text{CNOT}[Q_{X_i}, Q_{X_j}]$ & $1$ \\
            Quantum & $H[Q_{X_i}]$ & $1$ \\
            Quantum & $T[Q_{X_i}]$ & $1$ \\
            \hline
            Measurement & $X_i \gets M[Q_{X_j}]$ & $1$ \\
            \hline
            \end{tabular}
            \label{tab-readdress}
        \end{table}

        Now we are ready to describe how to construct $P'$. For every $0 \leq l < L$, we convert $P_l$ to one or more instructions in $P'$.

        {\vskip 3pt}

        \textbf{Case 1}. If $P_l$ is indirect addressing, the two modified instructions for $X_i \gets X_{X_j}$ and $X_{X_i} \gets X_j$ are indeed problematic in Table \ref{tab-shift-address}. To resolve this issue, we consider $X_i \gets X_{X_j}$ for example and use the following instructions with the help of $X_0$:
        \begin{align*}
          \mathit{label}(l): & X_0 \gets 0 \\
          & X_0 \gets X_0 - X_{j+\delta} \\
          & \text{TRA } L' \text{ if } X_0 > 0  \\
          & X_0 \gets \delta \\
          & X_0 \gets X_0 + X_{j+\delta} \\
          & X_{i+\delta} \gets X_{X_0}
        \end{align*}
Similarly, the instructions for $X_{X_i} \gets X_j$ are as follows:
        \begin{align*}
          \mathit{label}(l): & X_0 \gets 0 \\
          & X_0 \gets X_0 - X_{i+\delta} \\
          & \text{TRA } L' \text{ if } X_0 > 0  \\
          & X_0 \gets \delta \\
          & X_0 \gets X_0 + X_{i+\delta} \\
          & X_{X_0} \gets X_{j+\delta}
        \end{align*}

        \textbf{Case 2}. If $P_l$ is jumping, i.e. TRA $m$ if $X_j > 0$, we use a single modified instruction:
        \begin{align*}
          \mathit{label}(l): & \text{TRA } m' \text{ if } X_{j+\delta} > 0
        \end{align*}
        with $m' = \mathit{label}(m)$.

  {\vskip 3pt}

        \textbf{Case 3}. For other cases, use the instructions according to Table \ref{tab-readdress}.

  {\vskip 3pt}

        It can be seen that the constructed QRAM $P'$ can  simulate QRAM $P$ through shifting the address to the right by $\delta = k+1$, with $X_1, X_2, \dots, X_k$ untouched. Instead, the slowdown is a constant factor, which depends on $k$.

\subsection*{Proof of Lemma \ref{lemma-qram-address}}

        The construction of $P'$ is straightforward. Note that the only way to access to an invalid address is indirect addressing. Thus, we could avoid accessing to an invalid address by checking whether the indirect address is valid beforehand. For example, if an instruction $P_l$ is of the form $X_i \gets X_{X_j}$, then it can be replaced in $P'$ by introducing an independent variable $\mathit{tmp}$ with the code
        \begin{align*}
            & \mathit{tmp} \gets 0 \\
            & \mathit{tmp} \gets \mathit{tmp}-X_j \\
            & \text{TRA } L' \text{ if } \mathit{tmp}>0 \\
            & X_i \gets X_{X_j}
        \end{align*}
        where $L'$ denotes the length of $P'$. Instead, the value of $m$ in the jumping instruction should be changed to an appropriate value $m'$ similar to Lemma \ref{lemma-qram-address-shifting}.

        According to Lemma \ref{lemma-qram-address-shifting}, the variable $\mathit{tmp}$ is involved by shifting the address to the right by $\delta = 2$.

\section{QRAM instructions for simulating QRASPs}\label{qram-to-qrasp}

To obtain a program $P'$ consisting of only QRAM instructions from the pseudo-code given in Algorithm \ref{algo1}, we have to:
    \begin{enumerate}
      \item transfer the \textbf{if} and \textbf{while} statements to QRAM instructions; and
      \item replace every variable by a classical register with an explicit index.
    \end{enumerate} We will carefully describe how to transform the pseudo-code given in Algorithm \ref{algo1} into QRAM $P'$ in Appendices \ref{e-check}-\ref{sec:assert-valid-addr}.

    \subsection{Equality Checking}\label{e-check}

    Given two classical registers $a$ and $b$, it is a basic operation to check whether their contents are equal: $a = b$?
    We observe that $a \neq b$ if and only if $\abs{a-b} > 0$. Algorithm \ref{algo2} provides a simple method to compare $a$ and $b$ using three extra disposable registers, with the result $\mathit{res} = \abs{a-b}$.

    To simplify the QRAM code, we use $\mathit{res} \gets \abs{a-b}$ to indicate the code in Algorithm \ref{algo2} in the following discussions.

    \begin{algorithm}[!htp]
        \caption{QRAM code for checking whether $a = b$.}
        \label{algo2}
        \begin{algorithmic}[1]
        \Require $a$ and $b$.
        \Ensure $res > 0$ if $a \neq b$ and $0$ otherwise.

        \State $\mathit{tmp0} \gets a$; $\mathit{tmp1} \gets b$;
        \State $\mathit{tmp0} \gets \mathit{tmp0}-\mathit{tmp1}$;
        \State TRA 6 if $\mathit{tmp0} > 0$;
        \State $\mathit{tmp1} \gets 0$;
        \State $\mathit{tmp0} \gets \mathit{tmp1}-\mathit{tmp0}$;
        \State $\mathit{res} \gets \mathit{tmp0}$;

        \end{algorithmic}
    \end{algorithm}

    \subsection{Encoding the \textbf{if} and \textbf{while} statements by QRAM instructions}

    We interpret \textbf{if} and \textbf{while} statements by QRAM instructions in the general case separately.

    For \textbf{if} statement, e.g. Algorithm \ref{algo3}, we provide a QRAM interpretation in Algorithm \ref{algo4}.

    \begin{algorithm}[!htp]
        \caption{Example code for if $a = b$.}
        \label{algo3}
        \begin{algorithmic}[1]

        \If {$a = b$}
            \State label0: statements;
        \Else
            \State label1: statements;
        \EndIf
        \State label2: statements;

        \end{algorithmic}
    \end{algorithm}

    \begin{algorithm}[!htp]
        \caption{QRAM code for if $a = b$.}
        \label{algo4}
        \begin{algorithmic}[1]

        \State $\mathit{res} \gets \abs{a-b}$;
        \State TRA label1 if $\mathit{res} > 0$;
        \State label0: statements;
        \State $\mathit{res} \gets 1$;
        \State TRA label2 if $\mathit{res} > 0$;
        \State label1: statements;
        \State label2: statements;

        \end{algorithmic}
    \end{algorithm}

    For \textbf{while} statement, e.g. Algorithm \ref{algo5}, we provide a QRAM interpretation in Algorithm \ref{algo6}.

    \begin{algorithm}[!htp]
        \caption{Example code for while $a = b$.}
        \label{algo5}
        \begin{algorithmic}[1]

        \While {$a = b$}
            \State label1: statements;
        \EndWhile
        \State label2: statements;

        \end{algorithmic}
    \end{algorithm}

    \begin{algorithm}[!htp]
        \caption{QRAM code for while $a = b$.}
        \label{algo6}
        \begin{algorithmic}[1]

        \State label0: $\mathit{res} \gets \abs{a-b}$;
        \State TRA label2 if $\mathit{res} > 0$;
        \State label1: statements;
        \State $\mathit{res} \gets 1$;
        \State TRA label0 if $\mathit{res} > 0$;
        \State label2: statements;

        \end{algorithmic}
    \end{algorithm}

    \subsection{Replacing every variable by a classical register with an explicit index}

    In the previous interpretation, there are only three extra classical registers used, namely $\mathit{tmp0}, \mathit{tmp1}$ and $\mathit{res}$. We assign each of the nine classical registers $\mathit{tmp0}, \mathit{tmp1}, \mathit{res}, \text{IC}, \text{AC}, \mathit{flag}, \mathit{op}, j, k$ from the $0$-th to the $8$-th classical registers, respectively. Let $\delta = 9$ indicate the offset. Then the array $\mathit{memory}$ is assigned to begin at the $\delta$-th classical register. More precisely, $\mathit{memory}[j]$ is assigned to the $(\delta+j)$-th classical register.

    \subsection{Assertion for valid addressing} \label{sec:assert-valid-addr}

    In the QRAM construction, accessing to $\mathit{memory}[j]$ is dangerous, because $j$ can be negative but the address it is assigned to, i.e. $(\delta+j)$, could still be valid. Therefore, an assertion is needed before each access to $\mathit{memory}[j]$. Algorithm \ref{algo7} provides a possible solution. We use a QRAM instruction trick here: before accessing to $X_{j+\delta}$, we try to access to $X_j$ (in QRAM address) but ignore the addressing results. This works because if $j < 0$, $X_j$ will be invalid, then the QRAM terminates as we want; if $j \geq 0$, it goes as if nothing happened (we have accessed to $X_j$ without modifying anything).

    \begin{algorithm}[!htp]
        \caption{QRAM code for accessing $\mathit{memory}[j]$.}
        \label{algo7}
        \begin{algorithmic}[1]

        \State $\mathit{tmp1} \gets X_{j}$;
        \State $\mathit{tmp0} \gets \delta$
        \State $j \gets j+\mathit{tmp0}$;
        \State $\mathit{res} \gets X_{j}$;

        \end{algorithmic}
    \end{algorithm}

    \section{QRASP instructions for simulating QRAMs}\label{qrasp-to-qram}

     \begin{enumerate}
      \item $P_l$ is of the form $X_i \gets C$. The QRASP code is
          \begin{align*}
              \mathit{label}(l): & \text{LOD}, C \\
              & \text{STO}, i+\delta
          \end{align*}
      \item $P_l$ is of the form $X_i \gets X_j+X_k$. The QRASP code is
          \begin{align*}
              \mathit{label}(l): & \text{LOD}, 0 \\
              & \text{ADD}, j+\delta \\
              & \text{ADD}, k+\delta \\
              & \text{STO}, i+\delta \\
          \end{align*}
      \item $P_l$ is of the form $X_i \gets X_j-X_k$. The QRASP code is
          \begin{align*}
              \mathit{label}(l): & \text{LOD}, 0 \\
              & \text{ADD}, j+\delta \\
              & \text{SUB}, k+\delta \\
              & \text{STO}, i+\delta \\
          \end{align*}
      \item $P_l$ is of the form $X_i \gets X_{X_j}$. The QRASP code is
          \begin{align*}
              \mathit{label}(l): & \text{LOD}, \delta \\
              & \text{ADD}, j+\delta \\
              & \text{STO}, a+1 \\
              & \text{LOD}, 0 \\
              a: & \text{ADD}, 0 \\
              & \text{STO}, i+\delta \\
          \end{align*}
          It is noted that $a = \mathit{label}(l)+8$.
      \item $P_l$ is of the form $X_{X_i} \gets X_j$. The QRASP code is
          \begin{align*}
              \mathit{label}(l): & \text{LOD}, \delta \\
              & \text{ADD}, i+\delta \\
              & \text{STO}, a+1 \\
              & \text{LOD}, 0 \\
              & \text{ADD}, j+\delta \\
              a: & \text{STO}, 0 \\
          \end{align*}
          It is noted that $a = \mathit{label}(l)+10$.
      \item $P_l$ is of the form TRA $m$ if $X_j > 0$. The QRASP code is
          \begin{align*}
              \mathit{label}(l): & \text{LOD}, 0 \\
              & \text{ADD}, j+\delta \\
              & \text{BPA}, \mathit{label}(m) \\
          \end{align*}
      \item $P_l$ is of the form READ $X_i$. The QRASP code is
          \begin{align*}
              \mathit{label}(l): & \text{RD}, i+\delta \\
          \end{align*}
      \item $P_l$ is of the form WRITE $X_i$. The QRASP code is
          \begin{align*}
              \mathit{label}(l): & \text{PRI}, i+\delta \\
          \end{align*}
     \end{enumerate}

     \section{Correctness of simulating QRAMs by QRASPs} \label{sec:correct-simulate-qrams-by-qrasps}

     We note that $L' = \abs{P'} = \mathit{label}(L) = \mathit{label}(\abs{P}) \leq 15L < 20L = \delta$.

  \begin{definition}  We say that a QRASP configuration $c' = (\xi', \zeta', \mu', \ket{\psi'}, x', y')$ agrees with a QRAM configuration $c = (\xi, \mu, \ket\psi, x, y)$, denoted $c' \models c$, if
    \begin{enumerate}
    \item $\xi' = \downarrow$, $\mu'(i+\delta) = \mu(i)$ for every $i \in \mathbb{N}$, $\ket{\psi'} = \ket\psi$, $x'=x$ and $y' = y$ in the case $\xi = \downarrow$; or
    \item $\xi' = \mathit{label}(\xi)$, $\mu'(i+\delta) = \mu(i)$ for every $i \in \mathbb{N}$, $\ket{\psi'} = \ket\psi$, $x'=x$ and $y' = y$ in the case $\xi \in \mathbb{N}$.
    \end{enumerate}\end{definition}

    \begin{lemma} \label{lemma-qram-by-qrasp-unique}
        For every $c'$, there is a unique $c$ such that $c' \models c$.
    \end{lemma}

    Let $\mathcal{C}$ and $\mathcal{C}'$ be the set of configurations of $P$ and $P'$, respectively, and $c_0 \in \mathcal{C}$ and $c'_0 \in \mathcal{C}'$ be their initial configurations.
    Define $\mathcal{C}'_{\mathcal{L}} \subseteq \mathcal{C}'$ being the set of configurations, whose ICs are in
    \[
    \mathcal{L} = \{ \mathit{label}(0), \mathit{label}(1), \dots, \mathit{label}(L) \}.
    \]

    \begin{lemma} \label{lemma-qram-by-qrasp-simulate-single-step}
        Let $c' \in \mathcal{C}'_{\mathcal{L}}$ and $c, d \in \mathcal{C}$. If $c' \models c$, and $c \xrightarrow[T]{p} d$,
        then there is a $d' \in \mathcal{C}'_{\mathcal{L}}$ such that $d' \models d$ and $c' \xrightarrow[\Theta(T)]{p} d'$.
    \end{lemma}
    \begin{proof}
        Direct from the operational semantics.
    \end{proof}

    We write $\mathcal{P}$ and $\mathcal{P}'$ for the sets of all possible execution paths of $P$ and $P'$, respectively.
    Let $\pi' \in \mathcal{P}'_f(c_0')$ be a path of length $\abs{\pi'} = k$:
    \[
        \pi': c_0' \xrightarrow[T_1']{p_1'} c_1' \xrightarrow[T_2']{p_2'} \dots \xrightarrow[T_{k-1}']{p_{k-1}'} c_{k-1}' \xrightarrow[T_k']{p_k'} c_k'.
    \]
    Let $0 = i_0 < i_1 < \dots < i_{m-1} < i_m = k$ be all indices such that $c'_{(j)} = c'_{i_j} \in \mathcal{C}'_{\mathcal{L}}$ for $0 \leq j \leq m$. Then the execution path $\pi'$ can be written as
    \[
        \pi': c_0' = c'_{(0)} \xrightarrow[T'_{(1)}]{p'_{(1)}}^* c'_{(1)} \xrightarrow[T'_{(2)}]{p'_{(2)}}^* \dots \xrightarrow[T'_{(m-1)}]{p'_{(m-1)}}^* c'_{(m-1)} \xrightarrow[T'_{(m)}]{p'_{(m)}}^* c'_{(m)} = c'_k,
    \]
    where
    \[
    p'_{(j)} = \prod_{l=i_{j-1}+1}^{i_j} p'_{l},
    \]
    and
    \[
    T'_{(j)} = \sum_{l=i_{j-1}+1}^{i_j} T'_{l}
    \]
    for $1 \leq j \leq m$. We write $\Abs{\pi'} = m$.

    \begin{lemma}
        $c'_{(0)} \models c_0$.
    \end{lemma}

    \begin{definition}
        Let $\pi' \in \mathcal{P}'_f(c'_{(0)})$ and $\pi \in \mathcal{P}_f(c_0)$.
        Then we say that $\pi'$ agrees with $\pi$, denoted $\pi' \models \pi$, if
        \begin{enumerate}
          \item $\Abs{\pi'} = \abs{\pi}$.
          \item $c'_{(j)} \models c_j$ for $0 \leq j \leq \Abs{\pi'}$.
          \item $p'_{(j)} = p_j$ and $T'_{(j)} = \Theta(T_j)$ for $1 \leq j \leq \Abs{\pi'}$.
        \end{enumerate}
    \end{definition}

    \begin{lemma} \label{lemma-qram-by-qrasp-time-bound}
        For every $\pi' \in \mathcal{P}'_f(c'_{(0)})$, there is a unique $\pi \in \mathcal{P}_f(c_0)$ such that $\pi' \models \pi$.
    \end{lemma}

    It follows immediately from Lemma \ref{lemma-qram-by-qrasp-time-bound} that $P'$ is time bounded by $O(T(n))$. Let $T': \mathbb{N} \to \mathbb{N}$ be the worst case running time of $P'$.

    \begin{lemma}
        For every $\pi \in \mathcal{P}_f(c_0)$,
        there is a unique $\pi' \in \mathcal{P}'_f(c'_{(0)})$ such that $\pi' \models \pi$.
        We use $h: \mathcal{P}(c_0) \to \mathcal{P}'(c'_{(0)})$ to denote this bijection.
    \end{lemma}
    \begin{proof}
        (\textbf{Existence}) Direct from Lemma \ref{lemma-qram-by-qrasp-simulate-single-step}.

        (\textbf{Uniqueness}) For every $\pi \in \mathcal{P}(c_0)$, we choose an arbitrary $\pi' \in \mathcal{P}'(c'_{(0)})$ such that $\pi' \models \pi$ and define $h(\pi) = \pi'$. Then
        \[
            1 = \sum_{\pi \in \mathcal{P}(c_0)} \pi.p = \sum_{\pi \in \mathcal{P}(c_0)} h(\pi).p \leq \sum_{\pi' \in \mathcal{P}'(c'_{(0)})} \pi'.p = 1.
        \]
        The uniqueness of $h(\pi)$ follows.
    \end{proof}

    Now let $x \in \Sigma^*$ be the input string and the initial configuration of $P$ is $c_0 = (0, \mu_0, \ket{\psi_0}, \mathit{in}(x), \epsilon)$. Since $P$ is a $T(n)$-time QRAM, $\abs{\pi} \leq T(\abs{x})$ is finite for every $\pi \in \mathcal{P}(c_0)$, and $\abs{\mathcal{P}(c_0)} \leq 2^{T(\abs{x})}$ is also finite because each transition leads to at most two branches. So for every $y \in \Sigma^*$, we have:
    \begin{align*}
        P(x, y)
        & = \sum_{c \in \mathcal{C}_f: \mathit{out}(c.y) = y} \llbracket P \rrbracket (c_0) (c) \\
        & = \sum_{c \in \mathcal{C}_f: \mathit{out}(c.y) = y} \llbracket P \rrbracket^{T(\abs{x})} (c_0) (c) \\
        & = \sum_{\pi \in \mathcal{P}^{T(\abs{x})}(c_0): \mathit{out}(\pi.c_{\abs{\pi}}.y) = y} \pi.p \\
        & = \sum_{\pi \in \mathcal{P}^{T(\abs{x})}(c_0): \mathit{out}(h(\pi).c_{\abs{f(\pi)}}.y) = y} h(\pi).p \\
        & = \sum_{\pi' \in {\mathcal{P}'}^{T'(\abs{x})}(c'_{(0)}): \mathit{out}(\pi'.c_{\abs{\pi'}}.y) = y} \pi'.p \\
        & = \sum_{\pi' \in {\mathcal{P}'}^{T'(\abs{x})}(c'_{0}): \mathit{out}(\pi'.c_{\abs{\pi'}}.y) = y} \pi'.p \\
        & = \sum_{c' \in \mathcal{C}'_f: \mathit{out}(c'.y) = y} \llbracket P' \rrbracket^{T'(\abs{x})} (c'_0) (c) \\
        & = \sum_{c' \in \mathcal{C}'_f: \mathit{out}(c'.y) = y} \llbracket P' \rrbracket (c'_0) (c) \\
        & = P'(x, y).
    \end{align*}

     \section{$\text{TM}^\text{Q}$  instructions for simulating QRAMs}\label{qtm-in-sim}

     \begin{enumerate}
      \item If $P_l$ has the form $X_i \gets C$, the following shows several steps to achieve the simulation with two work tracks work1 and work2.
          \begin{align*}
            (p_l, 0): & M_{\text{write($i$)}}[\text{work1}] \\
            (p_l, 1): & M_{\text{write($C$)}}[\text{work2}] \\
            (p_l, 2): & M_{\text{update}}[\text{creg}, \text{work1}, \text{work2}] \\
            (p_l, 3): & M_{\text{clean}}[\text{work1}] \\
            (p_l, 4): & M_{\text{clean}}[\text{work2}] \\
            (p_l, 5): & \text{transition to } (p_{l+1}, 0) \\
          \end{align*}

      \item If $P_l$ has the form $X_i \gets X_j+X_k$,
          \begin{align*}
            (p_l, 0): & M_{\text{write($j$)}}[\text{work1}] \\
            (p_l, 1): & M_{\text{fetch}}[\text{creg}, \text{work1}, \text{work2}] \\
            (p_l, 2): & M_{\text{write($k$)}}[\text{work3}] \\
            (p_l, 3): & M_{\text{fetch}}[\text{creg}, \text{work3}, \text{work4}] \\
            (p_l, 4): & M_{\text{add}}[\text{work2}, \text{work4}, \text{work5}] \\
            (p_l, 5): & M_{\text{write($i$)}}[\text{work6}] \\
            (p_l, 6): & M_{\text{update}}[\text{creg}, \text{work6}, \text{work5}] \\
            (p_l, 7): & M_{\text{clean}}[\text{work1}] \\
            (p_l, 8): & M_{\text{clean}}[\text{work2}] \\
            (p_l, 9): & M_{\text{clean}}[\text{work3}] \\
            (p_l, 10): & M_{\text{clean}}[\text{work4}] \\
            (p_l, 11): & M_{\text{clean}}[\text{work5}] \\
            (p_l, 12): & M_{\text{clean}}[\text{work6}] \\
            (p_l, 13): & \text{transition to } (p_{l+1}, 0) \\
          \end{align*}
      \item If $P_l$ has the form $X_i \gets X_j-X_k$,
          \begin{align*}
            (p_l, 0): & M_{\text{write($j$)}}[\text{work1}] \\
            (p_l, 1): & M_{\text{fetch}}[\text{creg}, \text{work1}, \text{work2}] \\
            (p_l, 2): & M_{\text{write($k$)}}[\text{work3}] \\
            (p_l, 3): & M_{\text{fetch}}[\text{creg}, \text{work3}, \text{work4}] \\
            (p_l, 4): & M_{\text{sub}}[\text{work2}, \text{work4}, \text{work5}] \\
            (p_l, 5): & M_{\text{write($i$)}}[\text{work6}] \\
            (p_l, 6): & M_{\text{update}}[\text{creg}, \text{work6}, \text{work5}] \\
            (p_l, 7): & M_{\text{clean}}[\text{work1}] \\
            (p_l, 8): & M_{\text{clean}}[\text{work2}] \\
            (p_l, 9): & M_{\text{clean}}[\text{work3}] \\
            (p_l, 10): & M_{\text{clean}}[\text{work4}] \\
            (p_l, 11): & M_{\text{clean}}[\text{work5}] \\
            (p_l, 12): & M_{\text{clean}}[\text{work6}] \\
            (p_l, 13): & \text{transition to } (p_{l+1}, 0) \\
          \end{align*}
      \item If $P_l$ has the form $X_i \gets X_{X_j}$,
          \begin{align*}
            (p_l, 0): & M_{\text{write($j$)}}[\text{work1}] \\
            (p_l, 1): & M_{\text{fetch}}[\text{creg}, \text{work1}, \text{work2}] \\
            (p_l, 2): & M_{\text{fetch}}[\text{creg}, \text{work2}, \text{work3}] \\
            (p_l, 3): & M_{\text{write($i$)}}[\text{work4}] \\
            (p_l, 4): & M_{\text{update}}[\text{creg}, \text{work4}, \text{work3}] \\
            (p_l, 5): & M_{\text{clean}}[\text{work1}] \\
            (p_l, 6): & M_{\text{clean}}[\text{work2}] \\
            (p_l, 7): & M_{\text{clean}}[\text{work3}] \\
            (p_l, 8): & M_{\text{clean}}[\text{work4}] \\
            (p_l, 9): & \text{transition to } (p_{l+1}, 0) \\
          \end{align*}
      \item If $P_l$ has the form $X_{X_i} \gets X_j$,
          \begin{align*}
            (p_l, 0): & M_{\text{write($j$)}}[\text{work1}] \\
            (p_l, 1): & M_{\text{fetch}}[\text{creg}, \text{work1}, \text{work2}] \\
            (p_l, 0): & M_{\text{write($i$)}}[\text{work3}] \\
            (p_l, 1): & M_{\text{fetch}}[\text{creg}, \text{work3}, \text{work4}] \\
            (p_l, 4): & M_{\text{update}}[\text{creg}, \text{work4}, \text{work2}] \\
            (p_l, 5): & M_{\text{clean}}[\text{work1}] \\
            (p_l, 6): & M_{\text{clean}}[\text{work2}] \\
            (p_l, 7): & M_{\text{clean}}[\text{work3}] \\
            (p_l, 8): & M_{\text{clean}}[\text{work4}] \\
            (p_l, 9): & \text{transition to } (p_{l+1}, 0) \\
          \end{align*}
      \item If $P_l$ has the form TRA $m$ if $X_j > 0$,
          \begin{align*}
            (p_l, 0): & M_{\text{write($j$)}}[\text{work1}] \\
            (p_l, 1): & M_{\text{fetch}}[\text{creg}, \text{work1}, \text{work2}] \\
            (p_l, 2): & M_{\text{gtz}}[\text{work2}] \\
            (p_l, 3): & M_{\text{clean}}[\text{work1}] \\
            (p_l, 4),~ & 0_\text{work2} \to \#_\text{work2}, (p_l, 5), L \\
            (p_l, 4),~ & 1_\text{work2} \to \#_\text{work2}, (p_l, 6), L \\
            (p_l, 5),~ & \#_\text{work2} \to \#_\text{work2}, (p_{l+1}, 0), R \\
            (p_l, 6),~ & \#_\text{work2} \to \#_\text{work2}, (p_{m}, 0), R \\
          \end{align*}
      \item If $P_l$ has the form READ $X_i$,
          \begin{align*}
            (p_l, 0): & M_{\text{read}}[\text{input}, \text{work1}] \\
            (p_l, 1): & M_{\text{write($i$)}}[\text{work2}] \\
            (p_l, 2): & M_{\text{update}}[\text{creg}, \text{work2}, \text{work1}] \\
            (p_l, 3): & M_{\text{clean}}[\text{work1}] \\
            (p_l, 4): & M_{\text{clean}}[\text{work2}] \\
            (p_l, 5): & \text{transition to } (p_{l+1}, 0) \\
          \end{align*}
      \item If $P_l$ has the form WRITE $X_i$,
          \begin{align*}
            (p_l, 0): & M_{\text{write($i$)}}[\text{work1}] \\
            (p_l, 1): & M_{\text{fetch}}[\text{creg}, \text{work1}, \text{work2}] \\
            (p_l, 2): & M_{\text{gtz}}[\text{work2}] \\
            (p_l, 3): & M_{\text{append}}[\text{output}, \text{work2}] \\
            (p_l, 4): & M_{\text{clean}}[\text{work1}] \\
            (p_l, 5): & M_{\text{clean}}[\text{work2}] \\
            (p_l, 6): & \text{transition to } (p_{l+1}, 0) \\
          \end{align*}
      \end{enumerate}

      \section{Details of Section \ref{properties-of-tms-and-qtms}} \label{appendix-details-tms-and-qtms}

      \subsection{Several Lemmas for QTMs}

      In order to prove our lemmas, we need some lemmas mainly from \cite{Ber97}.

      \begin{lemma} [Lemma B.7 of \cite{Ber97}] \label{prop-copy}
        There is a stationary oblivious reversible TM $M$ such that
        \[
            x; \epsilon \xrightarrow[T]{M} x; x\ {\rm
        and}\
            x; x \xrightarrow[T]{M} x; \epsilon,
        \]
        where $T = 2\abs{x}+4$.
    \end{lemma}

    \begin{lemma} [Lemma B.6 of \cite{Ber97}] \label{prop-swap}
        There is a stationary oblivious reversible TM $M$ such that
        \[
            x;y \xrightarrow[T]{M} y;x,
        \]
        where $T = 2\max\{\abs{x},\abs{y}\}+4$.
    \end{lemma}

    \begin{lemma} [Reversal Lemma, Lemma 4.12 of \cite{Ber97}] \label{prop-reversal}
        For every well-formed and stationary QTM $M$, there is a well-formed and stationary QTM $M'$ such that
        \[
            \ket{\mathcal{T}} \xrightarrow[T]{M} \ket{\mathcal{T}'} \Longrightarrow \ket{\mathcal{T}'} \xrightarrow[T+2]{M'} \ket{\mathcal{T}}.
        \]
    \end{lemma}

    \subsection{Proof of lemmas in Section \ref{properties-of-tms-and-qtms}}

    The following Lemmas \ref{prop-history}, \ref{prop-RTM} and \ref{prop-unique} are essentially Theorem B.8 and Theorem B.9 in \cite{Ber97}, but here they are slightly strengthened for our purpose.

    \begin{lemma} \label{prop-history}
        For every stationary deterministic TM $M$, there is a stationary reversible TM $M'$ such that
        \[
            \mathcal{T} \xrightarrow[T]{M} \mathcal{T}' \Longrightarrow \mathcal{T};\epsilon;\epsilon \xrightarrow[T']{M'} \mathcal{T}';@;\mathcal{T}^{h},
        \]
        where:
        \begin{enumerate}
          \item $\mathcal{T}^h = {\#\$(p_0,\sigma_0)\dots(p_{T-1}, \sigma_{T-1})}$ encodes the history, and $p_t$ and $\sigma_t$ denote the state and the symbol at the head position at time $t$ in the execution of $M$, respectively;
          \item $T' = O(T^2)$.
        \end{enumerate}
        Moreover, if $M$ is oblivious, then so is $M'$.
    \end{lemma}

    \begin{proof}
    \textbf{Step 1}. At the very beginning, we construct a RTM $M_1$ that writes an end marker $@$ on the second track and end marker $\$$ on the third track, and then comes back to the initial position with state $q_0$, by including these instructions:
        \[
        \begin{matrix}
            q_a, & (\forall_1, \#_2, \#_3) & \to & (\forall_1, @, \#), & q_b, & R \\
            q_b, & (\forall_1, \#_2, \#_3) & \to & (\forall_1, \#, \$), & q_c, & R \\
            q_c, & (\forall_1, \forall_2, \forall_3) & \to & (\forall_1, \forall_2, \forall_3), & q_d, & L \\
            q_d, & (\forall_1, \forall_2, \forall_3) & \to & (\forall_1, \forall_2, \forall_3), & q_e, & L \\
            q_e, & (\forall_1, \forall_2, \forall_3) & \to & (\forall_1, \forall_2, \forall_3), & q_g, & L \\
            q_g, & (\forall_1, \forall_2, \forall_3) & \to & (\forall_1, \forall_2, \forall_3), & q_0, & R \\
        \end{matrix}
        \]
        It is easy to see that
        \[
            \ket{q_a, \mathcal{T};\epsilon;\epsilon, 0} \xrightarrow[T_1]{M_1} \ket{q_0, \mathcal{T};@;{\#\$}, 0},
        \]
        where $T_1 = 6$.

        \textbf{Step 2}. We construct RTM $M_2$ as follows. For $p \in Q \setminus \{q_f\}$ and $\sigma \in \Sigma$ with transition $\delta(p, \sigma) = (\tau, q, d)$ in $M$, we make transitions to go from $p$ to $q$ updating the first track, i.e. the simulated tape of $M$, and adding $(p, \sigma)$ to the end of the history. Include these instructions:
        \[
        \begin{matrix}
            p, & (\sigma, @, \forall_3) & \to & (\tau, \#_2, \forall_3), & (q, p, \sigma, 1), & d \\
            (q, p, \sigma, 1), & (\forall_1, \#_2, \$) & \to & (\forall_1, @, \$), & (q, p, \sigma, 3), & R \\
            (q, p, \sigma, 1), & (\forall_1, \#_2, \forall_3') & \to & (\forall_1, @, \forall_3'), & (q, p, \sigma, 3), & R \\
            (q, p, \sigma, 1), & (\forall_1, \#_2, \#_3) & \to & (\forall_1, @, \#_3), & (q, p, \sigma, 2), & R \\
            (q, p, \sigma, 2), & (\forall_1, \#_2, \#_3) & \to & (\forall_1, \#_2, \#_3), & (q, p, \sigma, 2), & R \\
            (q, p, \sigma, 2), & (\forall_1, \#_2, \$) & \to & (\forall_1, \#_2, \$), & (q, p, \sigma, 3), & R \\
            (q, p, \sigma, 3), & (\forall_1, \#_2, \forall_3') & \to & (\forall_1, \#_2, \forall_3'), & (q, p, \sigma, 3), & R \\
            (q, p, \sigma, 3), & (\forall_1, \#_2, \#_3) & \to & (\forall_1, \#_2, (p, \sigma)), & (q, 4), & R \\
        \end{matrix}
        \]
        When $(q, 4)$ is reached, the tape head is on the first blank after the end of the history (on the third track). Now we move the tape head back to the position of tape head of $M$ by including these instructions:
        \[
        \begin{matrix}
            (q, 4), & (\forall_1, \#_2, \#_3) & \to & (\forall_1, \#_2, \#_3), & (q, 5), & L \\
            (q, 5), & (\forall_1, \#_2, \forall_3') & \to & (\forall_1, \#_2, \forall_3'), & (q, 5), & L \\
            (q, 5), & (\forall_1, \#_2, \$) & \to & (\forall_1, \#_2, \$), & (q, 6), & L \\
            (q, 5), & (\forall_1, @, \forall_3') & \to & (\forall_1, @, \forall_3'), & (q, 7), & L \\
            (q, 5), & (\forall_1, @, \$) & \to & (\forall_1, @, \$), & (q, 7), & L \\
            (q, 6), & (\forall_1, \#_2, \#_3) & \to & (\forall_1, \#_2, \#_3), & (q, 6), & L \\
            (q, 6), & (\forall_1, @, \#_3) & \to & (\forall_1, @, \#_3), & (q, 7), & L \\
            (q, 7), & (\forall_1, \forall_2, \forall_3) & \to & (\forall_1, \forall_2, \forall_3), & q, & R \\
        \end{matrix}
        \]
        It is easy to see that if $\ket{q_0, \mathcal{T}, 0} \xrightarrow[T]{M} \ket{q_f, \mathcal{T}', 0}$, then
        \[
            \ket{q_0, \mathcal{T};@;{\#\$}, 0} \xrightarrow[T_2]{M_2} \ket{q_f, \mathcal{T}';@;\mathcal{T}^h, 0},
        \]
        where
        \[
            T_2 = \sum_{t=0}^{T-1} \left( 2\left[(t+3)-\mathit{pos}(\mathcal{T}, t)\right] + \begin{cases} 5 & d(\mathcal{T}, t)=L \\ 1 & d(\mathcal{T}, t)=R \end{cases} \right) = O(T^2),
        \]
        $\mathit{pos}(\mathcal{T}, t)$ and $d(\mathcal{T}, t)$ denote the head position and the chosen direction at time $t$ in the execution of $\mathcal{M}$ starting from tape content $\mathcal{T}$, respectively. We note that $d(\mathcal{T}, t) = \mathit{pos}(\mathcal{T}, t+1)-\mathit{pos}(\mathcal{T}, t)$.

        \textbf{Conclusion}. The RTM $M'$ is obtained by dovetailing the two RTMs $M_1$ and $M_2$ by Lemma \ref{prop-dovetail}, which immediately yields
        \[
            \mathcal{T} \xrightarrow[T]{M} \mathcal{T}' \Longrightarrow \mathcal{T};\epsilon;\epsilon \xrightarrow[T']{M'} \mathcal{T}';@;\mathcal{T}^{h},
        \]
        where $T' = T_1+T_2 = O(T^2)$. Moreover, if $M$ is oblivious, for every input $x \in \{0,1\}^*$, i.e. the initial tape content is $\mathcal{T}_x$, the running time and the head position of $M$ can be denoted by $T = T(\abs{x})$ and $\mathit{pos}(\mathcal{T}, t) = \mathit{pos}(\abs{x}, t)$, respectively. It can be seen that the constructed $M'$ is also oblivious by noticing that
        \begin{enumerate}
          \item During the simulation for time $t$ of $M$, the head position starts at $\mathit{pos}(\abs{x}, t)$ and goes right to $t+3$ and back. The head position of $M'$ during the whole execution of this part of simulation only depends on $\mathit{pos}(\abs{x}, t)$.
          \item $T_2$ depends only on $\mathit{pos}(\abs{x}, t)$ because $\mathit{pos}(\mathcal{T}, t) = \mathit{pos}(\abs{x}, t)$ and $d(\mathcal{T}, t) = d(\abs{x}, t) = \mathit{pos}(\abs{x}, t+1)-\mathit{pos}(\abs{x}, t)$. Therefore, the running time $T' = T_1+T_2$ of $M'$ depends only on $\abs{x}$.
        \end{enumerate}
    \end{proof}

    \begin{lemma} \label{prop-RTM}
        Let $M$ be a stationary deterministic TM such that for every input $x \in \{0, 1\}^*$,
        \[
            x \xrightarrow[T]{M} {M(x)}.
        \]
        There is a stationary reversible TM $M'$ such that
        \[
            x;\epsilon \xrightarrow[T']{M'} x;{M(x)}\ {\rm
        and}\
            x;{M(x)} \xrightarrow[T']{M'} x;{\epsilon},
        \]
        where $T' = O(T^2)$. Moreover, if $M$ is oblivious and the length of $M(x)$ only depends on $\abs{x}$, then $M'$ is oblivious.
    \end{lemma}
    \begin{proof}
    Let $M_h$ be the constructed RTM corresponding to $M$ in Lemma \ref{prop-history} and $M_h^{-1}$ be its reversal by Lemma \ref{prop-reversal}, and $M_{c}$ be the constructed RTM in Lemma \ref{prop-copy}. Then $M'$ is constructed by dovetailing $M_h[1,2,3]$, $M_c[1,4]$ and $M_h^{-1}[1,2,3]$ by Lemma \ref{prop-dovetail}.
        We could verify that:
        \begin{align*}
            x;\epsilon;\epsilon;\epsilon
            & \xrightarrow[T_h]{M_h[1,2,3]} {M(x)};@;\mathcal{T}^{h};\epsilon \\
            & \xrightarrow[T_c]{M_c[1,4]} {M(x)};@;\mathcal{T}^{h};{M(x)} \\
            & \xrightarrow[T_h+2]{M_h^{-1}[1,2,3]} x;\epsilon;\epsilon;{M(x)} \\
        \end{align*}
        and
        \begin{align*}
            x;\epsilon;\epsilon;{M(x)}
            & \xrightarrow[T_h]{M_h[1,2,3]} {M(x)};@;\mathcal{T}^{h};{M(x)} \\
            & \xrightarrow[T_c]{M_c[1,4]} {M(x)};@;\mathcal{T}^{h};\epsilon \\
            & \xrightarrow[T_h+2]{M_h^{-1}[1,2,3]} x;\epsilon;\epsilon;\epsilon \\
        \end{align*}
        with running time $T' = T_h+T_c+T_h+2 = O(T^2)$.
    \end{proof}

    \begin{lemma} \label{prop-unique}
        Let $M_1, M_2$ be two stationary deterministic TMs such that for every $x \in \{0, 1\}^*$,
        \[
            x \xrightarrow[T_1]{M_1} {M_1(x)} \
            {\rm and}\
            {M_1(x)} \xrightarrow[T_2]{M_2} x.
        \]
        There  are two stationary reversible TMs $N_1$ and $N_2$ such that
        \[
            x \xrightarrow[T']{N_1} {M_1(x)}\
        {\rm and}\
            {M_1(x)} \xrightarrow[T']{N_2} x,
        \]
        where $T' = O(T_1^2+T_2^2)$. Moreover, if $M_1$ and $M_2$ are oblivious and the length of $M_1(x)$ only depends on $\abs{x}$, then $N_1$ and $N_2$ are oblivious.
    \end{lemma}
    \begin{proof}
    Let $M_1'$ and $M_2'$ be the RTMs constructed by Lemma \ref{prop-RTM} and $M_\mathit{swap}$ be the RTM in Lemma \ref{prop-swap}. $N_1$ is constructed by dovetailing $M_1'[1,2]$, $M_\mathit{swap}[1,2]$ and $M_2'[1,2]$. Verify that
        \begin{align*}
            x;\epsilon \xrightarrow[T_1']{M_1'[1,2]} x;{M_1(x)} \xrightarrow[T_\mathit{swap}]{M_\mathit{swap}[1,2]} {M_1(x)};x \xrightarrow[T_2']{M_2'[1,2]} {M_1(x)};\epsilon,
        \end{align*}
        where $T_1' = O(T_1^2), T_2' = O(T_2^2)$ and $T_\mathit{swap} = O(\abs{x}+\abs{M(x)})$.
        $N_2$ is constructed by dovetailing $M_2'[1,2]$, $M_\mathit{swap}[1,2]$ and $M_1'[1,2]$. Verify that
        \begin{align*}
            {M_1(x)};\epsilon \xrightarrow[T_2']{M_2'[1,2]} {M_1(x)};x \xrightarrow[T_\mathit{swap}]{M_\mathit{swap}[1,2]} x;{M_1(x)} \xrightarrow[T_1']{M_1'[1,2]} x;\epsilon,
        \end{align*}
    \end{proof}

    Now we are ready to prove Lemma \ref{prop-inc} and Lemma \ref{prop-eq}.

    \subsection*{Proof of Lemma \ref{prop-inc}}
        The proof is immediately shown by giving two oblivious DTMs using Lemma \ref{prop-unique}.

        Below is an oblivious TM $M_+$:
        \[
        \begin{matrix}
            q_0, & \forall & \to & \forall, & q_1, & L \\
            q_1, & \# & \to & \#, & q_2, & R \\
            q_2, & x & \to & x, & q_2, & R \\
            q_2, & \# & \to & \#, & (q_3, 1), & L \\
            (q_3, 0), & 0 & \to & 0, & (q_3, 0), & L \\
            (q_3, 0), & 1 & \to & 1, & (q_3, 0), & L \\
            (q_3, 0), & \# & \to & \#, & q_f, & R \\
            (q_3, 1), & 0 & \to & 1, & (q_3, 0), & L \\
            (q_3, 1), & 1 & \to & 0, & (q_3, 1), & L \\
            (q_3, 1), & \# & \to & \#, & q_f, & R
        \end{matrix}
        \]
        It can be verified that $M_+$ increments $x$ by $1$ and has running time $2\abs{x}+4 = O(\abs{x})$.

        Below is an oblivious TM $M_-$:
        \[
        \begin{matrix}
            q_0, & \forall & \to & \forall, & q_1, & L \\
            q_1, & \# & \to & \#, & q_2, & R \\
            q_2, & x & \to & x, & q_2, & R \\
            q_2, & \# & \to & \#, & (q_3, 1), & L \\
            (q_3, 0), & 0 & \to & 0, & (q_3, 0), & L \\
            (q_3, 0), & 1 & \to & 1, & (q_3, 0), & L \\
            (q_3, 0), & \# & \to & \#, & q_f, & R \\
            (q_3, 1), & 0 & \to & 1, & (q_3, 1), & L \\
            (q_3, 1), & 1 & \to & 0, & (q_3, 0), & L \\
            (q_3, 1), & \# & \to & \#, & q_f, & R
        \end{matrix}
        \]
        It can be verified that $M_-$ decrements $x$ by $1$ and has running time $2\abs{x}+4 = O(\abs{x})$.

    \subsection*{Proof of Lemma \ref{prop-eq}}
        The proof is immediately shown by giving an oblivious DTM using Lemma \ref{prop-unique}. It is noted that the given DTM itself is the reversal of it.

        The an oblivious TM $M_=$ is as below:
        \[
        \begin{matrix}
            q_0, & (\forall_1, \forall_2, \forall_3) & \to & (\forall_1, \forall_2, \forall_3), & q_1, & L \\
            q_1, & (\forall_1, \forall_2, \forall_3) & \to & (\forall_1, \forall_2, \forall_3), & q_2, & R \\
            q_2, & (x_1, x_2, \forall_3) & \to & (x_1, x_2, \forall_3), & q_2, & R \\
            q_2, & (\#_1, \#_2, \#_3) & \to & (\#_1, \#_2, \#_3), & (q_3, 0), & L \\
            (q_3, 0), & (x, x, \forall_3) & \to & (x, x, \forall_3), & (q_3, 0), & L \\
            (q_3, 0), & (0, 1, \forall_3) & \to & (0, 1, \forall_3), & (q_3, 1), & L \\
            (q_3, 0), & (1, 0, \forall_3) & \to & (1, 0, \forall_3), & (q_3, 1), & L \\
            (q_3, 0), & (\#_1, \#_2, \#_3) & \to & (\#_1, \#_2, \#_3), & (q_4, 0), & R \\
            (q_3, 1), & (x_1, x_2, \forall_3) & \to & (x_1, x_2, \forall_3), & (q_3, 1), & L \\
            (q_3, 1), & (\#_1, \#_2, \#_3) & \to & (\#_1, \#_2, \#_3), & (q_4, 1), & R \\
            (q_4, 0), & (x_1, x_2, 0) & \to & (x_1, x_2, 0), & q_5, & L \\
            (q_4, 0), & (x_1, x_2, 1) & \to & (x_1, x_2, 1), & q_5, & L \\
            (q_4, 1), & (x_1, x_2, 0) & \to & (x_1, x_2, 1), & q_5, & L \\
            (q_4, 1), & (x_1, x_2, 1) & \to & (x_1, x_2, 0), & q_5, & L \\
            q_5, & (\forall_1, \forall_2, \forall_3) & \to & (\forall_1, \forall_2, \forall_3), & q_f, & R
        \end{matrix}
        \]
        It can be verified that $M_=$ checks whether $x$ and $y$ are equal and has running time $2\abs{x}+6 = O(\abs{x})$. Moreover, we have that $M_=(M_=(x;y;z)) = x;y;z$.

    \subsection*{}
        Finally, we give a tape shifting lemma in order to prove Lemma \ref{prop-shift}.

    \begin{lemma} [Tape Shifting] \label{prop-copy2}
        There is a stationary reversible TM $M$ that copies the first track to the second track if the content of the first track is shifted left or right by one step. Formally, for every $x \in \{0,1\}^+$,
        \[
            \operatorname{shl} x; \epsilon \xrightarrow[T]{M} \operatorname{shl} x; \operatorname{shl} x\ {\rm
        and}\
            \operatorname{shr} x; \epsilon \xrightarrow[T]{M} \operatorname{shl} x; \operatorname{shr} x,
        \]
        where $T = 2\abs{x}+8$.
    \end{lemma}

    \begin{proof}
        Below is the construction. We use $\sigma$ to denote any symbol other than $\#$, $d$ to denote both directions $L$ and $R$ and $\bar d$ to denote the reverse direction of $d$.
        \[
        \begin{matrix}
            q_0, & (\sigma, \#) & \to & (\sigma, \#), & (q_L, 1), & L \\
            q_0, & (\#, \#) & \to & (\#, \#), & (q_R, 1), & R \\
            (q_d, 1), & (\sigma, \#) & \to & (\sigma, \#), & (q_d, 2), & L \\
            (q_d, 2), & (\#, \#) & \to & (\#, \#), & (q_d, 3), & R \\
            (q_d, 3), & (\sigma, \#) & \to & (\sigma, \sigma), & (q_d, 3), & R \\
            (q_d, 3), & (\#, \#) & \to & (\#, \#), & (q_d, 4), & L \\
            (q_d, 4), & (\sigma, \sigma) & \to & (\sigma, \sigma), & (q_d, 4), & L \\
            (q_d, 4), & (\#, \#) & \to & (\#, \#), & (q_d, 5), & R \\
            (q_d, 5), & (\sigma, \sigma) & \to & (\sigma, \sigma), & (q_d, 6), & \bar d \\
            (q_L, 6), & (\sigma, \sigma) & \to & (\sigma, \sigma), & q_7, & L \\
            (q_R, 6), & (\#, \#) & \to & (\#, \#), & q_7, & L \\
            q_7, & (\forall, \forall) & \to & (\forall, \forall), & q_f, & R \\
        \end{matrix}
        \]
    \end{proof}

        Now we are ready to prove Lemma \ref{prop-shift}.

    \subsection*{Proof of Lemma \ref{prop-shift}}\label{prop-shift-proof}
        Below is an oblivious TM $M_r$:
        \[
        \begin{matrix}
            q_0, & \forall & \to & \forall, & q_1, & L \\
            q_1, & \# & \to & \#, & q_2, & R \\
            q_2, & x & \to & x, & q_2, & R \\
            q_2, & \# & \to & \#, & q_3, & L \\
            q_3, & x & \to & x, & (q_4, x), & R \\
            q_3, & \# & \to & \#, & q_6, & R \\
            (q_4, x), & \forall & \to & x, & q_5, & L \\
            q_5, & \forall & \to & \forall, & q_3, & L \\
            q_6, & \forall & \to & \#, & q_7, & L \\
            q_7, & \# & \to & \#, & q_f, & R \\
        \end{matrix}
        \]
        $\forall$ denotes any symbol in $\Sigma$ while $x$ denotes any symbol in $\Sigma$ other than $\#$. It can be verified that $M_r$ shifts the tape right by a cell and has running time $4\abs{x}+6 = O(\abs{x})$.

        Below is an oblivious TM $M_l$:
        \[
        \begin{matrix}
            q_0, & \forall & \to & \forall, & q_1, & L \\
            q_1, & \# & \to & \#, & q_2, & R \\
            q_2, & x & \to & x, & (q_3, x), & L \\
            q_2, & \# & \to & \#, & q_5, & L \\
            (q_3, x) & \forall & \to & x, & q_4, & R \\
            q_4 & \forall & \to & \forall, & q_2, & R \\
            q_5, & \forall & \to & \#, & q_6, & L \\
            q_6, & x & \to & x, & q_6, & L \\
            q_6, & \# & \to & \#, & q_7, & R \\
            q_7, & \forall & \to & \forall, & q_f, & R \\
        \end{matrix}
        \]
        It can be verified that $M_r$ shifts the tape left by a cell and has running time $4\abs{x}+6 = O(\abs{x})$.

        Let $M_l'$ and $M_r'$ be the RTMs constructed by Lemma \ref{prop-history} corresponding to $M_l$ and $M_r$, respectively. Let $M_c$ be the RTM in Lemma \ref{prop-copy2}. It is noted that for every $x \in \{0,1\}^+$,
        \[
            x;\epsilon;\epsilon;\epsilon \xrightarrow[T_l]{M_l'[1,2,3]} \operatorname{shl}x;@;\mathcal{T}^{hl};\epsilon \xrightarrow[T_c]{M_c[1,4]} \operatorname{shl}x;@;\mathcal{T}^{hl};\operatorname{shl}x
        \]
        and
        \[
            x;\epsilon;\epsilon;\epsilon \xrightarrow[T_r]{M_r'[1,2,3]} \operatorname{shr}x;@;\mathcal{T}^{hr};\epsilon \xrightarrow[T_c]{M_c[1,4]} \operatorname{shr}x;@;\mathcal{T}^{hr};\operatorname{shr}x.
        \]
        Moreover, it can be verified that $T_l = T_r = O(\abs{x}^2)$.

        We note that
        \[
        \operatorname{shl}x;@;\mathcal{T}^{hl};\operatorname{shl}x \xrightarrow[T_l+2]{M_l'^{-1}[1,2,3]} x;\epsilon;\epsilon;\operatorname{shl}x,
        \]
        \[
        \operatorname{shl}x;@;\mathcal{T}^{hl};\operatorname{shl}x \xrightarrow[T_l+2]{M_l'^{-1}[4,2,3]} \operatorname{shl}x;\epsilon;\epsilon;x,
        \]
        \[
        \operatorname{shr}x;@;\mathcal{T}^{hr};\operatorname{shr}x \xrightarrow[T_r+2]{M_r'^{-1}[1,2,3]} x;\epsilon;\epsilon;\operatorname{shr}x,
        \]
        \[
        \operatorname{shr}x;@;\mathcal{T}^{hr};\operatorname{shr}x \xrightarrow[T_r+2]{M_l'^{-1}[4,2,3]} \operatorname{shr}x;\epsilon;\epsilon;x.
        \]
        According to these four cases, we can obtain four RTMs as follows:
        \[
        x;\epsilon \xrightarrow[T_l^1]{M_l^1} \operatorname{shl}x;x,
        \]
        \[
        x;\epsilon \xrightarrow[T_l^2]{M_l^2} x;\operatorname{shl}x,
        \]
        \[
        x;\epsilon \xrightarrow[T_r^1]{M_r^1} \operatorname{shr}x;x,
        \]
        \[
        x;\epsilon \xrightarrow[T_r^2]{M_r^2} x;\operatorname{shr}x
        \]
        with $T_l^1 = T_l^2 = T_r^1 = T_r^2 = O(\abs{x}^2)$.

        We use $M_L$ and $M_R$ to denote the RTM that moves the tape head left and right, respectively, without modifying anything. Formally,
        \[
            \ket{\mathcal{T}, \xi} \xrightarrow[3]{M_L} \ket{\mathcal{T}, \xi-1}
        \]
        and
        \[
            \ket{\mathcal{T}, \xi} \xrightarrow[3]{M_R} \ket{\mathcal{T}, \xi+1}.
        \]
        The running time is $3$ because $M_L$ and $M_R$ should be in normal form, and we achieve this by making $M_L$ go left, left and right and making $M_R$ go right, left, right, both of which need three steps.
        The construction of $M_L$ and $M_R$ is trivial. Now we are able to build two RTMs that just shift left or right the whole tape. Note that
        \[
            \ket{x;\epsilon, 0} \xrightarrow[T_l^1]{M_l^1} \ket{\operatorname{shl}x;x, 0} \xrightarrow[3]{M_L} \ket{\operatorname{shl}x;x, -1} \xrightarrow[T_r^2]{M_r^2} \ket{\operatorname{shl}x;\epsilon, -1} \xrightarrow[3]{M_R} \ket{\operatorname{shl}x;\epsilon, 0},
        \]
        and
        \[
            \ket{x;\epsilon, 0} \xrightarrow[T_r^1]{M_r^1} \ket{\operatorname{shr}x;x, 0} \xrightarrow[3]{M_R} \ket{\operatorname{shr}x;x, 1} \xrightarrow[T_l^2]{M_l^2} \ket{\operatorname{shr}x;\epsilon, 1} \xrightarrow[3]{M_L} \ket{\operatorname{shr}x;\epsilon, 0}.
        \]
    
\end{document}